\documentclass[11pt]{article}
\usepackage[margin=1in]{geometry}
\usepackage[utf8]{inputenc}
\usepackage[T1]{fontenc}
\usepackage{authblk}
\usepackage{todonotes}
\usepackage{xspace}
\usepackage{cite}
\usepackage[numbers,sort]{natbib}

\usepackage{comment}

\usepackage{nicefrac}
\usepackage{algorithm}
\usepackage{algorithmicx}
\usepackage[noend]{algpseudocode}

\usepackage{amssymb}
\usepackage{amsmath}
\usepackage{amsthm}
\usepackage{dsfont}
\usepackage[backref=page]{hyperref}

\usepackage[all]{hypcap}
\usepackage{footnote}
\usepackage{cleveref}
\usepackage{thm-restate}
\usepackage[tableposition=bottom]{caption}
\usepackage[all]{hypcap}
\usepackage{subcaption}
\usepackage{framed}
\usepackage[framemethod=tikz]{mdframed}
\usepackage[skins]{tcolorbox}

\makeatletter
\g@addto@macro{\maketitle}{\@thanks}
\makeatother

\newtheorem{thm}{Theorem}[section]
\newtheorem{cor}[thm]{Corollary}

\newtheorem{lem}[thm]{Lemma}
\newtheorem{Def}[thm]{Definition}
\newtheorem{obs}[thm]{Observation}
\newtheorem{claim}[thm]{Claim}

\newtheorem{remark}[thm]{Remark}

\newcommand{\E}{\mathbb{E}}%
\newcommand{\Ber}{\mathrm{Ber}}

\newcommand{\eps}{\epsilon}%

\newcommand{\algotype}{sound\xspace}

\renewcommand{\algorithmiccomment}[1]{\bgroup\hfill$\rhd$~#1\egroup}

\newcounter{note}[section]

\newcommand{\calA}{\mathcal{A}}
\newcommand{\calM}{\mathcal{M}}
\newcommand{\ratio}{0.531}
\newcommand{\realratio}{0.536}

\newenvironment{wrapper}[1]
{
	\begin{center}
		\begin{minipage}{\linewidth}
			\begin{mdframed}[hidealllines=true, backgroundcolor=gray!20, leftmargin=0cm,innerleftmargin=0.4cm,innerrightmargin=0.4cm,innertopmargin=0.4cm,innerbottommargin=0.4cm,roundcorner=10pt]
				#1}
			{\end{mdframed}
		\end{minipage}
	\end{center}
} 

\renewcommand{\paragraph}[1]{\medskip\noindent\textbf{#1}}

\author[1]{Niv Buchbinder\thanks{Supported in part by by Israel Science Foundation grant 2233/19 and United States - Israel Binational Science Foundation grant 2018352.}}
\author[2]{Joseph (Seffi) Naor\thanks{Supported in part by by Israel Science Foundation grant 2233/19 and United States - Israel Binational Science Foundation grant 2018352.}}
\author[3]{David Wajc\thanks{Work done in part while the author was at Carnegie Mellon University and Stanford University.}}
\affil[1]{Tel Aviv University, niv.buchbinder@gmail.com}
\affil[2]{Technion, naor@cs.technion.ac.il}
\affil[3]{Google Research, david.wajc@gmail.com}
\date{}

\begin{document}
	
	\pagenumbering{gobble}
	
	\title{Lossless Online Rounding for Online Bipartite Matching \\(Despite its Impossibility)}

	\maketitle
	
	\begin{abstract}
	For numerous online bipartite matching problems, such as edge-weighted matching and matching under two-sided vertex arrivals, the state-of-the-art fractional algorithms outperform their randomized integral counterparts. This gap is surprising, given that the bipartite fractional matching polytope is integral, and so lossless rounding is possible. This gap was explained by Devanur et al.~(SODA'13), who showed that \emph{online} lossless rounding is impossible.

	\smallskip 
	
	Despite the above, we initiate the study of lossless online rounding for online bipartite matching problems. Our key observation is that while lossless online rounding is impossible \emph{in general}, randomized algorithms induce fractional algorithms of the same competitive ratio which by definition are losslessly roundable online. This motivates the addition of constraints that decrease the ``online integrality gap'', thus allowing for lossless online rounding. We characterize a set of non-convex constraints which allow for such lossless online rounding, and better competitive ratios than yielded by deterministic algorithms.

	\smallskip
	
	As applications of our lossless online rounding approach, we obtain two results of independent interest: (i) a doubly-exponential improvement, and a sharp threshold for the amount of randomness (or advice) needed to outperform deterministic online (vertex-weighted) bipartite matching algorithms, and (ii) an optimal semi-OCS, matching a recent result of Gao et al.~(FOCS'21) answering a question of Fahrbach et al.~(FOCS'20). 
	\end{abstract}	
	\newpage 
	\tableofcontents
	\newpage
	\pagenumbering{arabic}
	\section{Introduction}

Real-time decision-making is ubiquitous in real-world domains, from ride hailing, to online dating, to Internet advertising.
The core products of these sectors are solutions to variants of the classic online bipartite matching problem of \citet{karp1990optimal}.
Here, nodes of one side of a bipartite graph are given, and nodes on the opposite side arrive one by one, and must be matched (or not) immediately and irrevocably upon arrival.
For example, in Internet advertising, offline and online nodes correspond to advertisers and opportunities to display an ad, respectively \cite{mehta2007adwords}.
In general, the pervasiveness of web- and mobile-based user-facing apps provides an ever-increasing supply of online problems, and a growing demand for general methods for tackling such problems.

\smallskip

Fittingly, a concentrated and highly-successful effort has been dedicated to devising general techniques for the design and analysis of online algorithms.
Mirroring its central role in offline optimization \cite{lovasz2009matching}, the matching problem has been influential in this process, too, inspiring new general online optimization techniques, including the randomized primal-dual method \cite{devanur2013randomized}, and online correlated selection (OCS) \cite{fahrbach2020edge} (more on this below).

\smallskip

One \emph{old} design pattern for online optimization is online rounding:
designing an online algorithm for a fractional relaxation of the problem, and then randomly rounding the fractional solution online.
This approach has played a pivotal role in the resolution of many fundamental online minimization problems (see \Cref{sec:related}), but has had limited success for online bipartite matching problems.
This limited success is surprising, given that competitive fractional online algorithms are known for many such problems  \cite{buchbinder2007online,feldman2009online2,kalyanasundaram2000optimal,wang2015two}, and moreover fractional bipartite matchings can be rounded to integrality with no loss in the objective. Why, then, do we not have equally competitive \emph{integral} online algorithms for such problems as edge-weighted matching \cite{fahrbach2020edge,gao2021improved,blanc2021multiway}, matching under two-sided arrivals \cite{gamlath2019online} and AdWords without the small bids assumption \cite{huang2020adwords}?

\smallskip

A partial answer is that the myriad lossless rounding algorithms for fractional bipartite matching
(e.g., \cite{ageev2004pipage,gandhi2006dependent,goel2013perfect},\cite[6.5.11]{grotschel2012geometric}) 
are all \emph{offline} in nature, and seem difficult to implement in online settings.
A more complete answer was given by \citet{devanur2013randomized}, who noted that lossless online rounding is not only more difficult than its offline counterpart---it is \emph{impossible} (see \Cref{sec:impossible}).
Therefore, all prior rounding-based online matching algorithms \cite{cohen2018randomized,gamlath2019online,saberi2021greedy,papadimitriou2021online,gao2021improved} use \emph{lossy} rounding, attaining a competitive ratio strictly worse than the fractional algorithms on which they are based.

\smallskip

As we show, all hope is not lost, and the truth is more nuanced.

\subsection{An Overlooked Research Question}
We initiate the study of lossless online rounding for maximization problems, focusing on the classic online bipartite matching problem and its vertex-weighted generalization \cite{aggarwal2011online}. (We do so despite the preceding discussion suggesting the futility of this endeavor.)

\smallskip

Our starting point is the well-known observation that any randomized online algorithm induces a fractional algorithm with the same competitive ratio, by setting the fractional values of decision variables to be their expected value under the randomized algorithm.
This observation is frequently used in lower bounds (i.e., impossibility results), since it implies that lower bounds on fractional algorithms' competitive ratios naturally transfer to randomized algorithms.
Breaking with this tradition, we use this observation to obtain \emph{upper bounds} (i.e., algorithms).

\smallskip

Driving our work is the following corollary of the above observation:
any randomized online algorithm $\cal{A}$ for a problem $\Pi$  induces a fractional online algorithm with the same competitive ratio and which is also losslessly roundable online---just run algorithm $\cal{A}$!
This holds even if $\Pi$ does not allow for lossless online rounding in general.

\smallskip

We conclude that the aforementioned impossibility result of \cite{devanur2013randomized} only shows that lossless online rounding is impossible when relying only on the natural fractional constraints imposed for the offline problem. This holds even though these constraints induce a polytope with no integrality gap, which is losslessly roundable offline.\footnote{Recall that the \emph{integrality gap} of a polytope $\mathcal{P}\subseteq \mathbb{R}^m$, defined as $\max_{w\in \mathbb{R}^m} \frac{\max\{w\cdot x \mid x\in \mathcal{P}\}}{\max\{w\cdot x \mid x\in \mathcal{P}\cap \mathbb{Z}^m\}}$, is the highest multiplicative gap between the objective of fractional and integral points in the polytope $\mathcal{P}$ over all linear objectives $\vec{w}$. A polytope with integrality gap of one is often said to have no (non-trivial) integrality gap}.
Similarly to the use of additional constraints to decrease the integrality gap in offline settings (e.g., the influential work of \citet{edmonds1965maximum}),
% (e.g., Edmonds' Blossom constraints for the matching polytope in non-bipartite graphs \cite{edmonds1965maximum}),
we propose adding additional constraints to the problem that reduce the ``online integrality gap".
%Of course, this reasoning seems a tad circular.
With this perspective in mind, the design of randomized algorithms can now be reduced to the following two-step program: (1) design a fractional online algorithm which \emph{is} losslessly roundable online, and then (2) round it.
The search for the best possible competitive ratio thus motivates the following question.

\begin{quote}
	\centering \emph{Which fractional online bipartite matching algorithms are losslessly roundable online?}
\end{quote}

\subsection{Our Contributions}\label{sec:contributions}

We present a set of (non-convex) constraints which allows for lossless online rounding for bipartite matching problems (and high competitive ratios). From this, we obtain new online matching algorithms, a resolution to an open problem concerning randomness and advice complexity of online matching, and a systematic method of designing algorithms for online correlated selection.

\paragraph{A family of roundable algorithms.}
Describing the fractional algorithms we consider requires some notation, which we now provide.
The input to the online matching problem is a bipartite graph. Each offline (i.e., left-hand-side) node $i$ has weight $w_i>0$ (in the unweighted version $w_i=1$ for all $i$).
Initially, only the $n$ offline nodes are known.
At time $t=1,2,\dots$, online node $t$ arrives.
A fractional online matching algorithm must assign each edge $(i,t)\in E$ upon arrival, immediately and irrevocably, a value $x_{i,t}\in [0,1]$, such that each node $v$ has \emph{fractional degree} at most one, i.e., $\sum_{e\ni v} x_e \leq 1$.
So, for example, $x^{(t)}_i := \sum_{t'<t} x_{i,t'}$, the fractional degree of offline node $i$ before time $t$, satisfies $x^{(t)}_i\leq 1$.
The goal is to maximize the weighted value of the matching, $\sum_{i,t} w_i\cdot x_{i,t}$.

\smallskip

Following a number of recent breakthrough works in the online matching literature \cite{fahrbach2020edge,huang2020adwords}, we focus on \emph{two-choice} algorithms: these are randomized algorithms which randomize the matching choice of an online algorithm between two or fewer offline neighbors.
Similarly, \emph{fractional} two-choice algorithms are algorithms that for all time steps $t$, set $x_{i,t}>0$ for at most two offline nodes $i$.
Randomized two-choice algorithms were recently used to break the barrier of $\nicefrac{1}{2}$ for the competitive ratio of other online bipartite matching problems \cite{fahrbach2020edge,huang2020adwords}.
On the other hand, prior work shows that two-choice fractional algorithms are not generally losslessly roundable online \cite{devanur2013randomized} (see also \Cref{sec:impossible}).

\smallskip

A key ingredient in our work is the introduction of the following set of (non-convex) constraints for fractional two-choice algorithms, which we motivate and provide intuition for in \Cref{sec:techniques}.

\begin{wrapper}
	\begin{Def}\label{def:maximal}
		A two-choice fractional online matching algorithm $\mathcal{A}$ is \emph{\algotype} if for every online node $t$ with $P_t := \{i \mid x_{i,t}>0\}$, the fractional matching $\vec{x}$ of $\mathcal{A}$ satisfies
		\begin{equation}\label{rounding-weaker-condition}
		\sum_{i\in P_t} x_{i,t} \leq 1-\prod_{i\in P_t} x^{(t)}_i.
		\end{equation}	
		If \Cref{rounding-weaker-condition} is met at equality for each $t$, we say that $\calA$ is also \emph{maximal}.
	\end{Def}
\end{wrapper}

\noindent\textbf{Lossless online rounding.}
Our main technical contribution is a proof that the above condition is sufficient for rounding, i.e., \algotype two-choice  algorithms can be rounded losslessly online, even on an \emph{edge-by-edge basis}.
\begin{wrapper}
	\begin{thm}[Lossless Rounding]\label{thm:rounding-intro}
		Let $\mathcal{A}$ be a \algotype two-choice online fractional algorithm, and its output be $\vec{x}$. Then, there exists a randomized online algorithm whose output matching $\mathcal{M}$ matches each edge $(i,t)$ with probability $$\Pr[(i,t)\in \mathcal{M}]= x_{i,t}.$$
		If $\calA$ is also maximal, then this randomized online algorithm is implementable in poly-time.
	\end{thm}
\end{wrapper}

\noindent\textbf{Competitive roundable fractional algorithms.}
Complementing \Cref{thm:rounding-intro}, we present competitive \algotype two-choice fractional matching algorithms.
Specifically, we present $0.536$-competitive unweighted and $0.524$-competitive vertex-weighted algorithms, breaking the natural barrier of $\nicefrac{1}{2}$.

\subsubsection{Applications}\label{sec:applications}

Illustrating the potential of our techniques, we obtain the following results of independent interest.

\paragraph{Sharp Randomness Threshold.}
In their seminal work, \citet{karp1990optimal} proved that deterministic algorithm are at best $\nicefrac{1}{2}$ competitive, while randomization allows to break this bound.
This begs the question: \emph{how random is random}? In other words, how much randomness is needed to outperform deterministic algorithms?
The \textsc{ranking} algorithm \cite{karp1990optimal} requires $\log(n!) = O(n \log n)$ random bits, slightly improved to $O(n)$ in \cite{durr2016power}.
In contrast, \citet{pena2019extensions} show that $(1-o(1))\log\log n$ bits of randomness (or even advice, see \Cref{sec:related}) are needed. (All logarithms are in base $2$.)

\smallskip

We close this doubly-exponential gap for the online (vertex-weighted) bipartite matching problem, showing that the lower bound is tight.
To do so, we show how to implement our randomized rounding scheme using small-bias probability spaces \cite{naor1993small} while only losing a $1+o(1)$ factor in the competitive ratio. Applying this refinement to our fractional algorithms, we obtain the following \emph{sharp threshold} on the amount of randomness (and advice) needed to outperform deterministic algorithms for this problem.

\begin{wrapper}
	\begin{thm}\label{thm:threshold}
		$(1 \pm o(1))\log \log n$ random (or advice) bits are both necessary and sufficient to achieve a competitive ratio of $(\nicefrac{1}{2}+\Omega(1))$ for online vertex-weighted bipartite matching.
	\end{thm}
\end{wrapper}

%\color{red}
Due to our (near-)lossless rounding using bounded randomness, the above $\nicefrac{1}{2}+\Omega(1)$ ratios are precisely the $0.536$ and $0.524$ for the unweighted and vertex-weighted problem which our fractional algorithms achieve. These bounds are comparable to, and in most cases larger than, similar recent results breaking the natural $\nicefrac{1}{2}$ barrier in the online matching literature \cite{gamlath2019online,fahrbach2020edge,huang2020adwords,huang2020fully}.

\paragraph{Optimal semi-OCS via online rounding.}
In their groundbreaking work, \citet{fahrbach2020edge} introduced a powerful new algorithmic primitive: online correlated selection (OCS).
Here, pairs of elements arrive, and an algorithm must choose one item in each pair immediately and irrevocably upon arrival.
Picking randomly (and independently) guarantees that each element belonging to $k$ pairs is unmatched with probability precisely $2^{-k}$.
A \emph{$\gamma$-semi-OCS} is an algorithm that, using correlated choices, decreases this probability to $2^{-k}(1-\gamma)^{k-1}$.
A strengthening of semi-OCS, referred to as OCS, was pivotal to their breakthrough result, breaking the competitive ratio of $\nicefrac{1}{2}$ for online edge-weighted matching, as well as for a similar breakthrough for the related AdWords problem \cite{huang2020adwords}.
This hints at wider applicability of this algorithmic tool, motivating its extension and refinement.
Fahrbach et al.~explicitly asked for the highest $\gamma$ for which a $\gamma$-semi-OCS exists.

\smallskip

In independent and concurrent work, \citet{gao2021improved} present an improved, $\frac{1}{2}$-semi-OCS, which they prove is optimal.
Our framework applied to a natural \algotype fractional algorithm achieves the same bound, while hinting at a more generic method of devising such novel algorithmic primitives as OCS: use online rounding!
In particular, in \Cref{sec:OCS}, we prove the following.

\begin{wrapper}
	\begin{thm}\label{lem:OCS}
		An optimal semi-OCS can be obtained using lossless online rounding.
	\end{thm}
\end{wrapper}

To illustrate our online lossless rounding-based approach's flexibility, in \Cref{sec:OCS} we also show how it yields a bichromatic variant of semi-OCS, of possible independent interest.

\subsection{Overview of our Techniques}\label{sec:techniques}
\textbf{Motivation for and Intuition behind Condition \eqref{rounding-weaker-condition}.} Suppose we have managed to losslessly round a fractional matching $\vec{x}$ until time $t$, in the sense that $\Pr[(i,t)\in \calM] = x_{i,t}$.
Consequently, each offline node is matched before time $t$ with probability $x_i^{(t)}$.
Suppose for now that these probabilities are independent.
Then, since $t$ can only be matched to a neighbor in $P_t := \{i\mid x_{i,t}>0\}$ if at least one node in $P_t$ is free at this point, we find that for lossless rounding to be possible, we have
\[
\sum_{i \in P_t} x_{i,t} = \Pr[t \textrm{ matched}] \leq 1-\Pr[\textrm{all } P_t \textrm{ matched before time }t] = 1-\prod_{i\in P_t} x_{i}^{(t)}.
\]
Therefore, we have that $\sum_{i \in P_t} x_{i,t} \leq 1-\prod_{i\in P_t}x_i^{(t)}$ is a \emph{necessary condition} for lossless rounding in this scenario, where all offline nodes are matched independently.
We say that our fractional matchings are
\emph{\algotype} precisely since this inequality is a necessary condition for lossless rounding.
For two-choice algorithms (i.e., with $|P_t|\leq 2$), it is not hard to match edges of $t$ with marginal probabilities $x_{i,t}$, assuming $\vec{x}$ is \algotype and the matched statuses of offline nodes are independent. Rounding becomes a challenge when correlations between offline nodes are introduced.

\medskip
\noindent\textbf{Dealing with correlations.}
In general, the matched status of offline nodes may depend on each other in intricate ways, based on the rounding so far.
Positive correlations between nodes in $P_t$ may cause the probability of all $P_t$ being matched before time $t$ to be strictly greater than $\prod_{i\in P_t} x_{i}^{(t)}$. This would rule out lossless rounding of maximal \algotype algorithms, as it would imply
the following.
\[
\Pr[t \textrm{ matched}] = \sum_{i\in P_t} x_{i,t} = 1-\prod_{i\in P_t} x_{i}^{(t)} > 1-\Pr[\textrm{all $P_t$ matched before time }t].
\]
However, \emph{negative dependence} between the matched status of offline nodes in $P_t$ would still be consistent with
matching edges of $t$ with marginal probabilities $x_{i,t}$ output by a \algotype algorithm. Indeed, such dependence (and lossless rounding up to time $t$)
would guarantee the necessary condition for lossless online rounding, namely $\Pr[t \textrm{ matched}]\leq 1-\Pr[\textrm{ all $P_t$ matched before time $t$}]$:
\[
\Pr[t \textrm{ matched}] = \sum_{i\in P_t} x_{i,t} \leq 1-\prod_{i\in P_t} x_{i}^{(t)} \leq  1-\Pr[\textrm{all $P_t$ matched before time }t].
\]
Accordingly, the crux of our rounding algorithm for \algotype two-choice algorithms is in designing a way of matching edges of $t$ that simultaneously: (1) guarantees the marginal probabilities, and (2) preserves pairwise negative correlations between the matched status of different nodes.
The second requirement is challenging to prove directly, and we end up proving significantly stronger forms of negative correlation between offline nodes to obtain our guarantees.

\medskip
\noindent\textbf{Our negative correlation properties.}
Let $F_{i,t}$ be an indicator for offline node $i$ being free (unmatched by the rounding algorithm) by time $t$, and let $F_{I,t}:=\bigwedge_{i\in I} F_{i,t}$ be an indicator for all of offline node-set $I$ being free at time $t$. In addition to satisfying the target marginal probabilities, our algorithm also satisfies the following strong negative correlation property
\begin{align*}
\Pr[F_{i,t} \mid F_{K,t}]\leq \Pr[F_{i,t} \mid F_{J,t}] \qquad \forall t,\, \forall i \in [n],\, \forall J\subseteq K\subseteq [n]\setminus \{i\} \textrm{ s.t.} \Pr[F_{K,t}]> 0.
\end{align*}
In words, the probability of node $i$ to be free at time $t$ decreases when we condition on a larger set of other offline nodes being free at that time. This strong negative correlation property, while natural, is not shared by prior online bipartite matching algorithms, such as \textsc{ranking} of \citet{karp1990optimal}, which is easily seen to induce \emph{positively} correlated variables $\{F_{i,t}\}_i$.

If the fractional matching $x$ is also maximal, our rounding algorithm satisfies the stronger invariant that at any time $t$, offline nodes in every set $I$ are either matched independently, or at least one node in $I$ must be matched, i.e.,
$\Pr[F_{I,t}] \in \left\{0,\, \prod_{i\in I} \Pr[F_{i,t}]\right\}.$

Both above negative correlation properties are strong enough to imply exponential tail bounds, similarly to known offline lossless rounding schemes \cite{gandhi2006dependent}, which might find uses beyond this work.

\medskip
\noindent\textbf{Low-randomness implementation.}
For particularly well-structured fractional algorithms, which we refer to as $k$-level algorithms (see \Cref{sec:roundingklevel}), 
each node's matched status depends on a constant number of random variables of our rounding algorithm.
This property, combined with the theory of small-bias distributions (see \Cref{sec:prelims}), allows us to implement our rounding scheme when applied to these algorithms using little randomness, while only incurring a $(1-o(1))$ multiplicative loss in the competitive ratio. This underlies \Cref{thm:threshold}, and hints at a wider applicability of small-bias probability spaces for the study of randomness (and advice) complexity of online algorithms.

\medskip
\noindent\textbf{Fractional Algorithms.}
Using \Cref{thm:rounding-intro} requires fractional algorithms whose outputs satisfy both the fractional matching constraints, as well as Condition \eqref{rounding-weaker-condition}, which is non-convex in the decision variables $x_{i,t}$, complicating the design and analysis space.
Nonetheless, we present numerous fractional algorithms approximately optimizing over these constraints and the fractional matching constraints, analyzed via (intricate instantiations of) the online primal-dual method \cite{buchbinder2009design}.
%\dwnote{add something about wider applicability?}

\subsection{Extension to Multiple-Choice Algorithms}
Two-choice algorithms cannot yield competitive ratios above $\nicefrac{5}{9}$ (see \citet{huang2019understandingb}).
A natural question, then, is how to generalize our approach to multiple-choice fractional algorithms.
Unfortunately, we were unable to extend our approach to this level of generality. However, the partial progress we made may be informative for follow-up work, so we discuss it here briefly.

\smallskip

When attempting to generalize our approach and obtain optimal $(1-\nicefrac{1}{e})$-competitive algorithms, there are two major hurdles to overcome: the first is characterizing additional constraints which allow for lossless rounding, and the second is to show that these constraints are compatible with optimal competitive ratios. In \Cref{sec:multi-appendix}, we consider a natural generalization of Condition \eqref{rounding-weaker-condition} to multiple-choice algorithms, and show that while this condition, which a priori should harm the competitive ratio significantly, still allows for optimal $(1-\nicefrac{1}{e})$ competitive ratios.
The intricate analysis of this fractional algorithm may inform the analysis of follow-up roundable online fractional algorithms.
Unfortunately, as we show in the same appendix, this condition does not allow for lossless online rounding, and so new ideas are needed.
We leave the characterization of additional constraints which allow for both optimal competitive ratios and lossless online rounding as an open question.

\subsection{Related work}\label{sec:related}

\noindent\textbf{Online Matching.} The online matching literature is a rich one; numerous extensions and generalizations of the problem of \citet{karp1990optimal} have been studied.
Recent developments include algorithms breaking the barrier of $\nicefrac{1}{2}$ for matching in general graphs \cite{gamlath2019online,huang2020fully,huang2020fully2}, edge-weighted matching \cite{fahrbach2020edge,blanc2021multiway,gao2021improved} and the AdWords problem \cite{huang2020adwords}.
Instead of providing an exhaustive overview, we point to these works and citations therein, and to the excellent survey of \citet{mehta2013online}.

\paragraph{Online Rounding.}
Online rounding has played a pivotal role in the resolution of myriad long-standing open problems in the online algorithms literature,  such as the $k$-server problem \cite{bansal2011polylogarithmic,bubeck2018k,buchbinder2019k,lee2018fusible}, weighted paging \cite{bansal2012primal,adamaszek2012log}, generalized caching \cite{bansal2012randomized,adamaszek2012log},  Metrical task systems \cite{bansal2010metrical,BCLLM19,CL19}, online set cover and its generalizations \cite{alon2009online,naor2011online,alon2006general}, online edge coloring \cite{cohen2019tight,saberi2021greedy}, etc.
For the above (minimization) problems, the rounding step loses an additional constant  or even logarithmic factor in the competitive ratio. This loss is unavoidable in some cases, since even the offline relaxation has an integrality gap. In other cases, losing an additional constant is not crucial since the original competitive ratio is already logarithmic.
We remark that a small number of online minimization problems (e.g., ski-renal, TCP-acknowledgement, or even Metrical task systems on HSTs) fractional solutions can in fact be rounded online losslessly~\cite{KKR01,bansal2010metrical, wang2015two}.

%\paragraph{Online Bipartite Matching.}
For \emph{maximization} problems, and in particular online matching problems,
where the optimal competitive ratios lies in $[1/2,1-1/e]$, no optimal results were obtained via rounding, due to the seeming inevitable loss in online rounding. We show that this loss is not necessary, after all.

\paragraph{Online Correlated Selection.}
Following the recent breakthrough of \cite{fahrbach2020edge}, online correlated selection has emerged as a powerful algorithmic primitive, with a flurry of work refining and extending this tool (see \cite{gao2021improved,blanc2021multiway,shin2021making,delong2022online}).
Numerous approaches for designing such algorithms have been proposed, and the design space of such algorithms is still to be understood. In this work we suggest a systematic method for designing such algorithms, by relying on online rounding.

\paragraph{Randomness and Advice Complexity.}
A core goal of online algorithms research is determining the power of randomization in online settings, or at least determining the \emph{amount} of randomness needed to outperform deterministic algorithms.
A related question is studied in the {\em advice model}, introduced by \citet{emek2011online}. In this model a deterministic online algorithm is equipped with an advice string, and the algorithm's guarantees for any input are with respect to the best advice string for the input. Such non-deterministic advice bits are at least as powerful as random bits, and are often strictly more powerful \cite{renault2015online,renault2015online2,mikkelsen2016randomization,bockenhauer2014online,pena2019extensions,boyar2017online}.
For online bipartite matching, the best deterministic competitive ratio is the trivial $\nicefrac{1}{2}$ \cite{karp1990optimal}.
The best upper bounds on the number of advice bits sufficient to do better are $O(\log n)$, using exponential time, and $O(n)$, using polynomial time \cite{bockenhauer2014online,pena2019extensions,durr2016power}.
Similarly, $\tilde{O}(n)$ bits of randomness suffice to beat this bound of $\nicefrac{1}{2}$ \cite{karp1990optimal,durr2016power}.
On the other hand, \citet{pena2019extensions} showed that any algorithm that is $(\nicefrac{1}{2}+\epsilon)$-competitive for any constant $\epsilon>0$ must use at least $(1-o(1))\log \log n$ advice bits, and hence at least that many random bits.
We close these (doubly-)exponential gaps, proving that the lower bound is tight.
	\section{Preliminaries}\label{sec:prelims}

In the online bipartite matching problem, the underlying input is a bipartite graph $G=(L,R,E)$. Initially, only the $n$ offline nodes (nodes), $L=[n]$, are known, as well as the weight $w_i>0$ of each offline node $i\in L$.
At times $t=1,2,\dots$, online node $t\in R$ arrives, together with its edges.

A fractional online matching algorithm, at each time $t$,
must decide, immediately and irrevocably, what values $x_{i,t}$ to assign edges $(i,t)\in E$, while guaranteeing that the \emph{fractional degree} of each vertex $v\in L\cup R$ satisfies $\sum_{e\ni v} x_e \leq 1$.
We denote by $x^{(t)}_i := \sum_{t'<t} x_{i,t'}$ the fractional degree of vertex $i$ before time $t$, and by $x_t := \sum_i x_{i,t}$ the fractional degree of vertex $t$. Thus, a fractional online algorithm guarantees that $x_t\leq 1$ and $x^{(t)}_i \leq 1$ for any time $t$ and offline vertex $i$. Thus, such an algorithm maintains a feasible solution to the fractional (vertex-weighted) matching problem,
whose linear programming relaxation is given (together with its dual) in \Cref{fig:matching-lp}.

	\begin{figure}[h]
	\begin{small}
		\begin{center}
			\begin{tabular}{rl|rl}
				\multicolumn{2}{c|}{Primal} & \multicolumn{2}{c}{Dual} \\ \hline
				maximize & $\sum_{(i,t)\in E} w_i\cdot x_{it}$ &
				minimize & $\sum_{i\in L}{y_i} +  \sum_{t\in R}{y_t}$ \\
				subject to: & & subject to: & \\
				
				$\forall i\in L$: &  $\sum_{t}x_{i,t} \leq 1$ & $\forall (i,t)\in E$: & $y_i + y_t \geq w_i$ \\
				
				$\forall t \in R$: & $\sum_{i} x_{i,t} \leq 1$ & $\forall i\in L$: & $y_i\geq 0$ \\
				
				$\forall (i,t)\in E$: & $x_{i,t} \geq 0$ & $\forall t\in R$: & $y_t \geq0$
			\end{tabular}
		\end{center}
	\end{small}
	\vspace{-0.5cm}
	\caption{The fractional vertex-weighted bipartite matching LP and its dual}\label{fig:matching-lp}
\end{figure}

\vspace{-0.2cm}
The analysis of our fractional algorithms uses these LPs
and the well-established online primal-dual method \cite{buchbinder2009design}. The challenge in the analysis is dealing with the non-convex Constraint \eqref{rounding-weaker-condition}.

\paragraph{Bounded Independence.}
A useful notion we will make use of is $(\delta,k)$-dependence, introduced by \citet{naor1993small}, generalizing $k$-wise independence.

\begin{Def}[\cite{naor1993small}]
	Binary random variables $Y_1,Y_2,\dots,Y_m$ are \emph{$(\delta,k)$-dependent}
	if for any subset of $k$ or fewer indices, $I\subseteq [m]$, $|I|\leq k$,
	$$\sum_{\vec{v}\in \{0,1\}^{|I|}} \left|\Pr\left[\bigwedge_{i\in I} (Y_i = v_i)\right]-2^{-|I|}\right| \leq \delta.$$
\end{Def}	

A $(0,k)$-dependent distribution is \emph{$k$-wise independent}, satisfying that any subset of $k$ or fewer variables is independent.
More generally, a joint distribution $\vec{Y}$ is $(\delta,k)$-dependent if for any subset $I\subseteq [m]$ of $k$ or fewer variables, the total variation distance between the distribution on the variables indexed by $I$ and the uniform distribution on $|I|$ i.i.d $\Ber(\nicefrac{1}{2})$ variables is at most $\delta$. Consequently, such a distribution satisfies the following.
\begin{lem}\label{reduction-to-independent}
	Let $\mathcal{U}$ be the uniform distribution and let $\mathcal{D}$ be a $(\delta,k)$-dependent distribution over $m$ binary variables $Y_1,Y_2,\dots,Y_m$.
	Then, for any event $A$ which is determined by $k$ or fewer random variables in $Y_1,Y_2,\dots,Y_m$,
$$\Pr_{\vec{Y}\sim \mathcal{U}}[A] - \delta \leq \Pr_{\vec{Y}\sim \mathcal{D}}[A] \leq \Pr_{\vec{Y}\sim \mathcal{U}}[A] + \delta.$$
\end{lem}
\begin{proof}
	Let $I\subseteq [m]$ be a set of $k$ or fewer indices such that $\{Y_i \mid i\in I\}$ determine $A$, and let $S\subseteq 2^{|I|}$ be such that $A = \bigcup_{\vec{v}\in S} \left[\vec{Y}=\vec{v}\right]$.
	Then, by triangle inequality and definition of $(\delta,k)$-dependence,
	\begin{align*}
	\left|\Pr_{\vec{Y}\sim \mathcal{D}}[A] - \Pr_{\vec{Y}\sim \mathcal{U}}[A]\right| & \leq \sum_{\vec{v}\in S} \left|\Pr\left[\bigwedge_{i\in I} (Y_i = v_i)\right]-2^{-|I|}\right|\leq \sum_{\vec{v}\in \{0,1\}^{|I|}} \left|\Pr\left[\bigwedge_{i\in I} (Y_i = v_i)\right]-2^{-|I|}\right|\leq \delta.
	~~~~\qedhere\end{align*}
\end{proof}

A useful property of $(\delta,k)$-dependent distributions is that such distributions can be specified---and constructed in polynomial time---using a small random seed \cite{naor1993small,Ta-Shma17}.
For completeness, a proof of this lemma, following the construction of \citet{naor1993small}, is given in Appendix \ref{sec:small-bias}.

\begin{restatable}{lem}{smallbiaslem}\label{eps-k-constructions}
	For any $\delta>0$, a $(\delta,k)$-dependent joint distribution on $n$ binary variables can be constructed using
	$\log\log n + O(k + \log(\frac{1}{\delta}))$ random bits.
	Moreover, after polytime preprocessing, each random variable in this distribution can be sampled in $O(k\cdot \log n)$ time.
\end{restatable}

	\section{Lossless Online Rounding: A Special Case}\label{sec:rounding}
In this section we study a special case of our rounding algorithm for the case that the input \algotype two-choice algorithm is also maximal, i.e., satisfies Condition \eqref{rounding-weaker-condition} at equality at each time $t$.
At the end of the section, we refer to appendices where we show how to implement this algorithm in polynomial time, generalize it to round non-maximal inputs, and refine this rounding to require doubly-logarithmic random seed on some well-structured inputs.

\subsection{The Algorithm}\label{sec:main-round}
In this section we show that any \emph{maximal} \algotype two-choice fractional solution $\vec{x}$ can be rounded while preserving the marginal probabilities of all edges, with even stronger negative correlation properties than for non-maximal such algorithms (see \Cref{sec:two-choice-fkg}).
We denote by $F_{I,t}$ the probability that a set of offline vertices $I\subseteq [n]$ is free at time $t$, and use $F_{i,t}$ as shorthand for $F_{\{i\},t}$.
\begin{align}
\Pr\left[F_{i,t}\right] & = 1-x^{(t)}_i & \forall i=1, \ldots, n.	\label{invariant:marginals}\\
\Pr\left[F_{I,t}\right] & \in \left\{0,\prod_{i\in I} (1-x^{(t)}_{i})\right\} & \forall I\subseteq\{1, \ldots, n\}.	  \label{invariant:correlation}
\end{align}

We call a subset $I$ of offline nodes {\em negative} (at time $t$) if $\Pr\left[F_{I,t}\right] = 0$. Otherwise, we say it is  {\em independent}, noting that in this case $\Pr\left[F_{I,t}\right] = \prod_{i\in I} (1-x_i) = \prod_{i\in I}\Pr\left[ F_{i,t}\right]$.
We note that this second name is apt, since for any independent set $I$ at time $t$, the variables $\{F_{i,t} \mid i\in I\}$ are indeed independent, as observed in \Cref{independent-valid-name}. Before proving this fact,
we  make the following simpler observations.

\begin{obs}\label{obs:negativity}
	A set that is negative at time $t$ remains negative at all times $t'\geq t$.
\end{obs}
\begin{obs}\label{obs:negativity-subsets}
	If a set $I$ is negative at time $t$, then all supersets of $I$ are negative at time $t$.
\end{obs}
\begin{restatable}{obs}{independentvalid}\label{independent-valid-name}
	If set $I$ is independent at time $t$, then the variables $\{F_{i,t} \mid i\in I\}$ are independent.
\end{restatable}
\begin{proof}By \Cref{obs:negativity-subsets}, If $I$ is independent at time $t$, then so are all of its subsets.
Therefore, for any subset $J\subseteq I$, we have that $\Pr[F_{J,t}] = \prod_{j\in J} \Pr[F_{j,t}]$.
Consequently, for any disjoint subsets $K,J$ of $I$, by the inclusion-exclusion principle and Invariant \eqref{invariant:marginals}, we have that indeed
\begin{align*}
\Pr[F_{K,t}, \overline{F_{J,t}}] & = \sum_{r=0}^{|J|} (-1)^r \sum_{J'\subseteq J, |J'|=r} \Pr[F_{K\cup J',t}] = \sum_{r=0}^{|J|} (-1)^r \sum_{J'\subseteq J, |J'|=r} \prod_{j\in K\cup J'}(1-x_{j}^{(t)}) \\
& = \prod_{k\in K}(1-x_{k}^{(t)})\cdot \prod_{j\in J}x_{j}^{(t)} = \prod_{k\in K}\Pr[F_{k,t}]\cdot \prod_{j\in J}\left(1-\Pr[F_{j,t}]\right). \qedhere
\end{align*}
\end{proof}

Our rounding algorithm for maximal \algotype two-choice algorithms is a special case of \Cref{alg:pairwise-rounding} in \Cref{sec:two-choice-fkg}. Our reasons to present this special case are threefold: (i) it is simpler to describe, (ii) its analysis is more elegant, and (iii) it provides stronger negative correlation properties, which yield efficient polynomial-time implementation of this rounding scheme. 
%This last point proves particularly useful, given that the fractional matching algorithms we apply our rounding schemes to are all maximal.

As asserted above, our rounding scheme of this section is very simple to describe. It examines whether the (at most) two vertices whose fractional value increased are negative or independent, and which of the offline vertices is available to be matched. Then, it carefully decides probabilistically how to match the new online vertex. The formal definition of the algorithm is given in \Cref{alg:rounding}.
\begin{algorithm}[H]
	\caption{Lossless Rounding: The Maximal Case}
	\label{alg:rounding}
	\begin{algorithmic}[1]
		\For{arrival of online node $t$}
		\If{$t$ has less than two neighbors}
		\State add two dummy neighbors $i$ with $\Delta x_i = 0$ and $x_i=1$ \Comment{used to simplify notation}
		\EndIf
		\State let $P_t := \{1,2\}$ be the two neighbors of $t$ of highest $\Delta x_i := x_{i,t}$, and let $x_i:=x_{i}^{(t)}$
		\If{$\{1,2\}$ are negative}
		\State match $t$ to its sole free neighbor $i$ (if any) with probability $\frac{\Delta x_i}{1-x_i}$ \label{line:negative-sole-neighbor}
		\EndIf
		\If{$\{1,2\}$ are independent}
		\If{only one $i\in \{1,2\}$ is free}
		\State match $t$ to $i$ \label{line:single-match}
		\EndIf
		\If{both $1$ and $2$ are free}
		\State match $t$ to node $1$ with prob. $\frac{1-x_2- \Delta x_2}{(1-x_1)(1-x_2)}$ and to node $2$ with prob. $\frac{1-x_1- \Delta x_1}{(1-x_1)(1-x_2)}$ \label{line:independent-both-free}
		\EndIf
		\EndIf
		\EndFor
	\end{algorithmic}
\end{algorithm}	

We observe that the algorithm is well-defined (and dummy vertices are never matched). Indeed, the fractional solution guarantees that $\Delta x_i \in [0, 1-x_i]$, so $\frac{\Delta x_i}{1-x_i}\in[0,1]$ and  $\frac{1-x_i- \Delta x_i}{(1-x_1)(1-x_2)}\geq 0$.
On the other hand, since $\Delta x_1 + \Delta x_2 = 1-x_1x_2$, we have in \Cref{line:independent-both-free} that $\frac{1-x_2- \Delta x_2}{(1-x_1)(1-x_2)}+\frac{1-x_1- \Delta x_1}{(1-x_1)(1-x_2)}=1$. The final equality also implies that $1,2$ cannot both be free after time $t$, i.e., $\Pr[F_{\{1,2\},t+1}] = 0$.

We now turn to proving the key lemma in the analysis of \Cref{alg:rounding}, namely, that this algorithm maintains the above desired invariants.
\begin{lem}\label{lem:invariants}
	\Cref{alg:rounding} preserves invariants \eqref{invariant:marginals} and \eqref{invariant:correlation}.
\end{lem}
\begin{proof}
	We prove both invariants in tandem, by induction on $t$. Both invariants clearly hold for $t=1$. Assume the invariants hold for time $t\geq 1$. We prove that this implies the same for time $t+1$.
	For simplicity, we use the shorthand $x_i := x^{(t)}_i$ and $x'_i := x^{(t+1)}_i = x_i + \Delta x_i$, where $\Delta x_i := x_{i,t}$.
	
	\paragraph{Proof of Invariant \eqref{invariant:marginals}:} We prove that $\Pr[(i,t)\in \mathcal{M}] = x_{i,t}$ for each edge $(i,t)\in E$, which implies Invariant \eqref{invariant:marginals} by linearity of expectation.
    We prove the claim for $i=1$ (the proof for $i=2$ is symmetric).
	If $\{1,2\}$ are negative, then, by the inductive hypothesis, $$\Pr[(1,t)\in M] = \Pr[F_{1,t}, \overline{F_{2,t}}]\cdot \frac{\Delta x_1}{1-x_1} = \Pr[F_{1,t}]\cdot \frac{\Delta x_1}{1-x_1} = \Delta x_1.$$
	If $\{1,2\}$ are independent, then by the inductive hypothesis,
	$\Pr[F_{1,t}, F_{2,t}] =
	(1-x_1)(1-x_{2})$ and $\Pr[F_{1,t}, \overline{F_{2,t}}] = (1-x_1)x_2 = x_2-x_1 x_2$.
	Consequently, since $\Delta x_1 + \Delta x_2 = 1-x_1x_2$, we have
	\begin{align*}
	\Pr[(1,t)\in M] & = \Pr[F_{1,t}, \overline{F_{2,t}}] + \Pr[F_{1,t}, F_{2,t}] \cdot \frac{1-x_2- \Delta x_2}{(1-x_1)(1-x_2)} = x_2-x_1x_2 + 1-x_2 - \Delta x_2 = \Delta x_1.
	\end{align*}

	\paragraph{Proof of Invariant \eqref{invariant:correlation}.} By \Cref{obs:negativity}, a set $I$ that ever becomes negative stays negative. Therefore, we only need to consider the case that $I$ (before the current step) is independent, i.e.,
	$\Pr\left[F_{I,t}\right] =\prod_{i\in I}(1-x_i).$
	We prove that before the arrival of node $t+1$, we have \[\Pr\left[F_{I,t+1}\right] \in \left\{0,\prod_{i\in I}(1-x'_i)\right\}.\]
	
	\noindent{\bf The case $\{1,2\}\cap I = \emptyset$:} Then $F_{i,t}=F_{i,t+1}$ and $x_i = x'_i$ for each $i\in I$, and so trivially
	\begin{align*}
	\Pr\left[F_{I,t+1}\right] & = 	\Pr\left[F_{I,t}\right] = \prod_{i\in I}(1-x_i) = \prod_{i\in I}(1-x'_i).
	\end{align*}
	
	\paragraph{The case $\{1,2\} \subseteq I$:} As observed above, after time $t$ at least one of $1,2$ must be matched, i.e., $\Pr[F_{\{1,2\},t+1}] = 0$. Consequently, $\Pr\left[F_{I,t+1}\right] \leq \Pr[F_{\{1,2\},t+1}]=0$.
	
	\paragraph{The case $1\in I$ and $2\not\in I$ (the opposite case is symmetric):}
	Let $E_1$ be the event that $1$ is not matched to $t$.
	There are two sub-cases to consider.
	
%	\vspace{-0.3cm}	
	\paragraph{$\{1,2\}$ are negative.} By \Cref{obs:negativity-subsets}, $I\cup\{2\}$ is also negative, and so by independence of $I$,
	\begin{align*}
	\Pr\left[\overline{F_{2,t}},F_{I,t}\right] = \Pr\left[F_{I,t}\right] - \Pr\left[F_{2,t},F_{I,t}\right] = \Pr\left[F_{I,t}\right] = \prod_{i\in I} (1-x_i). %\label{eqn:Icup2negative}
	\end{align*}
	Consequently, since $x'_1 = x_1 + \Delta x_1$, we have that
	\begin{align*}
	\Pr\left[F_{I,t+1}\right] & =
	\Pr\left[E_1 \mid {F_{1,t}}, F_{2,t}\right] \cdot \Pr\left[{F_{2,t}}, F_{I,t} \right] + \Pr\left[E_1 \mid {F_{1,t}}, \overline{F_{2,t}}\right] \cdot \Pr\left[\overline{F_{2,t}}, F_{I,t} \right] \\
	& = 0 + \left(1-\frac{\Delta x_1}{1-x_1} \right)\cdot \prod_{i\in I}(1-x_i)
	= \prod_{i\in I}(1-x'_i).
	\end{align*}
	
	\vspace{-0.3cm}
	\paragraph{$\{1,2\}$ is independent:}
	If $I\cup \{2\}$ are also independent, then
	$\Pr[F_{2,t},F_{I,t}] = \prod_{i\in I\cup \{2\}}(1-x_i)$.
	Therefore, since $\Pr[E_1 \mid F_{1,t},F_{2,t}] = \frac{1-x_1 - \Delta x_1}{(1-x_1)(1-x_2)}$, and again using $x'_i = x_1+\Delta x_1$, we have that
	\begin{align*}
	\Pr\left[F_{I,t+1}\right] & =
	\Pr\left[E_1 \mid {F_{1,t}}, F_{2,t}\right] \cdot \Pr\left[{F_{2,t}}, F_{I,t} \right]
	+ \Pr\left[E_1 \,\middle\vert\,  F_{1,t}, \overline{F_{2,t}},\right] \cdot \Pr\left[\overline{F_{2,t}}, F_{I,t} \right] \\
	& = \frac{1-x_1 - \Delta x_1}{(1-x_1)(1-x_2)} \cdot (1-x_2) \cdot \prod_{i\in I} (1-x_i) + 0
	= \prod_{i\in I} (1-x'_i).
	\end{align*}
		
	Finally, we address the case that $\{1,2\}$ is independent and $I\cup \{2\}$ is negative.
	In this case, either $2$ is not matched before time $t$, and some node in $I$ must be matched, or
	$2$ is matched before time $t$ and the algorithm matches $1$ in \Cref{line:single-match}. Put otherwise, we have that $I$ becomes negative, as
	\begin{align*}
	\Pr\left[F_{I,t+1}\right] & =
	\Pr\left[E_1 \mid {F_{1,t}}, F_{2,t}\right] \cdot \Pr\left[{F_{2,t}}, F_{I,t} \right] + \Pr\left[E_1 \mid {F_{1,t}}, \overline{F_{2,t}}\right] \cdot \Pr\left[\overline{F_{2,t}}, F_{I,t} \right] = 0.\qedhere
	\end{align*}
\end{proof}

By Invariant \eqref{invariant:marginals} and linearity of expectation, we have that $\Pr[(i,t)\in \mathcal{M}]=x_{i,t}$ for each edge $(i,t)$.
That is, we obtain our main technical result: an online lossless rounding scheme for maximal \algotype two-choice  fractional algorithms.
%To conclude \Cref{thm:rounding-intro} it remains to give a polytime implementation of this algorithm.
%\color{red}
\subsection{Extensions and Refinements}
%\color{black}
\medskip
\noindent\textbf{An Efficient Implementation.}
\Cref{alg:rounding} requires knowledge of whether or not pairs $\{1,2\}$ are negative at time $t$.
That is, it must distinguish between $\Pr[F_{\{1,2\},t}]=0$ and $\Pr[F_{\{1,2\},t}] = \Pr[F_{1,t}]\cdot \Pr[F_{2,t}]$. This can be easily done in exponential time by maintaining the entire probability space.
In \Cref{sec:efficient-implementation} we give a polytime implementation. At the core of this efficient implementation is the (perhaps surprising) observation that the strong invariants of this algorithm allow us to implement it efficiently by explicitly keeping track only of \emph{pairwise correlations}.
\begin{restatable}{lem}{implementation}\label{efficient-implementation}
\Cref{alg:rounding} can be implemented in $O(n)$ time per online node arrival.
\end{restatable}

%\color{red}
\medskip
\noindent\textbf{Generalization to Non-Maximal Algorithms.}
\Cref{alg:rounding} requires a \emph{maximal} \algotype fractional matching algorithm as its input. 
%(This was crucially used to argue that if $\{1,2\}\subseteq I$, then $\Pr[F_{I,t+1}]=0$, serving in some sense as the ``root'' of the negative correlation this algorithm creates.)
The assumption of maximality, while sufficient for some of our applications, limits this rounding algorithm's applicability. Indeed, our semi-OCSes require us to round \emph{non-maximal} fractional matchings.
In \Cref{sec:two-choice-fkg} we show how to generalize \Cref{alg:rounding} to non-maximal inputs, significantly extending its applicability.

\medskip
\noindent\textbf{Low-randomness Implementation.}
Using \Cref{alg:rounding} and its generalization to break the barrier of $\nicefrac{1}{2}$ with only doubly-logarithmic randomness requires several other ideas. First, we design a simple restricted fractional online solution. 
We then show that our rounding algorithms applied to such solutions can be implemented with only $(1+o(1))\log\log n$ bits of randomness (using small-bias probability spaces), while only losing a $(1+o(1))$ factor in the competitive ratio. See \Cref{sec:roundingklevel}.
	\section{Online Roundable Fractional Matching Algorithms}\label{sec:fractional}

In this section we present competitive maximal \algotype two-choice fractional algorithms. We focus on the unweighted problem, deferring discussion of our vertex-weighted algorithms to \Cref{sec:weightedfractional}.

\subsection{The Bounded Water-Level Algorithm}\label{sec:waterlevel}

In this section we design a maximal \algotype two-choice fractional algorithm for unweighted matching.
The algorithm simply picks the two offline vertices with smallest fractional degree. It then applies ``water level" on these two vertices (i.e., it raises the fractional degree of the neighbors of lowest degree among the pair) until it is maximally \algotype (i.e., until \eqref{rounding-weaker-condition} holds with equality). The formal description appears as Algorithm \ref{alg:pairwise-fractional}.

\begin{algorithm}[h]
	\caption{Restricted Water Level}
	\label{alg:pairwise-fractional}
	\begin{algorithmic}[1]
		\State initially, set $\vec{x}\leftarrow \vec{0}$.\label{zero-init}
		\For{arrival of online node $t$}
		\If{$t$ has less than two neighbors} 
		\State add two dummy neighbors $i$ with $x^{(t)}_i = 1$
		\Comment{used to simplify notation}
		\EndIf
		\State let $0\leq x^{(t)}_1\leq x^{(t)}_2\leq \dots \leq x^{(t)}_k \leq 1$ be the fractional degrees of neighbors of $t$
        \State Let $x_f=\frac{x^{(t)}_{1} +x^{(t)}_{2}+ 1-x^{(t)}_1\cdot x^{(t)}_2}{2}$. Set $x_{1,t} \gets x_f - x^{(t)}_1$, $x_{2,t} \gets x_f - x^{(t)}_2$
		\EndFor
	\end{algorithmic}
\end{algorithm}

First, we note that \Cref{alg:pairwise-fractional} is well-defined, as it outputs a feasible fractional matching.

\begin{obs}\label{fractional-feasible}
	\Cref{alg:pairwise-fractional} outputs a maximal \algotype  two-choice fractional matching $\vec{x}$, with $x_{i,t}=0$ for all dummy nodes $i$ and online nodes $t$.
\end{obs}

\begin{proof}
We show by induction on $t$ that $\vec{x}$ satisfies the fractional matching constraints for all nodes.
We thus assume that for all offline nodes $x^{(t)}_i\in [0,1]$.
First,
$$x^{(t)}_2 \leq x^{(t)}_2 + \frac{(1-x^{(t)}_2)\cdot (1+x^{(t)}_{1})}{2} = \frac{x^{(t)}_{1} +x^{(t)}_{2}+ 1-x^{(t)}_1\cdot x^{(t)}_2}{2} = 1 - \frac{(1-x^{(t)}_{1})(1- x^{(t)}_{2})}{2} \leq 1,$$
meaning that $x^{(t)}_1 \leq x^{(t)}_2 \leq x_f \leq 1$. This proves that for all $i$, $x^{(t)}_i \leq x^{(t+1)}_i \leq 1$ (and $x_{i,t}\geq 0$).
In particular, if $x^{(t)}_2=1$ then $x_f=1$ and so we do not increase dummy vertices. Finally,
\begin{align*}
x_{1,t}+x_{2,t} &= 2x_f - (x^{(t)}_1 + x^{(t)}_2) = 1- x^{(t)}_1 \cdot x^{(t)}_2\leq 1. \qedhere\end{align*}
\end{proof}

We now turn to analyzing the competitive ratio of this algorithm. We prove the following:

\begin{lem}\label{lem:fracwater}
Let $g:[0,1]\rightarrow[0,1]$ be a twice differentiable function that is increasing, convex and bijective in $[0,1]$ (and so, in particular, satisfies $g(0)=0$ and $g(1)=1$).
Then, Algorithm \ref{alg:pairwise-fractional} is $\alpha_g$-competitive, where
\begin{equation}\label{comp-ratio-two-choice-WF}
	\alpha_g := \min_{x\in [0,1]} \frac{1-x^2}{1-3g(x)+2g\left(x+\frac{1-x^2}{2}\right)}.
	\end{equation}
\end{lem}

We first prove the following claim.
\begin{claim}\label{derivative-claim}
Let $g:[0,1]\rightarrow[0,1]$ be a twice differentiable function that is monotone increasing, convex and bijective in $[0,1]$.
Then, for any $x\in [0,1]$, we have
$$g'(x) \leq g'(1) \leq \frac{1}{\alpha_g}.$$
\end{claim}
\begin{proof}
	The first inequality follows from $f$ being convex and twice differentiable. For the second inequality, we note that the RHS can be written as
	\begin{align*}
	\alpha_g = \max_{x\in [0,1]} \frac{1-3g(x)+2g\left(x+\frac{1-x^2}{2}\right)}{1-x^2} & = \max_{x\in [0,1]} \left(\frac{1-g(x)}{1-x^2} + \frac{g\left(x+\frac{1-x^2}{2}\right)-g(x)}{\frac{1-x^2}{2}}\right) \\
	&  \geq \max_{x\in [0,1]} \frac{g\left(x+\frac{1-x^2}{2}\right)-g(x)}{\frac{1-x^2}{2}} \\
	& \geq g'(1),
	\end{align*}
	where the first inequality follows follows from $1-g(x)\geq 0$ for $x\in [0,1]$ and so $\frac{1-g(x)}{1-x^2}\geq 0$, while the second inequality follows by considering $x\rightarrow 1$ and applying L'h\^{o}pital's rule.
%	\footnote{For the reader uncomfortable thinking of $\lim_{x\rightarrow 1} \frac{g\left(x+\frac{1-x^2}{2}\right)-g(x)}{\frac{1-x^2}{2}}$ as being of the form $\lim_{x\rightarrow 1} \lim_{h\rightarrow 0} \frac{g(x+h)-g(x)}{h}$, since $h=(1-x^2)/2$ depends on $x$ here, consider applying L'h\^{o}pital's rule, which yields
%		$$\lim_{x\rightarrow 1}\frac{g\left(x+\frac{1-x^2}{2}\right)-g(x)}{\frac{1-x^2}{2}} = \lim_{x\rightarrow 1}\frac{g'\left(x+\frac{1-x^2}{2}\right)\cdot (1-x)-g'(x)}{-x} = g'(1) + \frac{g'(1)-g'(1)}{-1} = g'(1).$$}
\end{proof}

We next prove that \Cref{alg:pairwise-fractional} is $\alpha_g$-competitive.

\begin{proof}[Proof of \Cref{lem:fracwater}]
The proof relies on dual fitting. Let $g$ be a function that satisfies the conditions of the lemma.
We use this (monotone increasing) function $g(\cdot)$ to assign dual values to offline node with fractional degree $x_i$.
 When an online node $t$ arrives, we denote by $0\leq x_1\leq x_2\leq \dots \leq x_k \leq 1$ the fractional degrees of its neighbors (including dummy neighbors).
	\Cref{alg:pairwise-fractional} increases the fractional degree of neighbors $i=\{1,2\}$ of $t$ to $x_f := \frac{x_1+x_2 + 1-x_1x_2}{2} = \hat{x}+ \frac{1-x_1x_2}{2}$, where $\hat{x}:=\frac{x_1+x_2}{2}$. (Note that this is indeed an increase, as by \Cref{fractional-feasible}, for $i=1,2$, we have $x_f - x_i = x_{i,t} \geq 0$, and so $x_f\geq x_i$.) We set the dual of the online node at time $t$ to $1-g(x_2)$, while maintaining the invariant that each offline node $i$ with fractional degree $x_i$ has dual value $y_i = g(x_i)$. This satisfies the dual constraint for each edges $(i,t)$, due to the monotonicity of $y_i$ as a function of $x_i$, implying $1 - g(x_2) + g(x_i) \geq 1$ for all $x_i\geq x_2$, and due to the new final fractional degree of $1$ and $2$ satisfying $x_f \geq x_2$.
	
	We show that the primal gain is at least $\alpha_g$ times the dual cost, which implies the lemma, by weak duality.
	Let $\hat{x}_f= \hat{x}+ \frac{1-\hat{x}^2}{2}$. By the AM-GM inequality we have:
	\begin{align*}
	\hat{x}_f &=  \hat{x}+ \frac{1-\hat{x}^2}{2} \leq \hat{x}+ \frac{1-x_1x_2}{2} = x_f.
	\end{align*}
	We now show that the dual and primal changes satisfy $\Delta D - \frac{1}{\alpha_g }\cdot \Delta P\leq 0$. Indeed,
	\begin{align}
	\Delta D -\frac{1}{\alpha_g } \cdot \Delta P  & = 1-g(x_2) + (g(x_f) - g(x_1)) + (g(x_f)- g(x_2)) -\frac{1}{\alpha_g } \cdot 2 (x_f- \hat{x}) \nonumber \\
	& \leq 1 -g(\hat{x}) + 2g(x_f)- 2g(\hat{x}) -\frac{1}{\alpha_g } \cdot 2 (x_f- \hat{x})\label{ineqtwo11-unweighted}\\
	& \leq 1 -g(\hat{x}) + 2g(\hat{x}_f)- 2g(\hat{x}) -\frac{1}{\alpha_g } \cdot 2 (\hat{x}_f- \hat{x})\label{ineqtwo31-unweighted}\\
	& \leq \max_{x\in[0,1]}\left\{1 - 3g(x) + 2g\left(x+\frac{1-x^2}{2}\right) - \frac{1}{\alpha_g } \cdot (1-x^2)\right\} = 0. \nonumber
	\end{align}
	
	Inequality \eqref{ineqtwo11-unweighted} follows by convexity of $y$ implying $2g(\hat{x})\leq g(x_1)+g(x_2)$, and $x_2 \geq \hat{x}$ together with $y_i$ being monotone increasing in time.
	Inequality \eqref{ineqtwo31-unweighted} follows from $\hat{x}_f \leq x_f$  and the function $g(x) - \frac{1}{\alpha_g }\cdot x$ being monotone decreasing in $x$, as $g'(x)-\frac{1}{\alpha_g }\leq 0$ for all $x\in [0,1]$, by \Cref{derivative-claim}.
	The final equality follows by the definition of $\alpha_g $.
\end{proof}

Finally, we prove the upper and lower bounds on the competitive ratio of the algorithm.
\begin{thm}\label{thm:pairwise-fractional}
Algorithm \ref{alg:pairwise-fractional} is $\alpha$-competitive, where $\alpha$ is at least $\ratio$ and at most $\approx \realratio$.
\end{thm}

\begin{proof}
For the lower bound we note that the function $g(x):=\frac{a^x-1}{a-1}$ for $a=1.6$ satisfies the conditions of Lemma \ref{lem:fracwater}, and achieves a value of $\alpha_g \approx \ratio$.
For the upper bound, we design a bad example that shows that the competitive ratio of the algorithm is at most $\sum_{i\geq 0} \frac{1}{3}\cdot \left(\frac{2}{3}\right)^i \cdot \left(1-2^{-2^{i}+1}\right)\approx\realratio$.

The bad example consists of a bipartite graph with $n = 3^k$ nodes on either side, with a perfect matching.
The online nodes arrive in rounds, as follows.
At the beginning of round $i<k$, a subset of the offline nodes is active, and they all have the same fractional degree. The online nodes of a round each have three distinct neighbors among the active offline nodes. In every such three-tuple of offline nodes (neighboring a common online node in round $i$), one node is not matched at all. This node is chosen to be de-activated. A simple proof by induction shows that the number of active offline nodes in round $i=0,1,2,\dots,k-1$ is $\left(\frac{2}{3}\right)^{i}\cdot n$, while the fractional degree of active nodes in round $i$ is $1-2^{-2^i+1}$. Therefore, the nodes which are de-activated in round $i$ only accrue a gain of $1-2^{-2^i+1}$. As these de-activated nodes in round $i$ are a third of the active nodes in this round, we find that the total gain of the algorithm from nodes de-activaed in round $i$ is $\frac{1}{3}\cdot \left(\frac{2}{3}\right)^{i} n \cdot \left(1-2^{-2^i+1}\right)$. Finally, in the last round, each of the $\left(\frac{2}{3}\right)^k\cdot n = 2^k = o(n)$ active nodes has one distinct online neighbor, and so each of these offline nodes gets a gain of one. The nodes of this last round guarantee the existence of a perfect matching in $G$, consisting of the edges of the last round, together with an edge between every online node $t$ and its offline neighbor de-activated in the round $t$ arrived in. We conclude that the algorithm's competitive ratio is at most
	\begin{align*}
	& \inf_{k} \sum_{i=0}^k \frac{1}{3}\cdot \left(\frac{2}{3}\right)^i\cdot (1-2^{-2^i+1}) +\frac{2^k}{3^k}  \approx \realratio \qedhere
	\end{align*}
\end{proof}

%\color{red}
\begin{remark}
	Numerical approximations show that our upper bound of $0.536$ is tight for \Cref{alg:pairwise-fractional}.
\end{remark}
%\color{black}

\subsection{The $k$-Level Unweighted Algorithm}\label{sec:klevelfrac}

In this section we design a fractional $k$-level algorithm for the unweighted matching problem.
The algorithm uses levels $z_0=0<z_1<z_2<\ldots<z_k<1$, where $z_i= z_{i-1} + \frac{1-z_{i-1}^2}{2} = \frac{(1-z_{i-1})^2}{2}$. Solving the recursion yields $z_i := 1-2^{-2^i+1}$.
%\footnote{Indeed, by induction on $i$, we have that $z_{i+1} = z_i + \frac{1-z_i^2}{2} = \frac{(1-z_i)^2}{2} = \frac{(1-(1-2^{-2^i+1}))^2}{2} = (1-2^{-2^{i+1}+1})$.}
Our algorithm's pseudocode is given by \Cref{alg:k-level}.

\begin{algorithm}[h]
	\caption{The $k$-Level Algorithm}
	\label{alg:k-level}
	\begin{algorithmic}[1]
		\State initially, set $\vec{x}\leftarrow \vec{0}$\label{zero-init}
		\For{arrival of online node $t$}
		\If{$t$ has less than two neighbors} 
		\State add two dummy neighbors $i$ with $x^{(t)}_i = 1$
		\Comment{used to simplify notation}
		\EndIf
		\State let $0\leq x^{(t)}_1\leq x^{(t)}_2\leq \dots \leq x^{(t)}_k \leq 1$ be the fractional degrees of neighbors of $t$
		\If{$x^{(t)}_1 < x^{(t)}_2$ \textbf{or} $x^{(t)}_1 = z_k$}
		\State set $x^{(t+1)}_1\gets 1$ \Comment{$x_{1,t} \gets 1-x^{(t)}_1$}
		\ElsIf{$x^{(t)}_1 = x^{(t)}_2 =z_i < z_k$}		
		\State set $x^{(t+1)}_1 = x^{(t+1)}_2 \gets z_{i+1}$ \Comment{$x_{1,t}=x_{2,t} \gets z_{i+1} - z_i$}
		\EndIf
		\EndFor
	\end{algorithmic}
\end{algorithm}

The following observation shows that Algorithm \ref{alg:k-level} indeed satisfies Definitions \ref{def:precise} and \ref{def:klevel}, and so can be rounded with less randomness.
\begin{obs}\label{obs-klevel}
Algorithm \ref{alg:k-level} is a $k$-level, maximal \algotype and $2^{k-1}$-bit precise algorithm.
\end{obs}

\begin{proof}
It is easy to verify that the steps of the algorithm satisfy the $k$-level requirements in Definition \ref{def:klevel}. Also, since $z_i= z_{i-1} + \frac{1-z_{i-1}^2}{2}$ it is easy to see that the algorithm is a maximal \algotype algorithm.
Finally, we need to prove that the algorithm is $b$-bit precise (satisfies Definition \ref{def:precise}).
We should prove that at any time $t$:
$$\left\{\frac{x_{i,t}}{1-x^{(t)}_i},\, \frac{1-x^{(t)}_i-x_{i,t}}{(1-x_1)(1-x_2)}\right\} \in \left\{\frac{a}{2^b} \,\,\bigg|\,\, a\in  \{0,1,\dots,2^b\}\right\}.$$

As each time $t$, the first term is either $1$ or of the form: $\frac{z_{i} - z_{i-1}}{1-z_{i-1}} = \frac{1-z_{i-1}^2}{2(1-z_{i-1})}=\frac{1+z_{i-1}}{2} = 1-2^{-2^{i-1}}$, where $i\leq k$.
The second term is always of the form $\frac{1-z_{i}}{(1-z_{i-1})^2}=\frac{1}{2}$.
Therefore the algorithm is $b$-bit precise for $b=2^{k-1}$.
\end{proof}

\subsubsection{Warm-up: Analysis of the $2$-level Algorithm}\label{sec:algtwo}

As a warm-up, we analyze the algorithm when $k=2$.
In this case the algorithm uses only two levels: $z_1=\frac{1}{2}, z_2=\frac{7}{8}$. The algorithm has the following 4 cases:
\begin{itemize}
\item $x^{(t)}_{1}=x^{(t)}_{2}=0$: set $x^{(t+1)}_{1}=x^{(t+1)}_{2} \gets\frac{1}{2}$.
\item $x^{(t)}_{1}=x^{(t)}_{2}=\frac{1}{2}$: set $x^{(t+1)}_{1}=x^{(t+1)}_{2} \gets\frac{7}{8}$.
\item $x^{(t)}_{1}=x^{(t)}_{2}=\frac{7}{8}$: set $x^{(t+1)}_{1}\gets 1$.
\item $x^{(t)}_{1}<x^{(t)}_{2}$: set $x^{(t+1)}_{1}\gets 1$.
\end{itemize}

\begin{thm}
The fractional $2$-level  algorithm is $\nicefrac{1}{2}+\nicefrac{1}{36}\approx 0.527$-competitive.
\end{thm}

\begin{proof}
The analysis is via a dual fitting argument. Let $\{1,2\}$ be the two vertices that were increased, and for simplicity we denote by $x_1,x_2$ their fractional degree at time $t$.
We use the following (optimized) values for the dual nodes: $y_1= y(\frac{1}{2}) = 17/38, y_2= y(\frac{7}{8})= 67/76$. These are the values of the offline nodes at the corresponding levels.
To guarantee that the dual solution is feasible, we set the dual value of node $t$ to $1-y(x_2)$. This satisfies the dual constraints of all edges $(i,t)$, and as the dual values are only increasing the dual constraints remain satisfied. We next analyze the four cases, proving that in each one the ratio between the values of the primal and dual changes is at most $1+\frac{17}{19}$. This concludes the proof.

\noindent{\bf Case 1  ($x_1=0,x_2=0$):} $x_1=x_2\gets\frac{1}{2}$. The value of the online node can be set to $1-y(0)=1$. Thus,
\[\frac{\Delta D}{\Delta P} = \frac{1+y_1+y_1}{1} = 1+\frac{17}{19}.\]

\noindent{\bf Case 2  ($x_1=\frac{1}{2},x_2=\frac{1}{2}$):} $x_1=x_2\gets\frac{7}{8}$. The value of the online node can be set to $\min\{1-y_1\}$. Thus,
\[\frac{\Delta D}{\Delta P} = \frac{1-y_1 +2(y_2-y_1)}{2\cdot \frac{3}{8}} = 1+\frac{17}{19}.\]

\noindent{\bf Case 3  ($x_1=\frac{7}{8},x_2=\frac{7}{8}$):}
We set $x_1\gets1$ and we can set the value of the online node to $1-y_2$. Thus,
\[\frac{\Delta D}{\Delta P} = \frac{1-y_2+ 1-y_2}{\frac{1}{8}} = 1+\frac{17}{19}.\]

\noindent{\bf Case 4  ($x_1<x_2$):}
We set $x_1\gets1$ and the value of the online node can be set to $1-y(x_2)$. There  are several possible cases here (some are easily dominated by others in terms of competitiveness). $(x_1,x_2)\in\{(0,\frac{1}{2}), (0,\frac{7}{8}), (0,1),(\frac{1}{2},\frac{7}{8}),(\frac{1}{2},1), (\frac{7}{8},1) \}$. The worst ratio is obtained when $x_1=0, x_2=\frac{1}{2}$. Thus,
\begin{align*}\frac{\Delta D}{\Delta P} & = \frac{1-y(x_2)+ 1-y(x_1)}{1-x_1} \leq \frac{1-y_1 + 1}{1}= 1+\frac{21}{38}.\qedhere
\end{align*}
\end{proof}

\subsubsection{Analysis of the $k$-level Algorithm}

In this section we give a general analysis of the $k$-level algorithm.
We prove the following Lemma.

\begin{lem}\label{thm:k-levels-fractional}
Let $g:[0,1]\rightarrow[0,1]$ be a twice differentiable function that is monotone increasing, convex and bijective in $[0,1]$ (and so, in particular, satisfy $g(0)=0$ and $g(1)=1$).
Then, Algorithm \ref{alg:pairwise-fractional} is $\left(\alpha_g - O(2^{-2^k})\right)$-competitive, where
\begin{equation}\label{comp-ratio-two-choice-WF}
	\alpha_g := \min_{x\in [0,1]}\left\{\frac{1-x^2}{1-3g(x)+2g\left(x+\frac{1-x^2}{2}\right) },\, \frac{1-x}{2 - g(x) -g(x +\frac{1-x^2}{2})}\right\}.
	\end{equation}
\end{lem}

\begin{proof}
The proof relies on dual fitting. Let $g$ be a function that satisfies the conditions of the lemma.
    We use the function $g(\cdot)$ to assign dual values to offline node with fractional degree $x_i$. We assign a dual value $y_i = g(x_i)$.
    For any online node $t$ whose neighbors' fractional degrees at time $t$ are $x^{(t)}_1 \leq x^{(t)}_2 \leq \dots$, we set $y_t = 1-g(x^{(t)}_2)$.
	This dual solution is trivially feasible.
	Now, consider $\Delta P$ and $\Delta D$ following an online node's arrival. We will show that $\Delta P/\Delta D$ is at least $\alpha_g$ for all arrivals except for a small fraction of arrivals, when weighted by their contribution to $P$. This will prove the competitive ratio. There are three cases to consider.
	
	\paragraph{Case 1: $x^{(t)}_1 = x^{(t)}_2 = z_m < z_k$.}
	In this case we have that the primal gain and dual change are
	\begin{align*}
	\Delta P & = 1-z_m^2,\\
	\Delta D & = 2(g(z_{m+1}) - g(z_m)) + 1-g(z_m).
	\end{align*}
	By the definition of $\alpha_g$, and since $z_{m+1} = z_m + \frac{(1-z_m)^2}{2}$, we have $$\Delta P/\Delta D =  \frac{1-z_m^2}{1-3g(z_m)+2g\left(z_m+\frac{1-z_m^2}{2}\right)} \geq \min_{x\in [0,1]} \frac{1-x^2}{1-3g(x)+2g\left(x+\frac{1-x^2}{2}\right)} \geq \alpha_g.$$
	
	\paragraph{Case 2: $x^{(t)}_1 =z_m < x^{(t)}_2 < 1$, with $m<k$.}
	In this case we have that the primal gain and dual cost are
	\begin{align*}
	\Delta P & = 1-z_m,\\
	\Delta D & \leq 1 - g(z_m) + 1-g(z_{m+1}) = 2 - g(z_m) -g(z_{m+1}) = 2 - g(z_m) -g\left(z_m +\frac{1-z_{i-1}^2}{2}\right).
	\end{align*}
 By the definition of $\alpha_g$,
 $$\Delta P/\Delta D =  \frac{1-z_m}{2 - g(z_m) -g(z_m +\frac{1-z_{i-1}^2}{2})} \geq \min_{x\in [0,1]} \frac{1-x}{2 - g(x) -g(x +\frac{1-x^2}{2})} \geq \alpha_g.$$

	\paragraph{Case 3: $x^{(t)}_1 = z_k$.}
	In this case we have that the primal gain and dual change are
	\begin{align*}
	\Delta P & = 1-z_k,\\
	\Delta D & = 1 - g(z_k) + 1-g(z_{k}) = 2 - 2g(z_k).
	\end{align*}

Every time we get to case $3$ the fraction of a new offline node becomes 1. Hence the total dual cost in all these steps, $D'$, is at most $2(1 - g(z_k))P$, where $P$ is the final primal solution.

Summing up over all steps, we get that:
\[P \geq \alpha_g \cdot\left( D - D'\right) = \alpha_g \cdot \left(D- 2 (1 - g(z_k))P\right) \geq \alpha_g \cdot \left(D- \frac{2(1 - z_k)}{\alpha_g}P\right) = \alpha_g \cdot D- 2 (1-z_k) P, \]
where the last inequality follows since by Claim \ref{derivative-claim} for each $g\in \mathcal{F}$, we have
$\frac{1 - g(z_k)}{1-z_k} =g'(\hat{x}) \leq g'(1) \leq \frac{1}{\alpha_g}\leq \frac{1}{\alpha'_g}$, where $0<\hat{x}<1$.
Overall, we get $P\geq \frac{\alpha_g}{1+2(1-z_k)} D= \frac{\alpha_g}{1+2^{-2^k+2}}D = \left(\alpha_g - O(2^{-2^k})\right)D$.
\end{proof}

The proof of the competitive ratio of the $k$-level algorithm follows the same argument as Theorem \ref{thm:pairwise-fractional}. In particular, the same function $g$ can be used to show the competitive ratio, and the adversarial sequence that shows the upper bound is the same.

\begin{thm}\label{thm:pairwise-fractional2}
Algorithm \ref{alg:k-level} is $\left(\alpha_g - O(2^{-2^k})\right)$-competitive with $\alpha_g \in [\ratio,\realratio]$.
\end{thm}

	\section{Conclusion and Open Questions}
We give renewed impetus for the study of online rounding for online bipartite matching problems. 
Indeed, we observe that online rounding can (in principle) yield any competitive ratio achievable for these (and indeed, for any) online problems. This follows by the following ``nonconstructive'' argument:
for any randomized algorithm $\calA$, its marginals yield a fractional algorithm $\calA_f$ that can be rounded losslessly online, by running $\calA$.
As we show, making this approach constructive may require adding additional constraints for the fractional problem, similarly to the addition of constraints to polytopes in offline settings to decrease their integrality gap.
Our qualitative result is a set of new constraints which we prove are sufficient to round two-choice fractional matching algorithms.
Echoing recent results in the area, our obtained two-choice randomized algorithms allow us to break the barrier of $\nicefrac{1}{2}$ (\cite{fahrbach2020edge,huang2020adwords,gamlath2019online}), in the context of randomness and advice complexity.
We further show that this lossless online rounding approach yields simple optimal Semi-OCSes.

Our applications suggest some natural questions: What is the highest competitive ratio achievable using $(1\pm o(1))\log\log n$ random/advice bits? 
What other online correlated selection algorithms can one construct and use? 
One other question implied by our work stands out: 
What conditions allow for lossless roundability of \emph{multiple-choice} algorithms?
An answer to this question requires a deeper understanding of the space of randomized algorithms and constraints on their induced marginals. 
We see this work as a first step in this direction.
		
	\paragraph{Acknowledgements.} We thank the anonymous reviewers for helpful comments on presentation.
	We also thank the anonymous reviewer for pointing out that our additional constraints have some syntactic similarities to the constraints in Border's Theorem~\cite{border91}.

	\appendix
	\section*{Appendix}
	\section{Impossibility of Lossless Online Rounding}\label{sec:impossible}

In this short section we briefly present an example demonstrating the impossibility of online lossless rounding. This example can be seen as a prefix of the example discussed in \citet{devanur2013randomized} when discussing the impossibility of lossless online rounding, which is itself a special case of the lower bound of \citet{cohen2018randomized}.

\medskip

\noindent\textbf{Example 1.}
We consider a bipartite graph with simple a two-choice fractional matching assigning values $x_{i,t}\gets \nicefrac{1}{2}$ to all edges $(i,t)\in E$, and show that this matching cannot be rounded losslessly.
The first two online vertices neighbor offline vertex sets $\{1,2\}$ and $\{3,4\}$, respectively.
A lossless online rounding scheme must match each edge $(i,t)$ with probability $x_{i,t}=\nicefrac{1}{2}$, would match both offline vertices with probability one. Consequently, both these online vertices are matched with probability one.
A simple averaging argument shows that for some pair $(i,j)\in \{1,2\}\times \{3,4\}$, the probability that they are both matched after these two online vertices arrive is at least $\frac{1}{4}$.
Next, if another online vertex arrives that has these two vertices as neighbors, it can be matched with probability at most $\frac{3}{4}$, or strictly less than the fractional solution, which matches it to an extent of one. That is, this two-choice fractional matching cannot be rounded losslessly.

\medskip 

\noindent\textbf{Remark.} We note that the fractional solution in the above algorithm is obtained by the optimal fractional algorithm \textsc{balance} \cite{kalyanasundaram2000optimal}. Therefore, this optimal fractional algorithm is not induced by any randomized algorithm, and consequently is not losslessly roundable online. 
	\section{Rounding, and an FKG-like Inequality}\label{sec:two-choice-fkg}

In this section we show how to round losslessly \algotype two-choice  fractional solutions online, proving \Cref{thm:rounding-intro}.
In \Cref{rounding-algo-main} and \Cref{pairwise-rounding-analysis} we present the lossless rounding algorithm, and prove its properties. In \Cref{sec:rounding} an analysis of a special case of our algorithm for \emph{maximal} \algotype two-choice algorithms, which, while less general, has the advantage of being simpler to describe, having stronger negative correlation properties, and allowing for a polytime implementation.
In \Cref{sec:roundingklevel} we then prove that this algorithm can be implemented with $o(1/n^2)$ additive loss per edge using $(1+o(1))\log\log n$ bits of randomness.

\subsection{The Algorithm and its Invariants}\label{rounding-algo-main}
In this section we design our lossless rounding algorithm. We are given a $2$-choice fractional algorithm that satisfies property \eqref{rounding-weaker-condition} meaning that for any online node $t$, $\sum_{i\in P_t} x_{i,t} \leq 1-\prod_{i\in P_t} x^{(t)}_i$. Let $F_{i,t}$ be the event that offline node $i$ is free (unmatched in $\mathcal{M}$) by time $t$, and let $F_{I,t}:=\bigwedge_{i\in I}F_{i,t}$ be the event that all nodes in $I\subseteq [n]$ are free by time $t$. Our online rounding algorithm maintains the following two invariants for any time $t$.
\begin{align}
\Pr[(i,t)\in \calM] & = x_{i,t} & \qquad \forall (i,t)\in E \label{invariant:marginals-non-maximal}\\
\Pr[F_{i,t} \mid F_{K,t}] & \leq \Pr[F_{i,t} \mid F_{J,t}] & \qquad \forall t,\, \forall i \in [n],\, \forall J\subseteq K\subseteq [n]\setminus \{i\} \textrm{ s.t.} \Pr[F_{K,t}]> 0 \label{invariant:FKG}
\end{align}

The first condition is precisely losslessness, while the second monotonicity property
	is precisely \textbf{log-submodularity} of the function $f(I):=\Pr[F_{I,t}]$,
	\[
	\log \Pr[F_{K+i,t}] - \log \Pr[F_{K,t}] \leq \log \Pr[F_{J+i,t}] - \log \Pr[F_{J,t}].
	\]
	Equivalently, this is a (reverse) \textbf{FKG-like lattice condition}
	 \cite{fortuin1971correlation},
	\[
	\Pr[F_{A\cap B,t}]\cdot \Pr[F_{A\cup B,t}] \leq \Pr[F_{A,t}]\cdot \Pr[F_{B,t}].
	\]
	Invariant \eqref{invariant:FKG} implies negative pairwise correlation between the $F_{i,t}$ variables, i.e.,
	$\Pr[F_{i,t},F_{j,t}] \leq \Pr[F_{i,t}]\cdot \Pr[F_{j,t}]$ for all $i\neq j,$
	and hence between these variables' complements, $M_{i,t} := 1-F_{i,t}$, i.e.,
		$\Pr[M_{i,t},M_{j,t}] \leq \Pr[M_{i,t}]\cdot \Pr[M_{j,t}]$.
	Therefore, combining Condition \eqref{rounding-weaker-condition} with Invariants \eqref{invariant:marginals-non-maximal} and \eqref{invariant:FKG} we obtain the following bound on the probability of any online node $t$ being matched.
	\[
	\Pr[t \textrm{ matched}] \stackrel{\eqref{invariant:marginals-non-maximal}}{=}
	\sum_{i\in P_t} x_{i,t}
	\stackrel{\eqref{rounding-weaker-condition}}{\leq}
	1 - \prod_{i\in P_t} x^{(t)}_i
	\stackrel{\eqref{invariant:marginals-non-maximal}}{=} 1 - \prod_{i\in P_t} \Pr[M_{i,t}]
	\stackrel{\eqref{invariant:FKG}}{\leq }
	1 - \Pr\left[\bigwedge_{i\in P_t} M_{i,t}\right].
	\]
	The conclusion of this chain of inequalities, whereby $\Pr[t \textrm{ matched}]\leq 1 - \Pr\left[\bigwedge_{i\in P_t} M_{i,t}\right]  = 1-\Pr[\textrm{all nodes in $P_t$ matched before time $t$}]$,
	is a trivial necessary condition for any randomized matching algorithm.
	We show that the conditions we impose on our fractional solution,
	together with the invariants we maintain, allow us to inductively maintain
	these properties, while outputting a randomized matching, online.
	
We next describe formally the algorithm. Assume without loss of generality that online node $t$ increases two neighbors: $1,2$.
Let $x_1 := x^{(t)}_1,x_2:=x^{(t)}_2$ be their fractional values, and $\Delta x_1:=x_{1,t}, \Delta x_2:=x_{2,t}$ be their change. By the properties of the fractional agorithm we are guaranteed that  we have that $\Delta x_1+\Delta x_2\leq 1-x_1 x_2$. Our pseudocode is given in \Cref{alg:pairwise-rounding}.

	\begin{algorithm}[H]
		\caption{Online Lossless Rounding}
		\label{alg:pairwise-rounding}
		\begin{algorithmic}[1]
			\For{arrival of online node $t$}
			\If{$t$ has less than two neighbors}
			\State add two dummy neighbors $i$ with $\Delta x_i = 0$ and $x_i=1$ \Comment{used to simplify notation}
			\EndIf
			\State let $P_t := \{1,2\}$ be the two neighbors of $t$ of highest $\Delta x_i := x_{i,t}$, and let $x_i:=x_{i}^{(t)}$\label{line:potential-matches}
			\State let $p_{12} := \Pr[F_{\{1,2\},t}]$ \Comment{assuming $p_{12}$ is known}\label{line:pair-prob}
			\State let $a_1,a_2,b_1,b_2$ be solutions to Program \eqref{probs-program} with input $\Delta x1,\Delta x_2, x_1,x_2,p_{12}$
			\If{$1,2$ are both free}
			\State match $t$ to $1$ with probability $a_1$ and to $2$ with probability $a_2$ \label{line:both-free}
			\ElsIf{a single $i\in \{1,2\}$ is free}
			\State match $t$ to $i$ with probability $b_i$
			 \label{line:single-free}
			\EndIf
			\EndFor
		\end{algorithmic}
	\end{algorithm}

		\textbf{Probability-Setting Program($\Delta x1,\Delta x_2, x_1,x_2,p_{12}$):} \vspace{-0.8cm}
		\begin{align}
		& & & \label{probs-program}\tag{Prob-Program} \\
		a_1 + a_2 & \leq 1 \label{suma<=1}\\
		a_i & \geq 0 & \forall i = 1,2 \label{a>=0} \\
		b_i & \leq 1 & \forall i = 1,2 \label{b<=1} \\
		b_i & \geq a_i & \forall i=1,2 \label{b>=a} \\
		b_i & \leq \frac{a_i}{1-a_{3-i}} & \forall i=1,2 \label{bi<=a1/a2} \\
		a_i\cdot p_{12} + b_i\cdot (1-x_i-p_{12}) & = \Delta x_i & \forall i=1,2 \label{marginal-constraint}
		\end{align}
	
	\begin{remark}\label{remark:special-case}
		\Cref{alg:rounding} is the special case of \Cref{alg:pairwise-rounding} obtained from the solution to Program \eqref{probs-program} with $a_{i} = \frac{1-x_{3-i}- \Delta x_{3-i}}{(1-x_1)(1-x_2)}$ and $b_i = \frac{\Delta x_i}{1-x_i}$ or $b_i=1$ if $\{1,2\}$ are negative/independent.
	\end{remark}
	
	We first show that the algorithm's steps at time $t$ are well-defined, provided our claimed invariants hold until this time.
	\begin{lem}\label{lem:well-defined}
		Assuming invariants \eqref{invariant:marginals-non-maximal} and \eqref{invariant:FKG} hold before time $t$, then the algorithm's steps at time $t$ are well-defined. In particular, Program \eqref{probs-program} is solvable (efficiently). Consequently,
		\begin{enumerate}
			\item $a_1 + a_2 \leq 1$ and $a_i\geq 0$ for all $i=1,2$. \hfill (\Cref{line:both-free} is well-defined) \label{both-free-well-defined}
			\item $b_i\in [a_i,1]\subseteq [0,1]$ for all $i=1,2$. \hfill (\Cref{line:single-free} is well-defined) \label{single-free-well-defined}
		\end{enumerate}
	\end{lem}
	\begin{proof}
		Properties \ref{both-free-well-defined} and \ref{single-free-well-defined} follow from constraints \eqref{suma<=1}, \eqref{a>=0}, \eqref{b<=1} and \eqref{b>=a} of Program \eqref{probs-program}. It remains to prove that this program is (efficiently) solvable, which we do using the following algorithm.
		Initially, we set $a_i,b_i\leftarrow \Delta x_i / (1-x_i)$ for both $i=1,2$.
		If $\sum_i \Delta x_i / (1-x_i) \leq 1$, we terminate, as this solution satisfies all the constraints of Program \eqref{probs-program},
		with the non-trivial constraints following from the fractional matching constraints implying $\Delta x_i/(1-x_i)\in [0,1]$.
		Otherwise, for $i=1,2$, in any order, we decrease $a_i$ and increase $b_i$ while maintaining \Cref{marginal-constraint}, until $a_1+a_2=1$ or $b_i=1$. (While we state this algorithm as a continuous algorithm, it is trivial to discretize and implement it in constant time.)
		We note that one of the two stopping conditions will occur. Indeed, if we set $a_i=0$, then, since $a_{3-i}\leq \Delta x_{3-i}/(1-x_{3-i})\leq 1$, we have that $a_1+a_2\leq 1$.
		We conclude that by the algorithm's termination, constraints \eqref{suma<=1} and \eqref{a>=0} and \eqref{b<=1} are satisfied. Moreover, by construction (of the algorithm), the equality constraint \eqref{marginal-constraint} is satisfied. From this, we obtain the following.
		\begin{align}
		a_i \cdot \frac{p_{12}}{1-x_i} + b_i\cdot \frac{1-x_i-p_{12}}{1-x_i} & = \frac{\Delta x_i}{1-x_i}. \label{Delta-convex-combo-a+b}
		\end{align}
		Now, by invariants \eqref{invariant:marginals-non-maximal} and \eqref{invariant:FKG}, we have that $p_{12} = Pr[F_{\{1,2\},t}] \leq Pr[F_{1,t}]\cdot Pr[F_{2,t}] = (1-x_1)(1-x_2)$. Therefore, $p_{12}\leq (1-x_i)$ and so \Cref{Delta-convex-combo-a+b} implies that $\Delta x_i/(1-x_i)$ is a convex combination of $a_i$ and $b_i$. Since we initialize $a_i,b_i\leftarrow \Delta x_i/(1-x_i)$, and decrease $a_i$ while increasing $b_i$, we obtain $b_i\geq \Delta x_i/(1-x_i)\geq a_i$, implying Constraint \eqref{b>=a}.
		Finally, to prove Constraint \eqref{bi<=a1/a2}, we show that if $\sum_i \Delta x_i/(1-x_i) > 1$, then $a_1+a_2=1$, and so Constraint \eqref{bi<=a1/a2} follows from Constraint \eqref{b<=1}, since $a_i/(1-a_{3-i})= 1$. Indeed, if $a_1+a_2>1$ by the algorithm's termination, then we must have stopped both iterations of the loop decreasing $a_i$ and increasing $b_i$ after reaching $b_i=1$. But then, we have
		\begin{align*}
			\Delta x_1 + \Delta x_2 & = b_1\cdot (1-x_1-p_{12}) + a_{1}\cdot p_{12} + b_2\cdot (1-x_2-p_{12}) + a_{2}\cdot p_{12} & \\
			& = (1-x_1)+(1-x_2) + (a_1+a_2-2)\cdot p_{12} & b_1=b_2=1 \\
			& > (1-x_1)+(1-x_2) - p_{12} & a_1+a_2>1 \\
			& \geq (1-x_1)+(1-x_2) - (1-x_1)(1-x_2) & \textrm{inv. \eqref{invariant:marginals-non-maximal} and \eqref{invariant:FKG}} \\
			& = 1-x_1\cdot x_2,
		\end{align*}
		thus contradicting Condition \eqref{rounding-weaker-condition}, i.e., $\Delta x_1 + \Delta x_2 \leq 1-x_1\cdot x_2$. We conclude that this algorithm terminates with a feasible solution to \eqref{probs-program}, and thus \Cref{alg:pairwise-rounding} is well-defined.
	\end{proof}
	
	\subsection{Lossless Rounding using \Cref{alg:pairwise-rounding}}\label{pairwise-rounding-analysis}
	So far, we have proven that assuming the claimed invariants---\eqref{invariant:marginals-non-maximal} and \eqref{invariant:FKG}---hold prior to time $t$, then \Cref{alg:pairwise-rounding} is well defined.
	We now prove that if these invariants hold prior to time $t$, then they likewise hold prior to time $t+1$.
	
	For our proof we will need the following simple corollary of Bayes' Law.
	\begin{obs}\label{obs:bayes-double-condition-ind}
		If $(A,B)\bot (C,D)$, (i.e., $(A,B)$ and $(C,D)$ are independent), then
		$$\Pr[A,C\mid B,D] = \Pr[A\mid B]\cdot \Pr[C\mid D].$$
		The special case of $\Pr[C] = 1$ implies that if $D\bot (A,B)$, then $\Pr[A\mid B,D] = \Pr[A\mid B].$
	\end{obs}
	\begin{proof}
		By Bayes' Law, we have that
		\begin{align*}
		\Pr[A,C\mid B,D] & = \frac{\Pr[A,B,C,D]}{\Pr[B,D]} =\frac{\Pr[A,B]\cdot \Pr[C,D]}{\Pr[B]\cdot\Pr[D]} = \Pr[A\mid B]\cdot \Pr[C\mid D].\qedhere
		\end{align*}
	\end{proof}

	\begin{lem}\label{both-invariants}
		\Cref{alg:pairwise-rounding} satisfies Invariants \eqref{invariant:marginals-non-maximal} and \eqref{invariant:FKG}.
	\end{lem}
	\begin{proof}
		First, we prove Invariant \eqref{invariant:marginals-non-maximal}.
		Fix a time $t$. By construction (of the algorithm) we trivially have $\Pr[(i,t)\in \mathcal{M}] = 0$ for all $i$ with $\Delta x_i=0$. Now, let $P_t=\{1,2\}$ be as in \Cref{line:potential-matches}, and let $i\in P_t$.
		Then, by our choice of $a_i,b_i$, and Constraint \eqref{marginal-constraint}, the probability $i$ is matched to $t$ is precisely
		\begin{align*}
		\Pr[(i,t)\in \mathcal{M}] & =
		\Pr[(i,t)\in \mathcal{M}, F_{P_t,t}] + \Pr[(i,t)\in \mathcal{M}, F_{i,t}, \overline{F_{3-i,t}}] \\
		& = a_i \cdot p_{12} + b_i \cdot (p_1 - p_{12}) = \Delta x_i.
		\end{align*}		
We now turn to proving Invariant \eqref{invariant:FKG}.
		We prove this invariant holds for all tuples $(i,J,K,t)$, by induction on $t$. The invariant clearly holds for $t=1$. Assume the invariant holds for time $t\geq 1$. We prove that this implies the same for time $t+1$.
		For the inductive step, when wishing to prove Invariant \eqref{invariant:FKG} for the tuple $(i,J,K,t+1)$, we may safely assume that both $\Pr[F_{K,t}] \neq 0$ and $\Pr[F_{K\cup \{i\},t}] \neq 0$ hold.
		Indeed, the converse would imply that $\Pr[F_{K\cup \{i\},t+1}]\leq \Pr[F_{K\cup \{i\},t}] = 0$, in which case Invariant \eqref{invariant:FKG} holds trivially for this tuple.

		Let $E_i$ denote the event that the algorithm does not match $(i,t)$ at time $t$.
		We further denote by $C_1\sim \Ber(a_1)$ and $C_2\sim \Ber(a_2)$ the Bernoulli random variables corresponding to the probability of matching $t$ to $1$ and $2$, respectively, in \Cref{line:both-free}, if both $1$ and $2$ are free at time $t$. We can imagine our algorithm tosses these (correlated) coins regardless of the event $F_{\{1,2\},t}$, and only inspects these variables if the event $F_{\{1,2\},t}$ occurs.
		We note that the random variables $C_1$ and $C_2$ are independent of all events determined by random choices made by the algorithm until time $t$.

		With this notation and these observations at hand, we now turn to proving the desired invariant holds for the tuple $(i,J,K,t+1)$.
		There are five cases to consider, based on the inclusions between $P_t = \{1,2\}$ and $K\cup\{i\}$, where if $i\in \{1,2\}$, we assume without loss of generality that $i=1$.
		
		\paragraph{\underline{Case 1: $\{1,2\}\cap (K\cup \{i\}) = \emptyset$.}} In this case $F_{I,t+1}\equiv F_{I,t}$ for all $I\subseteq K\cup \{i\}$, and so the invariant follows trivially from the inductive hypothesis.
		
		\paragraph{\underline{Case 2: $\{1,2\}\subseteq K$:}}
		By the inductive hypothesis, and independence of $(C_1,C_2)$ from $F_{i,t}$, we obtain the desired inequality for the tuple $(i,J,K,t+1)$.
\begin{align*}
\Pr[F_{i,t+1} \mid F_{J,t+1} ] = \Pr[F_{i,t} \mid F_{J,t}]  \geq \Pr[F_{i,t} \mid F_{K,t}] = \Pr[F_{i,t+1} \mid F_{K,t+1}].
\end{align*}
Here, the equalities follow from $F_{i,t}\equiv F_{i,t+1}$ and $F_{i,t} \bot (C_1,C_2)$ together with \Cref{obs:bayes-double-condition-ind}, while the inequality follows from the inductive hypothesis.
		
		\paragraph{\underline{Case 3: $\{1,2\}\cap K = \emptyset, i=1$.}}
		For this case we rely on the probability of $1$ not being matched decreasing when we condition on a larger set of offline nodes being free, as in the following inequality.
		\begin{align}\label{monotone-skip-prob}
			\Pr[E_1 \mid F_{K\cup \{1\},t}]\leq \Pr[E_1 \mid F_{J\cup\{1\},t}].
		\end{align}
		Indeed, subtracting $\Pr[E_1 \mid F_{J\cup\{1\},t}]$ from both sides, expanding both terms using the law of total probability,
		we get
		\begin{align}
		& \Pr[E_1 \mid F_{K\cup\{1\},t}] - \Pr[E_1 \mid F_{J\cup\{1\},t}] \nonumber \\
		= & \Pr[E_1 \mid F_{\{1,2\},t}]\cdot \Big(\Pr[F_{2,t} \mid  F_{K\cup \{1\},t}] - \Pr[F_{2,t} \mid  F_{J\cup \{1\},t}]\Big)  \nonumber \\
		+ & \Pr[E_1 \mid F_{1,t},\overline{F_{2,t}}]\cdot \Big(\Pr[\overline{F_{2,t}} \mid F_{K\cup\{1\},t}]-\Pr[\overline{F_{2,t}} \mid F_{J\cup\{1\},t}]\Big) \nonumber \\
		= & ((1-a_1)-(1-b_1))\cdot (\Pr[F_{2,t} \mid  F_{K\cup \{1\},t}] - \Pr[F_{2,t} \mid  F_{J\cup \{1\},t}]) ~\leq~0.
\label{monotone-skip-chance}
		\end{align}		
		Here, the second equality follows from $\Pr[E_1 \mid F_{\{1,2\},t}] = 1-a_1$ and $\Pr[E_1 \mid F_{1,t}, \overline{F_{2,t}}] = 1-b_1$ by definition. For any event $A$, $\Pr[F_{2,t} \mid A] + \Pr[\overline{F_{2,t}} \mid A ]=1$,  implying
		$$\Pr[F_{2,t} \mid  F_{K\cup \{1\},t}] - \Pr[F_{2,t} \mid  F_{J\cup \{1\},t}] = -(\Pr[\overline{F_{2,t}} \mid  F_{K\cup \{1\},t}] - \Pr[\overline{F_{2,t}} \mid  F_{J\cup \{1\},t}]).$$
		Finally, Inequality \eqref{monotone-skip-chance} follows from Constraint \eqref{b>=a} implying that $(1-a_1) - (1-b_1)\geq 0$, and by the inductive hypothesis together with the assumption that $\Pr[F_{K\cup \{1\},t}]\neq 0$ implying that $\Pr[F_{2,t} \mid  F_{K\cup \{1\},t}] - \Pr[F_{2,t} \mid  F_{J\cup \{1\},t}]\leq 0$. We conclude that \Cref{monotone-skip-prob} holds.

		The desired inequality of Invariant \eqref{invariant:FKG} for the tuple $(i,J,K,t+1)$ then follows from \Cref{monotone-skip-prob}, the inductive hypothesis and the assumption that $\Pr[F_{K,t}]\neq 0$, implying
		\begin{align*}
			\Pr[F_{1,t+1} \mid F_{K,t+1}] &= \Pr[F_{1,t+1} \mid F_{K,t}] \\
			%&= \Pr[E_1, F_{1,t} \mid  F_{K,t}] \\
			& = \Pr[E_1 \mid F_{K\cup\{1\},t}] \cdot \Pr[F_{1,t} \mid F_{K,t}] \\
			& \leq \Pr[E_1 \mid F_{J\cup\{1\},t}]\cdot \Pr[F_{1,t} \mid F_{J,t}] & \textrm{I.H. + \eqref{monotone-skip-prob}} \\
			& = \Pr[E_1, F_{1,t} \mid  F_{J,t}] \\
			%    & = \Pr[F_{1,t+1} \mid F_{J,t}]\\
			& = \Pr[F_{1,t+1} \mid F_{J,t+1}].
		\end{align*}
		
		\paragraph{\underline{Case 4:  $\{1,2\}\cap (K\setminus J)  =\{2\}$, and $i=1$.}}
		%Negative correlation of
		Independence of $(C_1,C_2)$ from $(F_{1,t},F_{K,t})$, and the inductive hypothesis yield the desired inequality for the tuple $(i,J,K,t+1)$, as follows.
		\begin{align*}
			\Pr[F_{1,t+1} \mid F_{K,t+1}] & = \Pr[\overline{C_1},  F_{1,t} \mid F_{K,t},  \overline{C_2}] \\
			& = \Pr[\overline{C_1} \mid \overline{C_2}]\cdot  \Pr[F_{1,t} \mid F_{K,t}]
			&  (C_1,C_2)\bot (F_{1,t},F_{K,t}) + \textrm{\ref{obs:bayes-double-condition-ind}} \\
			& = \left(1-\frac{a_1}{1-a_2}\right)\cdot \Pr[F_{1,t} \mid F_{K,t}] \\
			& \leq (1-b_1) \cdot \Pr[F_{1,t} \mid F_{K,t}] & \textrm{\eqref{bi<=a1/a2}} \\
			& \leq (1-b_1) \cdot \Pr[F_{1,t} \mid F_{J,t}] & \textrm{I.H.} \\
			& \leq \Pr[E_1 \mid F_{1,t}] \cdot \Pr[F_{1,t} \mid F_{J,t}] & a_1\leq b_1 \\
			& = \Pr[F_{1,t+1}\mid F_{J,t}] \\
			& = \Pr[F_{1,t+1}\mid F_{J,t+1}],
		\end{align*}
		where the last inequality relied on Constraint \eqref{b>=a}, whereby $a_1 \leq b_1$, implying that $$\Pr[E_1 \mid F_{1,t}] = (1-a_1)\cdot \Pr[F_{2,t} \mid F_{1,t}] + (1-b_1)\cdot \Pr[\overline{F_{2,t}} \mid F_{1,t}] \geq (1-b_1).$$
		
		\paragraph{\underline{Case 5: $\{1,2\}\cap J=\{2\}$, and $i=1$.}}
		Independence of $(C_1,C_2)$ from $(F_{1,t},F_{J,t},F_{K,t})$, together with the inductive hypothesis, proves the desired inequality for the tuple $(i,J,K,t+1)$.
		\begin{align*}
			\Pr[F_{1,t+1} \mid F_{K,t+1}] & = \Pr[\overline{C_1}, F_{1,t} \mid \overline{C_2}, F_{K,t}] \\
			& = \Pr[\overline{C_1} \mid \overline{C_2}] \cdot \Pr[F_{1,t} \mid F_{K,t}]
			& (C_1,C_2)\bot (F_{1,t}, F_{K,t})+ \textrm{\ref{obs:bayes-double-condition-ind}}\\
			& \leq \Pr[\overline{C_1} \mid \overline{C_2}] \cdot \Pr[F_{1,t} \mid F_{J,t}]
			& \textrm{I.H.} \\
			& = \Pr[\overline{C_1}, F_{1,t} \mid \overline{C_2}, F_{J,t}]
			& (C_1,C_2)\bot (F_{1,t}, F_{J,t}) + \textrm{\ref{obs:bayes-double-condition-ind}}\\
			& = \Pr[F_{1,t+1}\mid F_{J,t+1}].
			& & \qedhere
		\end{align*}
	\end{proof}
Combining \Cref{lem:well-defined} and \Cref{both-invariants}, we find that \Cref{alg:pairwise-rounding} is well-defined throughout its execution. Moreover, we find that each edge is matched with the appropriate marginal probability prescribed by the fractional solution. In other words, we obtain the following.
	
	\begin{thm}\label{thm:lossless-rounding}
		\Cref{alg:pairwise-rounding}, when run on a fractional matching $\vec{x}$ satisfying Condition \eqref{rounding-weaker-condition}, outputs a random matching $\mathcal{M}$ such that
		$$\Pr[(i,t)\in \mathcal{M}] = x_{i,t} \qquad \forall (i,t)\in E.$$
	\end{thm}

	\section{Small Random Seed for $k$-level Algorithms}\label{sec:roundingklevel}
The randomized algorithms derived from \Cref{alg:pairwise-rounding} require (at least) polynomially-large random seeds.
In this section we show that this is not really necessary, at least for the special case of \Cref{alg:pairwise-rounding} given by \Cref{alg:rounding}. In particular, we show that essentially the same competitive ratio can be achieved using only a doubly-logarithmic random seed.
%We show this for the the special case of \Cref{alg:pairwise-rounding} given by \Cref{alg:rounding}.

The need for a large random seed of our rounding algorithms of the previous sections is due to two reasons.
The first one is because of precision issues:
some of the probabilities in this algorithm can be arbitrarily small, and so these require arbitrarily-large random
seeds.
We overcome this first issue by explicitly restricting our attention to algorithms requiring only $b$ bits of randomness to determine the random choices of \Cref{alg:rounding}, as follows.

\begin{Def}\label{def:precise}
	A fractional algorithm $\mathcal{A}$ is \emph{$b$-bit precise} if for each online node $t$ with $P_t = \{1,2\}=\{i \mid x_{i,t}>0\}$, the fractional matching $\vec{x}$ output by $\mathcal{A}$ satisfies
	$$\left\{\frac{x_{i,t}}{1-x^{(t)}_i}, \frac{1-x^{(t)}_i-x_{i,t}}{(1-x_1)(1-x_2)}\right\} \in \left\{\frac{a}{2^b} \,\,\bigg\vert\,\, a\in  \{0,1,\dots,2^b\}\right\}.$$
\end{Def}
The second, more fundamental reason, for the large random seed is our (implicit) requirement of complete independence between the random choices during each time step.
For $n$ random variables---one per arrival---this trivially requires at least $n$ random bits.
As we show, a significant saving over this amount of randomness can be obtained by considering small-bias distributions.
For this, we will further restrict our attention to the following kind of two-choice algorithms.

\begin{Def}[$k$-level Algorithm]\label{def:klevel} A $k$-level algorithm has some $k+2$ possible values, denoted by $0=z_0<z_1<\ldots<z_k<z_{k+1}=1$, and maintains the invariant that each offline node has fractional degree equal to one of these $z_i$. At each step of the algorithm, the fractions $x^{(t)}_1,x^{(t)}_2$ of at most two offline nodes $\{1,2\}$ are increased to $x^{(t+1)}_1, x^{(t+1)}_2$, with the following options:
	\begin{itemize}
		\item {\bf deterministic step:} $x^{(t)}_1$ is increased to $1$. ($x^{(t+1)}_2\gets x^{(t)}_2$.)
		\item {\bf random step:} $x^{(t)}_1,x^{(t)}_2$ are increased to $x^{(t+1)}_1, x^{(t+1)}_2>\max\{x^{(t)}_1,x^{(t)}_2\}$ (strict inequality).
		\item {\bf shift step:} $x^{(t)}_1=0$ and $x^{(t)}_2\in (0,1)$ are increased to $x^{(t+1)}_1\gets x^{(t)}_2, x^{(t+1)}_2\gets 1$.
	\end{itemize}
\end{Def}

We show that when rounding such fractional algorithms using \Cref{alg:rounding}, the following holds: for each edge $(i,t)$, the event that $(i,t)$ is matched in the random matching $\mathcal{M}$ output by \Cref{alg:rounding} depends on a bounded number of random choices of this algorithm. To this end, we denote by $A_{i,t} \sim \Ber(a^t_i)$ the random variable corresponding to the random choice in \Cref{line:negative-sole-neighbor} of \Cref{alg:rounding}
and by $B_{i,t} \sim \Ber(b^t_i)$ the random variables of \Cref{line:independent-both-free}, where $a^t_i, b^t_i$ are the solution to Program \eqref{probs-program} used by the algorithm at time $t$ (see \Cref{remark:special-case}).
We prove the following.

\begin{lem}\label{low-dependence}
	The event $[(i,t)\in \mathcal{M}]$ is determined by at most $2^{k+2}$ random variables  $\{A_{i',t'},B_{i',t'}\}_{i',t'}$.
\end{lem}
\begin{proof}
	We say an offline node $i$ has level $\ell$ before time $t$ if $x_i^{(t)}=z_\ell$.
	We prove by induction on all times $t$ and on the level $\ell\leq k$ of node $i$ before time $t$ that $F_{i,t}$ is determined by at most $2^{\ell+1}-2$ random variables.
	From this we find that the edge $(i,t)$ is matched, $\mathds{1}[(i,t)\in \mathcal{M}] = F_{i,t} - F_{i,t+1}$, is determined by at most $2^{k+1} - 2 + 2^{k+1} - 2 \leq 2^{k+2}$ random variables in the set $S:=\{A_{i',t'},B_{i',t'}\}_{i',t'}$.
	
	For nodes at level $0$, we have that $F_{i,t}\equiv 1$. Consequently, since at time $t=1$ all offline nodes $i$ are at level $\ell=0$, the variables $F_{i,t}$ are determined by $0=2^{0}-1$ random variables.
	Now, consider a time $t$ where the level of $i$ increases, depending on what kind of step caused the increase to level $\ell$.
	If this increase is due to a \underline{deterministic step} (i.e., $\ell=k$), then $F_{i,t+1}\equiv 0$ is deterministic, by Invariant \eqref{invariant:marginals-non-maximal}, and so this variable depends on $0\leq 2^{k+1}-2$ random variables in $S$.
	Suppose next that the level increase of $i$ is due to a \underline{random step}, with $j$ the second neighbor of $t$ whose level increases at time $t$.
	Then we have that $F_{i,t+1} = F_{i,t}\cdot \left(F_{j,t}\cdot \overline{A_{i,t}} + \overline{F_{j,t}}\cdot \overline{B_{i,t}}\right)$.
	Consequently, since both $i$ and $j$ had level strictly lower than the new level $\ell$ of $i$, we have by the inductive hypothesis that $F_{i,t+1}$ is determined by $2+2\cdot (2^{\ell}-2) = 2^{\ell+1}-2$ random variables in $S$.
	
	Finally, if the level of $i$ increased to level $\ell$ due to a \underline{shift step}, then if $\ell=k$, as argued before, $F_{i,t+1}\equiv 0$, and therefore $F_{i,t+1}$ is a function of zero variables. Otherwise, the level of the other node $j$ whose level increased was $\ell$ before, while the previous level of $i$ was zero. 	
	Consequently, by Invariant \eqref{invariant:marginals-non-maximal}, we have that $t$ is matched with probability one.
	That is, $(F_{i,t} - F_{i,t+1}) + (F_{j,t} - F_{j,t+1}) = 1$.
	But, similarly, by Invariant \eqref{invariant:marginals-non-maximal} we have that $F_{j,t+1}\equiv 0$, and since $i$ was at level zero before time $t$, we have that $F_{i,t}\equiv 1$.
	Putting the above together, we find that $F_{i,t+1}=F_{j,t}$. Consequently, by the inductive hypothesis, since $j$ previously had level $\ell$, we have that $F_{i,t+1}=F_{j,t}$ is determined by at most $2^{\ell+1}-2$ random variables in $S$.
\end{proof}

We can now rely on our analysis for \Cref{alg:rounding} using independent random variables to analyze the same algorithm when using $(\delta
,b\cdot 2^{k+1})$-dependent binary variables to sample variables  $A_{i,t},B_{i,t}$.
In particular, we can show that such a random seed yields an essentially lossless rounding.

\begin{thm}\label{thm:low-randomness}
	Let $\mathcal{M}$ be the random matching output by
	\Cref{alg:rounding} when rounding a $b$-bit precise
	$k$-level algorithm, using a distribution $\mathcal{D}$ over $(\delta,b\cdot 2^{k+2})$-dependent binary variables for the random choices.
	Then,
	$$\Pr_\mathcal{D}[(i,t)\in \mathcal{M}] = x_{i,t} \pm \delta \qquad \forall (i,t)\in E.$$
\end{thm}
\begin{proof}
	By definition of $b$-bit precise algorithms, each probability used by \Cref{alg:rounding} can be specified using $b$ (random) bits.
	On the other hand, by \Cref{low-dependence}, each event $[(i,t)\in \mathcal{M}]$ is determined by $2^{k+2}$ random variables, or $b\cdot 2^{k+2}$ random binary variables.
	Now, if we denote by $\mathcal{U}$ the uniform distribution, then by  \Cref{thm:lossless-rounding} and \Cref{reduction-to-independent} we obtain the desired result,
	\begin{align*}
	\Pr_{\mathcal{D}}[(i,t)\in \mathcal{M}] & = \Pr_{\mathcal{U}}[(i,t)\in \mathcal{M}] \pm \delta = x_{i,t} \pm \delta \qedhere.
	\end{align*}
\end{proof}

Since there are $n$ offline nodes and each node can change levels at most $k$ times, the total number of random variable $A_{i,t}, B_{i,t}$ is bounded by $O(nk)$.
Hence, by \Cref{eps-k-constructions} and the above lemma, all $k$-level algorithms can be rounded with additive loss $\delta = \frac{1}{\log \log n} = o(1)$, using $(1+o(1))\log\log n+ 2^{k+2}b$ random bits. Now, using the fractional (weighted) $k$-level algorithms of \Cref{sec:fractional}, and observing that for any constant $k$, these algorithms satisfy $x_{i,t}=\Omega(1)$ whenever $x_{i,t}\neq 0$, we have that the above yields a $(1+o(1))$ \emph{multiplicative} loss compared to these $0.5363$- and $0.524$-competitive fractional matching and vertex-weighted matching algorithms.
This concludes the proof of \Cref{thm:threshold}.

We note that using standard $k$-wise independence, one can round a $k$-level $b$-bit algorithm without \emph{any} loss (even $o(1)$) using $O(2^{k+2}b\cdot\log (nk))$ bits of randomness.

Finally, we note that an efficient implementation of \Cref{alg:pairwise-rounding} (or its special case, \Cref{alg:rounding}) with perfect independence implies a low-randomness implementation with similar running time, only slowed down by the time to sample from a $(\delta,b\cdot 2^{k+2})$-dependent distributions. So, for example, \Cref{efficient-implementation} together with \Cref{thm:low-randomness} yields randomness-efficient \emph{polytime} implementations of \Cref{alg:rounding} when applied to the (maximal \algotype two-choice) $k$-level fractional matching algorithms that we design in \Cref{sec:klevelfrac} and \Cref{sec:weightedfractional}.
	\section{An Efficient Implementation of \Cref{alg:rounding}}\label{sec:efficient-implementation}

In this section we present an efficient implementation of \Cref{alg:rounding} when run with independent random variables.
Recall that \Cref{alg:rounding} assumes as input a maximal \algotype two-choice fractional input $\vec{x}$.

The only non-trivial part of an implementation of
\Cref{alg:rounding} is determining whether or not pairs $\{1,2\}$ are negative.
That is, we need to distinguish between $\Pr[F_{1,t},F_{2,t}]=0$ and $\Pr[F_{1,t},F_{2,t}] = (1-x^{(t)}_1)\cdot (1-x^{(t)}_2)$. This is trivial to check if $x^{(t)}_1=1$ or $x^{(t)}_2=1$, since Invariant \eqref{invariant:correlation} implies that any pair (and indeed, any set) $I$ containing a vertex $i$ with fractional degree $x^{(t)}_i=1$ is negative at time $t$, since then $\Pr\left[F_{I,t}\right] \in \{0,\prod_{i} (1-x^{(t)}_i)\} = \{0\}$. We therefore focus on pairs which are not trivially negative, as in the following definition.
\begin{Def}
	A set $I\subseteq [n]$ is \emph{strictly negative} if it is negative and $x^{(t)}_i\neq 1$ for all $i\in I$.
\end{Def}

Now, determining whether a pair $\{1,2\}$ is strictly negative can be easily implemented in exponential time, by considering the decision tree defined by the algorithm. A much more efficient implementation is possible, however, as we now show.

Recall that by \Cref{obs:negativity-subsets}, if a set $I$ contains a pair of nodes which are negative, then $I$ must itself be negative. The following lemma, which will prove useful in order to implement our algorithm efficiently, shows that the converse is also true for strictly negative sets and pairs. That is, any strictly negative set $I$ has a strictly negative ``witness'' consisting of a pair of nodes in $I$.

\begin{lem}\label{lem:pair-witness}
	A set of offline nodes $I$, $|I|\geq 2$ with $x^{(t)}_i\neq 1$ for all $i\in I$ is strictly negative if and only if it contains a pair $J\subseteq I$, $|J|=2$ which is itself strictly negative.
\end{lem}
\begin{proof}
	The ``if'' direction follows from \Cref{obs:negativity-subsets} and definition of strict negativity. We prove the ``only if'' direction for all sets $I$ by induction on $t$. The claim holds vacuously at time $t=1$, at which point there are no negative sets. For the inductive step, consider some such set $I$ with $x^{(t+1)}_i \neq 1$ for all $i\in I$ and time $t+1$.
	
	If $I$ was (strictly) negative by time $t$, then there exists a pair $J\subseteq I$, $|J|=2$ which is strictly negative by time $t$, and by \Cref{obs:negativity}, both $I$ and $J$ remain negative at time $t+1$. Therefore, $J$ is the desired strictly negative pair at time $t+1$ contained in $I$.
	
	Now, suppose $I$ was not strictly negative at time $t$, but it is at time $t+1$.
	Denote by $\{1,2\}$ the neighbors of $t$ with non-zero probability of being matched to $t$.
	(Note that these must indeed be a pair, since if $t$ can only be matched to at most one node $1$, this node must reach fractional degree $x_1+\Delta x_1 = \frac{x_1+x_2+1-x_1x_2}{2} = 1$, and so all sets that become negative at time $t$ are not strictly negative.)
	Inspecting the proof of \Cref{lem:invariants}, we find that either $I\supseteq \{1,2\}$, in which case $J=\{1,2\}$ is the desired pair, or (wlog) $I\cap\{1,2\} = \{1\}$, and we have that $\{1,2\}$ is independent by time $t$ and $I\cup \{2\}$ is (strictly) negative by time $t$. Then, by the inductive hypothesis, we have that $I\cup \{2\}$ contains a (strictly) negative pair $J'$ by time $t$. Since $I$ is not negative by time $t$, \Cref{obs:negativity-subsets} implies that the pair $J'$ cannot be a subset of $I$, and since $\{1,2\}$ is also not negative at time $t$, we know that $J'\neq \{1,2\}$, and so $J' = \{2,i\}$ for some $i\in I$.
	That is, we have that
	$\Pr[F_{i,t},F_{2,t}]=0$, and therefore $F_{i,t} \leq \overline{F_{2,t}}$. (In words, if $i$ is free, $2$ must be matched.)
	Consequently, we find that the pair $\{1,i\}\subseteq I$ becomes negative, since
	\begin{align*}
	\Pr[F_{1,t+1}, F_{i,t+1}] & \leq \Pr[F_{1,t+1}, F_{i,t}] \leq \Pr[F_{1,t+1}, \overline{F_{2,t}}] = 0,
	\end{align*}
	where the equality above relies on $\{1,2\}$ previously being independent, and so $\Pr[F_{1,t+1} \mid \overline{F_{2,t}}] = 0.$ We conclude that if $I$ satisfying $x^{(t+1)}_i \neq 1$ for all $i\in I$ is strictly negative at time $t+1$, then there exists some pair $J\subseteq I$ which is itself strictly negative at time $t+1$.
\end{proof}

For any offline node $i$ which has $x^{(t)}_i\neq 1$,
we denote the all offline nodes $j$ such that the pair $\{i,j\}$ is strictly negative by time $t$ by $$S^{(t)}_i := \{j \mid \{i,j\} \textrm{ are strictly negative by time }t \}.$$
The following lemmas characterize the changes to these sets from time $t$ to $t+1$, allowing for simple maintenance of these sets over time.

\begin{lem}\label{lem:negative-sets}
	If $t$ has a single neighbor $1$ with non-zero probability of being matched to $t$, then $S^{(t+1)}_1=\emptyset$ and $S^{(t+1)}_i = S^{(t)}_i\setminus \{1\}$ for all $i\neq 1$, while $S^{(t+1)}_1 = \emptyset$.
\end{lem}
\begin{proof}
	Follows from $1$ reaching fractional degree $x^{(t+1)}_1 = 1$ in this case, and therefore $1$ no longer belongs to any strictly negative set, while for all other nodes $i\neq 1$, we have that $F_{i,t} \equiv F_{i,t+1}$, and so all pairs $\{i,j\}\not\ni 1$ are strictly negative at time $t+1$ if and only if they are strictly negative at time $t$.
\end{proof}
\begin{lem}
	Let $t$ be an online node with non-zero probability of being matched to nodes in $\{1,2\}$.
	Then, we have
	\begin{align*}
	S^{(t+1)}_i =  \begin{cases}
	S^{(t)}_1 \cup S^{(t)}_2 \cup \{2\} & i=1\\
	S^{(t)}_2 \cup S^{(t)}_1 \cup \{1\} & i=2 \\			
	S^{(t)}_i \cup \{2\} & i\in S^{(t)}_1 \\			
	S^{(t)}_i \cup \{1\} & i\in S^{(t)}_2 \\
	S^{(t)}_i & i\not\in \{1,2\}\cup S^{(t)}_1 \cup S^{(t)}_2.
	\end{cases}
	\end{align*}
\end{lem}
\begin{proof}
	We note that for independently-maximal fractional algorithms no node reaches fractional degree $x^{(t+1)}_i=1$ at time $t$. Therefore, by \Cref{obs:negativity}, for all $i$, we have that $S^{(t+1)}_i \supseteq S^{(t)}_i$. We will show that our expression for $S^{(t+1)}_i\setminus S^{(t)}_i$ is precisely the set of all other nodes $j$ such that $\{i,j\}$ is strictly negative at time $t+1$ but not at time $t$.
	
	Consider a pair  $I=\{i,j\}$ which was not strictly negative at time $t$, but became strictly negative at time $t+1$. In particular, by monotonicity of $x_i$ over time, this implies $I$ must have been independent at time $t$.
	By the proof of \Cref{lem:invariants}, this implies that one of two cases must hold:
	\begin{enumerate}
		\item $I = \{1,2\}$.
		\item $1\in I$ and $2\not\in I$ (or vice versa) and $\{1,2\}$ is independent at time $t$.
	\end{enumerate}
	For the former case,
	this implies that $\Pr[F_{1,t+1},F_{2,t+1}] = 0$, and therefore $1\in S^{(t+1)}_2$ and $2\in S^{(t+1)}_1$.
	For the latter case, consider a node $i\in S^{(t)}_2$. That is, some node $i$ such that $\{i,2\}$ is negative, and so $\Pr[F_{i,t},F_{2,t}]=0$. This implies that $F_{i,t} \leq \overline{F_{2,t}}$. (In words, if $i$ is free, $2$ must be matched.)
	Consequently, we find that the pair $\{1,i\}\subseteq I$ becomes negative, since
	\begin{align*}
	\Pr[F_{1,t+1}, F_{i,t+1}] & \leq \Pr[F_{1,t+1}, F_{i,t}] \leq \Pr[F_{1,t+1}, \overline{F_{2,t}}] = 0,
	\end{align*}
	where the equality above relies on $\{1,2\}$ previously being independent, and so $\Pr[F_{1,t+1} \mid \overline{F_{2,t}}] = 0.$ 
	And indeed, we have that $S^{(t)}_2 \subseteq S^{(t+1)}_1$. (Symmetrically, we have that $S^{(t)}_1 \subseteq S^{(t+1)}_2$.)
	
	We conclude that our expression for $S^{(t+1)}_i$ is correct.
\end{proof}

The two preceding lemmata yield a simple linear-time algorithm for maintaining the negative pairs (by maintaining the sets $S^{(t)}_i$, in addition to the sets $x^{(t)}_i$), which by the preceding discussion yields an efficient implementation of our algorithm.
\implementation*

	\section{Application: Optimal Semi-OCS}\label{sec:OCS}

We recall the definition of $\gamma$-semi-OCS.
\begin{Def}[\cite{fahrbach2020edge}]
	A $\gamma$-semi-OCS is an algorithm which, given pairs of items in an online manner, picks one item per pair upon arrival, such that for each item $i$ appearing in $k$ pairs,
	$$\Pr[i \textrm{ never picked}] = 2^{-k}(1-\gamma)^{-(k-1)}.$$
\end{Def}

A stronger guarantee is given by applying our online rounding scheme of \Cref{alg:pairwise-rounding} to the following fractional matching algorithm, which assumes each online node neighbors precisely two offline neighbors. (These correspond to items in the definition of OCS).
Here $\ell_i$ is the ``level'' of offline node $i$, which corresponds to the number of pairs which $i$ belonged to so far (inclusive).

\begin{algorithm}[h]
	\caption{Semi-OCS-\underline{Inducing} Fractional Algorithm}
	\label{alg:OCS-fractional}
	\begin{algorithmic}[1]
		\State initially, set $\vec{x}\leftarrow \vec{0}$
		\State initially, set $\vec{\ell} \leftarrow \vec{0}$ \Comment{$\ell_i$ is the number of pairs $i$ belonged to so far}
		\For{arrival of online node $t$}
		\State for $i=1,2$ the two neighbors of online node $t$, increase $\ell_i$ \Comment{In particular, $\ell_i>0$}
		%\If{$t$ has less than two neighbors}
		%\State add two dummy neighbors $i$ with $x^{(t)}_i = 1$
		%\Comment{used to simplify notation}
		%\EndIf
%		\State let $0\leq x^{(t)}_1\leq x^{(t)}_2 \leq 1$ be the fractional degrees of the two neighbors of $t$	
		\State set $x_{i,t}\gets (1-2^{-2^{\ell_i}+1}) - (1-2^{-2^{\ell_i-1}+1}) = 2^{-2^{\ell_i-1}+1} - 2^{-2^{\ell_i}+1}$ \Comment{$x_{i,t} > 0$}
		\EndFor
	\end{algorithmic}
\end{algorithm}

First, we show that our rounding scheme can be applied to this fractional matching algorithm.
%
%The above algorithm indeed maintains a fractional matching.
\begin{obs}\label{obs:indeed-OCS}
	\Cref{alg:OCS-fractional} is a \algotype two-choice fractional matching algorithm.
\end{obs}
\begin{proof}
	A simple proof by induction implies that the fractional degree of every item $i$ belonging to $\ell$ pairs by time $t$ is $x^{(t)}_i=1-2^{-2^{a}+1}$.
	On the other hand, at any time $t$, if $1$ and $2$ are the two neighbors of $t$, with previous levels $\ell_1$ and $\ell_2$ respectively (so, their new levels are $\ell_1+1$ and $\ell_2+1$), then we have the desired inequality \eqref{eqn:natural-condition-for-rounding}, implying that this two-choice algorithm is indeed \algotype.
	\begin{align*}
		x_{1,t} + x_{2,t} & =
%		
%		f(\ell_1) - f(\ell_1 - 1) + f(\ell_2) - f(\ell_2 - 1) \\
%	& = (1-f(\ell_1-1)) - (1-f(\ell_1)) + (1-f(\ell_2-1)) - (1-f(\ell_2)) \\
	2^{-2^{\ell_1}+1} - 2^{-2^{\ell_1+1}+1} + 2^{-2^{\ell_2}+1} - 2^{-2^{\ell_2+1}+1} \\
	& = 2^{-2^{\ell_1}+1} + 2^{-2^{\ell_2}+1} - \frac{(2^{-2^{\ell_1}+1})^2}{2} - \frac{(2^{-2^{\ell_2}+1})^2}{2} \\
	& \leq 2^{-2^{\ell_1}+1} + 2^{-2^{\ell_2}+1} - \sqrt{(2^{-2^{\ell_1}+1})^2\cdot (2^{-2^{\ell_2}+1})^2} & \textrm{AM-GM}
	\\
	& = 2^{-2^{\ell_1}+1} + 2^{-2^{\ell_2}+1} - 2^{-2^{\ell_1}+1}\cdot 2^{-2^{\ell_2}+1} \\
	& =  1 - (1-2^{-2^{\ell_1}+1})\cdot (1-2^{-2^{\ell_2}+1}). & &\qedhere
	\end{align*}
\end{proof}

By \Cref{thm:rounding-intro}, the output fractional matching of \Cref{alg:OCS-fractional} can be rounded losslessly online. So, by this fractional matching's closed-form solution, we obtain the following optimal semi-OCS.

\begin{cor}
	There exists a $\frac{1}{2}$-semi-OCS. Moreover, for any element $i$ appearing in $k$ pairs, this semi-OCS satisfies
	$$\Pr[i \textrm{ never picked}] = 2^{-2^k+1}.$$
	(Both bounds are optimal, by \cite{gao2021improved}.)
\end{cor}

\begin{remark}
	The above algorithm satisfies the required inequalities of semi-OCS at equality, which may prove useful for various notions of fairness. If such fairness considerations are secondary for the application at hand, then the above algorithm can easily be extended to be maximal, resulting in the strong negative correlation property \eqref{invariant:correlation}.
\end{remark}

The advantages of our design and analysis of this semi-OCS over that of \cite{gao2021improved} are twofold: first, our analysis provides sharper negative concentration properties: the output randomized matching of \Cref{alg:pairwise-rounding} satisfies the strong FKG-like negative correlation property \eqref{invariant:FKG}, or even independence of incompatibility ($\Pr[F_{I,t}]\in \{0,\prod_{i\in I} \Pr[F_{i,t}]$), while \cite{gao2021improved} proved a weaker dependence property for this semi-OCS: $\Pr[F_{A\cup B,t}]\leq \Pr[F_{A,t}]\cdot \Pr[F_{B,t}]$; it is not hard to see that our former property implies the latter, and is sometimes strictly stronger.
A second advantage of our approach is that the design of this OCS follows from a general approach: rounding. This hints at more algorithms for explicitly negatively correlating choices which can be derived in this principled manner.

\subsection{Bichromatic Semi-OCS}
To emphasize the generality of our rounding-based OCS, we explore extensions of this algorithmic primitive, and show how to achieve optimal algorithms for these extensions directly via online rounding of \algotype two-choice fractional matching algorithms.

By \cite{gao2021improved}, we cannot guarantee a uniform selection probability for all items higher than $1-2^{-2^k+1}$.
We can, however, increase the probabilities for \emph{some} of the items in some settings, as we now show.

Suppose each item is colored either red and blue. Moreover, suppose we are guaranteed that in every pair contains one red item and one blue item.
Then, we can guarantee a higher selection probability for, say, blue items at the expense of red items, as follows.

\begin{Def}
	A bichromatic-$(f,g)$-semi-OCS receives one-by-one pairs of items, one red and one blue, and must select one item per pair immediately and irrevocably. Moreover, it must guarantee each red (blue) item appearing in $k$ pairs a probability of $f(k)$ ($g(k)$) of being selected at least once.
\end{Def}

\begin{lem}
	Let $a \geq 1$ and $\bar{a}\geq 1$ be such that $\frac{1}{a}+\frac{1}{\bar{a}}=1$. Then, there exits a bichromatic-$(1-2^{-a^k+1},1-2^{-\bar{a}^k+1})$-semi-OCS.
\end{lem}
\begin{proof}
	Such a bichromatic semi-OCS is obtained by applying the lossless online rounding \Cref{alg:pairwise-rounding} to the natural generalization of \Cref{alg:OCS-fractional}, where the invariant that we maintain is that for any red (resp., blue) item $i$ appearing in $k$ pairs before time $t$, we have $x^{(t)}_i = (1-2^{-a^k+1})$ (resp., $x^{(t)}_i = (1-2^{-\bar{a}^k+1})$). The proof that the above fractional matching algorithm is a \algotype two-choice fractional matching algorithm (and we can thus apply \Cref{alg:pairwise-rounding} to this algorithm) generalizes the proof of \Cref{obs:indeed-OCS}. The only difference is in the application of the \emph{weighted} AM-GM inequality, which implies the following inequality. If at time $t$ the pair contains red item $1$ and blue item $2$, both appearing previously to time $t$ in $\ell_1$ and $\ell_2$ pairs, then
	\begin{align*}
	x_{1,t} + x_{2,t} & =
	2^{-a^{\ell_1}+1} - 2^{-a^{\ell_1+1}+1} + 2^{-\bar{a}^{\ell_2}+1} - 2^{-\bar{a}^{\ell_2+1}+1} \\
	& =
	2^{-a^{\ell_1}+1} + 2^{-\bar{a}^{\ell_2}+1} - \frac{(2^{-a^{\ell_1}+1})^a}{a} - \frac{(2^{-\bar{a}^{\ell_2}+1})^{\bar{a}}}{\bar{a}} \\
	& \leq 2^{-a^{\ell_1}+1} + 2^{-\bar{a}^{\ell_2}+1} - \sqrt[a]{(2^{-a^{\ell_1}+1})^a}\cdot
	\sqrt[\bar{a}]{(2^{-\bar{a}^{\ell_2}+1})^{\bar{a}}}
	 & \textrm{weighted AM-GM} \\
	& =	2^{-a^{\ell_1}+1} + 2^{-\bar{a}^{\ell_2}+1} - 2^{-a^{\ell_1}+1}\cdot 2^{-\bar{a}^{\ell_2}+1}
	\\
	& =  1 - (1-2^{-{a}^{\ell_1}+1})\cdot (1-2^{-\bar{a}^{\ell_2}+1}).
	\end{align*}
	That is, this two-choice algorithm is indeed \algotype. Consequently, we can apply \Cref{alg:pairwise-rounding} to this fractional algorithm and obtain the desired marginal probabilities, and the implied bichromatic semi-OCS.
\end{proof}

		\section{Beyond Two Choices: Challenges}\label{sec:multi-appendix}
	
	In this section we discuss challenges in generalizing our characterization of online roundable fractional matchings beyond two-choice algorithms, to multiple-choice algorithms.
	A natural extension of Condition \eqref{rounding-weaker-condition} to multiple-choice algorithms, for which $P_t := \{i \mid x_{i,t}>0\}$ need not satisfy $|P_t|\leq 2$, is the following.
	\begin{equation}\label{eqn:natural-condition-for-rounding}
	\sum_{i\in I} x_{i,t} \leq 1-\prod_{i\in I} x^{(t)}_i \qquad \forall I\subseteq P_t.
	\end{equation}
	
	There are two natural challenges to achieving improved guarantees this way: the first is that it is unclear whether such additional constraints (which are not satisfied by prior fractional algorithms) are compatible with high competitive ratios. In \Cref{sec:multi-choice} we show that these constraints are compatible with an optimal competitive ratio of $1-\nicefrac{1}{e}$. The second challenge is finding such additional constraints which allow for lossless online rounding. Unfortunately, in \Cref{sec:generalization-challenges}, we show that the above condition is not sufficient for lossless online rounding.

	\subsection{A Multi-Choice Fractional Algorithm}\label{sec:multi-choice}
	In this section we present a $(1-\nicefrac{1}{e})$-competitive online fractional matching algorithm that satisfies Condition \eqref{eqn:natural-condition-for-rounding}.

	Our algorithm's approach will be to perform (restricted) water-filling; we start with a water level $\ell$ of zero and increase this water level continuously, increasing the fractional degree -- and thus $x_{i,t}$ values -- of all nodes of current fractional degree $x^{(t)}_i+x_{i,t}$ less than $\ell$. We do so until a constraint in \Cref{eqn:natural-condition-for-rounding} is met at equality. 
	Note that since $\sum_{i\in N(t)} x_{i,t}\leq 1-\prod_{i\in I} x^{(t)}_i \leq 1$, this is a feasible fractional matching.
	Now, finding the final water level might seem to require time exponential in $|N(t)|$, given the above phrasing. However, as we show in \Cref{prefix-sufficient-for-marginal-rounding}, for such a water-filling algorithm, it is enough to guarantee that \Cref{eqn:natural-condition-for-rounding} holds for all subsets $I$ containing $|I|$ nodes of lowest fractional degree. This allows to compute $\ell$ in time linear in $|N(t)|$. The pseudocode for our algorithm is given in \Cref{alg:fractional-roundable-WF}.
	
	\begin{algorithm}[h]
		\caption{Multi-Choice Fractional Algorithm}
		\label{alg:fractional-roundable-WF}
		\begin{algorithmic}[1]
			\For{arrival of online node $t$}
%			\For{all neighbors $i\in N(t)$}
%			\State let $x^{(t)}_i = \sum_{t'<t} x_{i,t}$
%			\EndFor	
			\State let $0\leq x^{(t)}_1\leq x^{(t)}_2\leq \dots \leq x^{(t)}_k \leq 1$ be the fractional degrees of neighbors of $t$	
			\For{$k=1,2,\dots,|N(t)|$}
			\State let $I_k := [k]$
			\EndFor 
			\State set $\ell \leftarrow \max\left\{\ell \,\bigg\vert\, \sum_{i\in I_k}\left(\ell - x^{(t)}_i\right)^+ \leq 1-\prod_{i\in I_k} x^{(t)}_i, \quad \forall k\in [|N(t)|]\right\}$
			\For{all neighbors $i\in N(t)$} \label{choice-of-ell}
			\State set $x_{i,t} \leftarrow \left(\ell-x^{(t)}_i\right)^+$
			\EndFor
			\EndFor
		\end{algorithmic}
	\end{algorithm}
	
	\begin{lem}\label{prefix-sufficient-for-marginal-rounding}
		\Cref{alg:fractional-roundable-WF} satisfies \Cref{eqn:natural-condition-for-rounding}. 
	\end{lem}
	\begin{proof}
		Fix some online node $t$, and let $x_i$ denote $x^{(t)}_i$.
		\Cref{alg:fractional-roundable-WF} explicitly satisfies \Cref{eqn:natural-condition-for-rounding} for all subsets $I_k\in N(t)$, $k\in [|N(t)|]$. That is, if we relabel the neighbors of $t$ as $1,2,\dots,|N(t)|$ in increasing order of $x_i$ value, we have that for all $k\in [|N(t)|]$
		\begin{equation}\label{prefixes}
		\sum_{i=1}^k(\ell-x_i)^+ \leq 1-\prod_{i=1}^k x_i.
		\end{equation}
		Now, consider some set $I\subseteq N(t)$ of  neighbors of $t$. 
		We wish to show that 
		\begin{equation}\label{prefixes-general}
		\sum_{i\in I}(\ell-x_i)^+ \leq 1-\prod_{i\in I} x_i.
		\end{equation}
		Let $k\leq |I|$ be the number of non-zero summands in the LHS of \Cref{prefixes-general}. That is, the number of $i\in I$ such that $x_i < \ell$. If we denote by $I' \triangleq \{i\in I \mid x_i < \ell\}$ the neighbors in $I$ contributing to this LHS, we find that to prove \Cref{prefixes-general} it is sufficient to prove that
		\begin{equation}\label{prefixes-nonzero-general}
		\sum_{i\in I'}(\ell-x_i) = \sum_{i\in I'}(\ell-x_i)^+ \leq 1-\prod_{i\in I'} x_i,
		\end{equation}    
		since $1-\prod_{i\in I'} x_i \leq 1-\prod_{i\in I} x_i$, as $x_i\in [0,1]$ for all $i$. As $I$ contains $k$ neighbors $i$ such that $x_i < \ell$, we have that $x_i < \ell$ for all $i\in I_k$. Consequently we have that 
		\begin{equation}\label{prefixes-nonzero}
		\sum_{i\in I_k}(\ell-x_i) = \sum_{i\in I_k}(\ell-x_i)^+ \leq 1-\prod_{i\in I_k} x_i.
		\end{equation}    
		We will show that \Cref{prefixes-nonzero} implies \Cref{prefixes-nonzero-general}, which in turn implies \Cref{prefixes-general}.
		
		To this end, define the function $g:\mathbb{R}^k\rightarrow \mathbb{R}$ to be $g(y_1,y_2,\dots,y_k) \triangleq \prod_{i=1}^k y_i + \sum_{i=1}^k \left(\ell-y_i\right) - 1$.
		\Cref{prefixes-nonzero} is equivalent to $g(x_1,x_2,\dots,x_k)\leq 0$, and similarly, to prove \Cref{prefixes-nonzero-general} we wish to prove the equivalent condition, $g(x_{i_1},x_{i_2},\dots,x_{i_k})\leq 0$ for $I' = \{i_1,i_2,\dots,i_k\}\subseteq N(t)$.
		But indeed, this follows from $g(x_1,x_2,\dots,x_k)\leq 0$ and the partial derivatives of $g(\vec{x})$, all of the form $\frac{d}{dx_i} g(\vec{x})=\prod_{i\in [k]\setminus\{j\}} x_i - 1$, being non-positive for all $\vec{x}\in [0,1]^n$. In particular, assuming without loss of generality that $x_{i_1}\leq x_{i_2}\leq \dots \leq x_{i_k}$, we have that $x_j \leq x_{i_j}$ for all $j\in [k]$, and so
		\[
		g(x_{i_1},x_{i_2},\dots,x_{i_k}) \leq g(x_{1},x_{i_2},\dots,x_{i_k})\leq g(x_{1},x_{2},\dots,x_{i_k})\leq \dots \leq 
		g(x_{1},x_{2},\dots,x_{k})\leq 0. \qedhere
		\]
	\end{proof}
	
	\subsubsection{Analysis of Competitive Ratio}
	In this section we analyze the competitive ratio of \Cref{alg:fractional-roundable-WF}. In particular, we prove the following.
	
	\begin{thm}\label{roundable-WF-comp-ratio}
		\Cref{alg:fractional-roundable-WF} is $(1-\nicefrac{1}{e})$-competitive.
	\end{thm}
	To prove \Cref{roundable-WF-comp-ratio}, we will follow the online primal-dual method \cite{buchbinder2009design}. In particular, we construct a feasible dual solution, that is, a fractional vertex cover, such that the increase in primal value after each arrival $t$, namely $\sum_{i\in N(t)} x_{i,t}$, is at least $1-\nicefrac{1}{e}$ times the increase in the value of the dual solution. 
	Summing over all arrivals and relying on weak LP duality, this implies that, for $(P)$ and $(D)$ the values of the primal and dual solutions' values,
	$$(P) \geq \left(1-\nicefrac{1}{e}\right)\cdot (D) \geq \left(1-\nicefrac{1}{e}\right)\cdot OPT.$$
	Our fractional vertex cover will have a particularly simple form. For each offline vertex $i$ with fractional degree $x_i$, we let its dual value be $y_i = y(x_i)$, where $y:[0,1]\rightarrow [0,1]$ is some monotone increasing function, to be defined shortly.
	For online node $t$, if $\ell$ is the water level at time $t$, we let $y_t = 1-y(\ell)$. We say this solution is \emph{induced} by $y$.
	\begin{obs}
		A dual induced by a monotone increasing function $y:[0,1]\rightarrow [0,1]$ is feasible.
	\end{obs}
	\begin{proof}
		Consider an edge $(i,t)$. After arrival of $t$, we have that $x_i \geq \ell$. Consequently, since $y(\cdot)$ is monotone increasing and since dual values never decrease, we have that at the end of the algorithm's run, $y_i + y_t \geq y(\ell) + 1 - y(\ell) = 1$.
	\end{proof}
	
	Our choice of monotone increasing function $y$ which will induce our dual solution is 
	$$y(x) := \frac{e^x-1}{e-1}.$$
	Our analysis via the primal-dual method, and indeed this precise choice of dual values, is used in the analysis of the standard Water-Filling algorithm or the RANKING algorithm, giving short and elegant direct proofs of these algorithms' competitive ratio.
	For our algorithm, the analysis becomes significantly more involved due to the non-linear (and indeed, non-convex) constraints, given by \Cref{eqn:natural-condition-for-rounding}.
	
	We show that for any vector $\vec{x}\in \mathbb{R}^k$ for any $k\in \mathbb{Z}$, we have that the primal gain is at least $(1-\nicefrac{1}{e})$ times the dual cost for any arrival of an online node $t$ with neighbors' loads upon arrival equal to $\vec{x}$. For simplicity, we will denote the neighbors of $t$ by $1,2,\dots,|N(t)|$, such that $x_1\leq x_2\leq \dots x_{|N(t)|}$. Thus, we wish to show that for $\ell = \ell(\vec{x})$ chosen in \Cref{choice-of-ell} at time $t$, the following holds.
	
	\begin{equation}\label{primal-vs-dual-roundable-WF}
	\sum_{i=1}^{|N(t)|} \left(\ell - x_i\right)^+ \geq \left(1-\frac{1}{e}\right)\cdot \left(1-y(\ell) + \sum_{i=1}^{|N(t)|}\left(y(\ell) - y(x_i)\right)^+\right).
	\end{equation}
	
	We start by showing that \Cref{primal-vs-dual-roundable-WF} holds if all neighbors of $t$ have the same load prior to the arrival of $t$. That is, we show that this inequality holds for $\vec{x}\in \mathbb{R}^k$ proportional to the all-ones vector, say $\vec{x} = x\cdot \vec{1}$ (here $k=|N(t)|$). In our proof we will rely on the simple observation that for this case,  $\ell(\vec{x}) = x+(1-x^k)/k$.
	
	\begin{lem}\label{primal-dual-roundable-WF-uniform-proof}
		For any $k\in \mathbb{N}$, and $\vec{x} = (x,x,\dots,x)\in \mathbb{R}^k$ \Cref{primal-vs-dual-roundable-WF} holds with $\ell = x+(1-x^k)/k$.
	\end{lem}
	\begin{proof}
		Fix $k$. For a vector $\vec{x}$ as above, \Cref{primal-vs-dual-roundable-WF} simplifies to 
		\begin{align*}
		1-x^k & \geq \left(1-\nicefrac{1}{e}\right)\cdot \left(1-y(x + (1-x^k)/k) + k\cdot \left(y(x+(1-x^k)/k) - y(x)\right)\right),
		\end{align*}
		which can be rewritten as
		\begin{align*}
		\left(\frac{e}{e-1}\right)\cdot(1-x^k) 
		& \geq  \left(1-\frac{e^{x+(1-x^k)/k}-1}{e-1}\right) + k\cdot \left(\frac{e^{x+(1-x^k)/k} - e^{x}}{e-1}\right)
		\end{align*}
		Simplifying this expression further, we want to show that the following function $f_k(x)$ is non-positive for all $x\in [0,1]$. That is, we wish to show that for all $x\in [0,1]$, 
		\begin{align*}
		f_k(x) \triangleq (k-1)e^{(x+\frac{1-x^k}{k})}- k \cdot e^x + e \cdot x^k \leq 0.
		\end{align*}
		Normalizing by $e^x(\geq 0)$, this yields the function $g_k(x) \triangleq f_x(x)/e^x$, which we will show is non-positive for all $x\in [0,1]$, implying the same for $f_k(x)$. That is, we will show that for all $x\in [0,1]$, 
		\begin{equation}\label{intermediate-normalized-inequality}
		g_k(x) = (k-1)e^{\frac{1-x^k}{k}}- k + e^{1-x} \cdot x^k \leq 0.
		\end{equation}
		First, we note that $g_k(1)=0.$ We next show that $g_k(x)$ is increasing in the range $[0,1]$, which together with $g_k(1)=0$ implies \Cref{intermediate-normalized-inequality}, and consequently, the lemma. But indeed, 
		\begin{align*}
		g'_k(x) = -(k-1)x^{k-1}e^{\frac{1-x^k}{k}} - e^{1-x}x^k + k\cdot e^{1-x}x^{k-1} \geq 0
		\end{align*}
		in the range $x\in [0,1]$, since this inequality holds if and only if 
		\begin{align*}
		(k-x)\cdot e^{1-x} \geq (k-1)\cdot e^{\frac{1-x^k}{k}},
		\end{align*}
		which holds for $x\in[0,1]$, since for such $x$ we have both that $(k-x)\geq (k-1)$ and that $(1-x)\geq (1-x)\cdot (1+x+x^2+\dots +x^{k-1})/k = \frac{1-x^k}{k}$.
	\end{proof}
	
	As we will show,  \Cref{primal-vs-dual-roundable-WF} holding for the restricted uniform case implies the same for \emph{all} vectors $\vec{x}$. To show this, we will rely on the following observation regarding $\ell$.
	
	\begin{obs}\label{mean-inequality}
		Suppose that $x_1,x_2,\dots,x_k \leq \ell$ satisfy $\sum_{i=1}^k(\ell - x_i) = 1-\prod_{i=1}^k x_i$. Then the solution $x$ to $k\ell - kx = 1-x^k$ is at least the arithmetic mean of $\vec{x}$. That is,
		$$x\geq \left(\sum_{i=1}^k x_i\right)/k.$$
	\end{obs}
	\begin{proof}
		Let $x':=(\sum_{i=1}^k x_i)/k$. By the AM-GM inequality and the definition of $\ell$, we have that
		\begin{align*} 
		1-(x')^k & \leq  %1-\left(\sqrt[k]{\prod_{i=1}^k x_i}\right)^k & (AM-GM) \\ =
		1-\prod_{i=1}^k x_i = k\cdot \ell - \sum_{i=1}^k x_i = k\cdot \ell - k\cdot x'.
		\end{align*}
		That is, $\ell \geq x' + (1-(x')^k)/k = \ell(x'\cdot \vec{1}_k)$.
		As $\frac{d}{dz} \ell(z\cdot \vec{1}_k) = \frac{d}{dz}\left(z + (1-z^k)/k\right) = 1 - z^{k-1} \geq 0$ for all $z\leq 1$, the univariate function $f(z) = \ell(z\cdot\vec{1}_k)$ is monotone increasing in $z$, and so we have that $\ell(x\cdot \vec{1}_k)\geq \ell(x' \cdot \vec{1}_k)$ implies the claimed inequality, namely that $x \geq x'$. 
	\end{proof}

	The next lemma will prove instrumental in proving our algorithm's competitiveness.
	\begin{lem}\label{any-prefix-suffices}
		If $x_1\leq x_2\leq \dots\leq x_{|N(t)|}$ and for some $k\leq |N(t)|$ it holds that $\ell\geq x_k$ and 
		\begin{equation*}
		\sum_{i=1}^{k} \left(\ell - x_i\right)^+ \geq \left(1-\frac{1}{e}\right)\cdot \left(1-y(\ell) + \sum_{i=1}^{k}\left(y(\ell) - y(x_i)\right)^+\right),
		\end{equation*}    
		then \Cref{primal-vs-dual-roundable-WF} holds. That is,
		$$\sum_{i=1}^{|N(t)|} \left(\ell - x_i\right)^+ \geq \left(1-\frac{1}{e}\right)\cdot \left(1-y(\ell) + \sum_{i=1}^{|N(t)|}\left(y(\ell) - y(x_i)\right)^+\right).$$
	\end{lem}
	\begin{proof}
		Denote by $k'\leq |N(t)|$ the largest index such that $x_{k'}\leq \ell$. 
		Then, we have that
		\begin{align*}
		\sum_{i=1}^{|N(t)|} \left(\ell - x_i\right)^+ & = 
		\sum_{i=1}^{k'} \left(\ell - x_i\right) \\
		& = \sum_{i=1}^{k} \left(\ell - x_i\right) + \sum_{i=k+1}^{k'} \left(\ell - x_i\right) \\
		& \geq \left(1-\frac{1}{e}\right)\cdot \left(1-y(\ell) + \sum_{i=1}^k (y(\ell) - y(x_i))\right) + \sum_{i=k+1}^{k'} \left(\ell - x_i\right) \\
		& \geq \left(1-\frac{1}{e}\right)\cdot \left(1-y(\ell) + \sum_{i=1}^k (y(\ell) - y(x_i))\right) + \left(1-\frac{1}{e}\right)\cdot \left(\sum_{i=k+1}^{k'} (y(\ell) - y(x_i))\right) \\
		& = \left(1-\frac{1}{e}\right)\cdot \left(1-y(\ell) + \sum_{i=1}^{k'}\left(y(\ell) - y(x_i)\right)\right) \\
		& = \left(1-\frac{1}{e}\right)\cdot \left(1-y(\ell) + \sum_{i=1}^{|N(t)|}\left(y(\ell) - y(x_i)\right)^+\right),
		\end{align*}
		where the last inequality follows from $\ell-x_i \geq \left(1-\frac{1}{e}\right)\cdot (y(\ell) - y(x_i)) = e^{\ell-1} - e^{x_i-1}$, which in turn follows from $e^{x-1}$ growing slower than $x$ in the domain $[0,1]$, where $\frac{d}{dx} e^{x-1} = e^{x-1}\leq 1 = \frac{d}{dx}x$.
	\end{proof}

	We now prove that \Cref{primal-vs-dual-roundable-WF} holds for all vectors $\vec{x}$.
	
	\begin{lem}\label{primal-dual-roundable-WF-general-proof}
		\Cref{primal-vs-dual-roundable-WF} holds for all vectors $\vec{x} \in \mathbb{R}^n$ with $\ell(\vec{x})$ as defined in \Cref{choice-of-ell}.
	\end{lem}
	\begin{proof}
		Let $k\leq n$ be the lowest index such $\sum_{i=1}^k (\ell - x_i) = 1-\prod_{i=1}^k x_i$. (Note that such an index must exist, by our choice of $\ell$.) 
		Then, we have in particular that $\ell \geq x_i$ for all $i\in [k]$, as the converse would imply $\sum_{i=1}^k (\ell- x_i) < \sum_{i=1}^k (\ell- x_i)^+ \leq 1-\prod_{i=1}^k x_i$.
		Let $x\in \mathbb{R}$ be the solution to $k\cdot (\ell - x) = 1-x^k$. 
		We rely on the definition of $x$ and $\ell$ to prove that \Cref{primal-vs-dual-roundable-WF} holding for $x\cdot \vec{1}_k$ implies the same inequality for $\vec{x}' = (x_1,x_2,\dots,x_k)$, which by \Cref{any-prefix-suffices} implies that \Cref{primal-vs-dual-roundable-WF} holds for $\vec{x}$. (Note that since $\sum_{i=1}^k (\ell - x_i) = 1-\prod_{i=1}^k x_i$, we have that $\ell(\vec{x}')=\ell(\vec{x})=\ell$.)
		
		First, we note that  \Cref{primal-vs-dual-roundable-WF} for $x\cdot \vec{1}_k$ is equivalent to 
		\begin{equation}\label{pd-roundable-WF-equiv-uniform}
		k\cdot \ell - k\cdot x + \left(1-\frac{1}{e}\right)\cdot k\cdot y(x) \geq \left(1-\frac{1}{e}\right)\cdot \left(1 + (k-1)\cdot y(\ell)\right).
		\end{equation}
		Similarly,  \Cref{primal-vs-dual-roundable-WF} for $\vec{x}'$ is equivalent to  
		\begin{equation}\label{pd-roundable-WF-equiv-general}
		k\cdot \ell - \sum_{i=1}^{k} x_i + \left(1-\frac{1}{e}\right)\cdot \sum_{i=1}^{k} y(x_i) \geq \left(1-\frac{1}{e}\right)\cdot \left(1 + (k-1)\cdot y(\ell)\right).
		\end{equation}
		
		To prove \Cref{pd-roundable-WF-equiv-general} we will prove that its LHS is greater than the LHS of \Cref{pd-roundable-WF-equiv-uniform}. 
		Equivalently, we will show that
		\begin{equation}\label{LHS-minus-LHS}
		\left(1-\frac{1}{e}\right)\cdot \sum_{i=1}^k y(x_i) - \sum_{i=1}^k x_i \geq
		\left(1-\frac{1}{e}\right)\cdot k \cdot y(x) - k\cdot x.
		\end{equation}
		
		Let $g(z) \triangleq \left(1-\frac{1}{e}\right)y(z)  - z = e^{z-1} - \frac{1}{e} - z$, and $X$ be a uniformly-random number in $\{x_1,x_2,\dots,x_k\}$. Then dividing both sides of \Cref{LHS-minus-LHS} by $k$ we find that this equation is equivalent to 
		\begin{align*}
		\mathbb{E}[g(X)] \geq g(x).
		\end{align*}
		To prove the above, we will prove the two following inequalities.
		\begin{align*}
		\mathbb{E}[g(X)] \geq g(\mathbb{E}[X]) \geq g(x).
		\end{align*}
		The first inequality follows from Jensen's Inequality and convexity of $g(z)$ in the domain $[0,1]$, where $g''(z) = e^{z-1}\geq 0$. The second inequality follows from $g(z)$ being non-increasing in the domain $[0,1]$, where $g'(z) = e^{z-1} - 1 \leq 0$, and \Cref{mean-inequality} implying
		$x \geq (\sum_{i=1}^k x_i)/k = \mathbb{E}[X]$.
	\end{proof}
	
	\Cref{roundable-WF-comp-ratio} follows from \Cref{primal-dual-roundable-WF-general-proof} and the preceding discussion.
%	\color{black}
	
	\subsection{Challenges: Stronger Constraints Needed}\label{sec:generalization-challenges}

	Similarly to our discussion for two-choice algorithms, condition \eqref{eqn:natural-condition-for-rounding} can be shown to be necessary for lossless online rounding in some scenarios. 
	Perhaps surprisingly, we show that unlike for two-choice algorithms, this natural generalization of Condition \eqref{rounding-weaker-condition} to multiple-choice algorithms is \emph{not sufficient} for lossless online rounding.
	
		\begin{restatable}{lem}{generalconditionneeded}
			There exists a three-choice fractional matching algorithm $\calA_f$ whose output $\vec{x}$ satisfies Condition \eqref{eqn:natural-condition-for-rounding}, such that for any randomized online matching algorithm $\calA$, there exists a graph on which the fractional matching of $\calA_f$ has value strictly greater than the expected matching size of $\calA$. That is, $\calA_f$ is not losslessly roundable.
		\end{restatable}

	\begin{proof}
		Our proof goes via Yao's Lemma \cite{yao1977lemma}. We consider the following distributions over graphs of maximum degree three, on which any fractional or randomized algorithm are trivially three-choice algorithms.
		We label the offline nodes $i_{1,1},i_{1,2},i_{2,1},i_{2,2},i_{3,1},i_{3,2}$.
		Each online node $t\in [3]$ neighbors the two offline nodes $i_{t,1},i_{t,2}$.
		In addition, we have two online nodes, $4$ and $5$, that neighbor a random node in each of the pairs $\{i_{t,1},i_{t,2}\}_{t\in [3]}$.
		The three-choice fractional matching we consider is the following:
		\[
		x_{i,t} = \begin{cases}
		0 & (i,t)\not\in E \\
		\frac{1}{2} & (i,t)\in E, i=i_{t,j}, t\in [3] \\
		\frac{7}{24} = \frac{1-\left(\frac{1}{2}\right)^3}{3}  & (i,t)\in E, t=4 \\
		\frac{6965}{41472} = \frac{1-\left(\frac{19}{24}\right)^3}{3} & (i,t)\in E, t=5.
		\end{cases}
		\]
		We note that $\vec{x}$ is a three-choice fractional matching satisfying Condition \eqref{eqn:natural-condition-for-rounding}. We note moreover that $x_{i,t} > \frac{1}{6}$ for $(i,t)\in E$ and $t=5$.
		Consequently, this fractional matching has value at least $\sum_{i,t} x_{i,t} > 3 + \frac{7}{8} + \frac{1}{2}$.
		We now proceed to show that every randomized algorithm $\calA$ outputs a matching $\calM$ of expected size $\E[|\calM|] = 3 + \frac{7}{8} + \frac{1}{2}$. That is, $\E[|\calM|] < \sum_{i,t} x_{i,t}$.
		
		Consider a deterministic algorithm $\calA'$ run on an input drawn from the above distribution. By simple exchange arguments, due to the one-sided vertex arrivals, we may safely assume that $\calA'$ is greedy, and matches whenever presented with an online node with at least one free neighbor \cite{karp1990optimal}.
		Therefore, precisely one node in each pair $\{i_{t,1},i_{t,2}\}$ is matched to online node $t$, and therefore the number of free nodes in the neighborhood of $4$ and $5$ before time $4$ is distributed $Y\sim \textrm{Bin}(3,1/2)$. Consequently, the expected number of online nodes matched among $4$ and $5$ is precisely $\E[\max\{2,Y\}] = 1\cdot \Pr[Y=1] + 2\cdot \Pr[Y\geq 2] = \frac{3}{8} + 1 = \frac{7}{8}+\frac{1}{2}$.
		We conclude that the matching $\calM$ output by the deterministic algorithm $\calA'$ on the above distribution has expected size
		strictly less than the value of the fractional matching output by $\calA_f$, namely
		$$\E[|\calM|] = 3+\frac{7}{8}+\frac{1}{2} < \sum_{i,t} x_{i,t}.$$
		Therefore, by Yao's Lemma, for each randomized algorithm $\calA$, one of the graphs in the support of the above distribution results in $\calA$ outputting a matching $\calM$ whose expected size is strictly smaller than the fractional matching satisfying Condition \eqref{eqn:natural-condition-for-rounding} output by the three-choice algorithm $\calA_f$.
	\end{proof}
	
	\begin{cor}
		Condition \eqref{eqn:natural-condition-for-rounding} is not sufficient to round multiple-choice algorithms losslessly.
	\end{cor}

	\section{A $2$-level Fractional Vertex-Weighted Algorithm}\label{sec:weightedfractional}

In this section we design a two-choice $2$-level fractional algorithm for the more general \emph{vertex-weighted} problem, where offline nodes have a weight associated with them, and we wish to output a matching of maximum weight.
We prove the following:
\begin{restatable}{thm}{vwalgo}\label{thm:vwalgo}
	There exists a fractional $2$-level $\nicefrac{11}{21}\approx 0.524$-competitive vertex-weighted online matching algorithm.
\end{restatable}
\begin{proof}
	The algorithm draws ideas from the 2-level algorithm for the unweighted case in Section \ref{sec:algtwo}, but is more involved, due to the offline weights adding another dimension of asymmetry.
	As in the unweighted case, the algorithm has two possible levels for the offline nodes: $z_1=\frac{1}{2}, z_2=\frac{7}{8}$. Let $w_i>0$ be the weight of node $i$.  At time $t$, let $w_1(1-y(x^{(t)}_{1}))\geq w_2(1-y(x^{(t)}_{2})) \geq \ldots, \geq w_k(1-y(x^{(t)}_{k}))$ be the neighbors of $t$ sorted by their dual slack. We again assume wlog that there are at least two neighbors. Otherwise, we add dummy neighbors with $x^{(t)}_{i}=1$ that will not change the behavior of the algorithm or the analysis. Let $\{1,2\}$ be the two offline nodes with the maximal slack (not necessarily sorted). By normalizing, we assume wlog of generality that $w_1=1$. We use the following (optimized) numbers: $y_1 = \frac{5}{11}, y_2=\frac{79}{88}$.
	The algorithm is defined by the following cases:
	\begin{itemize}
		\item $x_1=0,x_2=0$: If $w_2\leq \frac{1}{1-y_1}$ then set  $x_1=x_2=\frac{1}{2}$, otherwise set $x_1\gets 1$.
		\item $x_1=\frac{1}{2},x_2=\frac{1}{2}$: If $w_2\leq \frac{1-y_1}{1-y_2}$ then set  $x_1=x_2=\frac{7}{8}$, otherwise set $x_1\gets 1$.
		\item $x_1=\frac{7}{8},x_2=\frac{7}{8}$: If $w_2\leq 1$ then set  $x_1\gets 1$, otherwise set $x_2\gets 1$.
		\item $x_1=0,x_2=\frac{1}{2}$: If $w_2\leq \frac{3}{2}$ then set  $x_1\gets 1$, otherwise set $x_1\gets \frac{1}{2}, x_2\gets 1$.
		\item $x_1=0,x_2=\frac{7}{8}$: If $w_2\leq 5.5$ then set  $x_1\gets 1$, otherwise set $x_1\gets \frac{7}{8}, x_2\gets 1$.
		\item $x_1=\frac{1}{2},x_2=\frac{7}{8}$: If $w_2\leq 4$ then set  $x_1\gets 1$, otherwise set $x_2\gets 1$.
		\item $x_1<1,x_2=1$: set $x_1\gets 1$.
	\end{itemize}
	
	As in Observation \ref{obs-klevel} it is not hard to verify the following.
	\begin{obs}\label{vtx-weighted-maximal}
		The above algorithm is a $2$-level maximal \algotype algorithm that is $2$-bit precise.
	\end{obs}
	
	\paragraph{Analysis.}
	The analysis proceeds via a dual fitting argument. Let $\{1,2\}$ be the two vertices with the highest $w_i(1-x^{(t)}_{i})$. We note that $\{1,2\}$ are not numbered according to their slack, but rather according to their fractional degree, as in the algorithm's description. We assume wlog that $w_1=1$.
	We use dual values $y_1 =y(\frac{1}{2})=\frac{5}{11}, y_2 =y(\frac{7}{8})=\frac{79}{88}$. We note that if we give the online node a value of $\max\{\min\{1-y(x^{t}_1), w_2(1-y(x^{t}_2)\},1-y(x^{t+1}_1), w_2(1-y(x^{t+1}_2)\}$, it will satisfy all dual constraints for edges $(i,t)$ at time $t$ (and hence at all future times).
	We therefore have that the obtained dual's cost upper bounds the optimal matching's weight.
	To that end, we prove the following claim.
	\begin{restatable}{claim}{vmclaim}\label{claim:vwalgo}
		The changes to the weighted algorithm's dual and primal values at each time $t$ satisfy
		$$(\Delta P)_t \geq \frac{11}{21}\cdot (\Delta D)_t.$$
	\end{restatable}
	
	Before providing the (rather tedious) proof of the above claim, we note that it implies our claimed competitive ratio.
	Summing up over all time steps,  we have that by \Cref{claim:vwalgo} and weak duality, the primal gain (i.e., the fractional matching's value) is at least
	\begin{align*}
	P & = \sum_t (\Delta P)_t \geq \frac{11}{21} \cdot \sum_t (\Delta D)_t = \frac{11}{21} \cdot D \geq \frac{11}{21}\cdot OPT.\qedhere
	\end{align*}
\end{proof}

\begin{proof}[Proof of \Cref{claim:vwalgo}]
	Dropping the subscript $t$, since it will be clear from context,
	what we wish to prove is that for any online time $t$ and each of the cases in the algorithm's definition, we have that $\frac{\Delta D}{\Delta P}\leq 1+\frac{10}{11}$.
	We next analyze all cases, showing that in all cases the change in the primal value divided by the dual cost is at most $1+\frac{10}{11}$.
	
	\noindent{\bf Case 1  ($x_1=0,x_2=0$):} In this case we can assume wlog that $w_1=1$ and $w_2\geq 1$.
	If $w_2\leq \frac{1}{1-y_1} = \frac{38}{19}$ then, $x_1=x_2=\frac{1}{2}$. The value of the online node can be at most $\max\{1, w_2(1-y_1)\}\leq 1$. Thus,
	\[\frac{\Delta D}{\Delta P} = \frac{1+y_1+w_2 y_1}{\frac{1}{2}(1+w_2)}\leq \frac{1+y_1+ y_1}{1} = 1+\frac{10}{11}.\]
	If $w_2\geq \frac{1}{1-y_1}$ we set $x_2=1$ and we may set the value of the online node to $w_2$. In this case,
	\[\frac{\Delta D}{\Delta P} = \frac{1+w_2}{w_2}\leq \frac{1+\frac{1}{1-y_1}}{\frac{1}{1-y_1}} = 2-y_1= 1+ \frac{6}{11}.\]
	
	\noindent{\bf Case 2  ($x_1=\frac{1}{2},x_2=\frac{1}{2}$):}
	In this case we can assume wlog that $w_1=1$ and $w_2\geq 1$.
	If $w_2\leq \frac{1-y_1}{1-y_2} = \frac{16}{3}\approx 5.3$ then, $x_1=x_2=\frac{7}{8}$. The value of the online node can be set to at most $\max\{1-y_1, w_2(1-y_1)\}\leq 1-y_1$. Thus,
	\[\frac{\Delta D}{\Delta P} = \frac{1-y_1 +y_2-y_1+w_2(y_2-y_1)}{\frac{3}{8}(1+w_2)}\leq \frac{1-y_1 +y_2-y_1+1(y_2-y_1)}{\frac{3}{8}(1+1)} = 1+\frac{10}{11}.\]
	If $w_2\geq \frac{1-y_1}{1-y_2}$ we set $x_2=1$ and the value of the online node can be set to $1-y_1$. In this case,
	\[\frac{\Delta D}{\Delta P} = \frac{1-y_1+w_2(1-y_1)}{\frac{1}{2}w_2}\leq \frac{1-y_1+\frac{(1-y_1)^2}{1-y_2}}{\frac{1}{2}\frac{1-y_1}{1-y_2}} = \frac{1-y_2 + 1- y_1}{0.5}= 1+ \frac{13}{44}.\]
	
	\noindent{\bf Case 3  ($x_1=\frac{7}{8},x_2=\frac{7}{8}$):}
	In this case we can assume wlog that $w_1=1$ and $w_2\geq 1$.
	We set $x_2=1$. Thus,
	\[\frac{\Delta D}{\Delta P} = \frac{1-y_2+ w_2(1-y_2)}{\frac{1}{8}w_2}\leq  \frac{2-2y_2}{\frac{1}{8}} = 1+\frac{7}{11}.\]
	
	\noindent{\bf Case 4  ($x_1=0,x_2=\frac{1}{2}$):} Assume that $w_1=1$ If $w_2\leq \frac{3}{2}$ then we set $x_1=1$. In this case $w_2(1-y_1)\leq 1$ and so the online node can get value of $w_2(1-y_1)$. Thus,
	\[\frac{\Delta D}{\Delta P} = \frac{w_2(1-y_1) +1}{1}\leq \frac{\frac{3}{2}(1-y_1) +1}{1} = 1+\frac{9}{11} .\]
	If $w_2\geq \frac{3}{2}$ then we set $x_1=\frac{1}{2}, x_2=1$. In this case we may set the online node to 1.
	Thus,
	\[\frac{\Delta D}{\Delta P} = \frac{1+ y_1 + w_2(1-y_1)}{\frac{1}{2}(1+w_2)}\leq\frac{1+ y_1 +  \frac{3}{2}(1-y_1)}{\frac{1}{2}(1+ \frac{3}{2})} =1+ \frac{9}{11}.\]

	\noindent{\bf Case 5  ($x_1=0,x_2=\frac{7}{8}$):}
	Assume that $w_1=1$ If $w_2\leq 5.5$ then we set $x_1=1$. In this case  $w_2(1-y_2)\leq 1$ and we may set the online node to $w_2(1-y_2)$. Thus,
	\[\frac{\Delta D}{\Delta P} = \frac{w_2(1-y_2) +1}{1}\leq \frac{5.5(1-y_2) +1}{1} = 1+\frac{9}{16} .\]
	If $w_2\geq 5.5$ then we set $x_1=\frac{7}{8}, x_2=1$. In this case we may set the online node to 1.
	Thus,
	\[\frac{\Delta D}{\Delta P} = \frac{1+ y_2 + w_2(1-y_2)}{\frac{7}{8} + \frac{1}{8}w_2}\leq\frac{1+ y_2 + 5.5(1-y_2)}{\frac{7}{8} + \frac{5.5}{8}} \approx 1.57 .\]
	
	\noindent{\bf Case 6  ($x_1=\frac{1}{2},x_2=\frac{7}{8}$):}
	Assume that $w_1=1$. If $w_2\leq 4$ then we set $x_1=1$. In this case $w_2(1-y_2)\leq 1$ and we may set the online node to $w_2(1-y_2)$. Thus,
	\[\frac{\Delta D}{\Delta P} = \frac{w_2(1-y_2) +1-y_1}{\frac{1}{2}}\leq \frac{4(1-y_2) +1-y_1}{\frac{1}{2}} = 1+\frac{10}{11}.\]
	If $w_2\geq 4$ then we set $x_2=1$. In this case we may set the online node to $1-y_1$.
	Thus,
	\[\frac{\Delta D}{\Delta P} = \frac{1-y_1 + w_2(1-y_2)}{\frac{1}{8}w_2}\leq  \frac{1-y_1 + 4(1-y_2)}{\frac{4}{8}} =1+ \frac{10}{11} .\]
	
	\noindent{\bf Case 7  ($x_1<1,x_2=1$):} In this case we may set the online node to $0$. Then,
	\begin{align*}\frac{\Delta D}{\Delta P} & = \frac{1-y(x_1)}{1-x_1}\leq  \frac{12}{11}.\qedhere \end{align*}
\end{proof} 

	\section{Construction of Small-Bias Probability Spaces}\label{sec:small-bias}
\newcommand{\eqdef}{\stackrel{\Delta}{=}}
In this section we prove \Cref{eps-k-constructions} for the sake of completeness, as it pertains to the construction of $(\delta,k)$-dependent distributions with $\delta>0$ that we use.

\smallbiaslem*

We describe a construction suggested by \citet{naor1993small}. Consider a distribution $\mathcal{\mathcal{D}}$ over $\{0,1\}^n$ random variables.
\begin{Def}
  The \emph{bias} of a subset $S \subseteq \{ 1 , \ldots , n \}$ for
  a distribution $\mathcal{D}$ is
  $${\rm bias}_\mathcal{D}(S):=\left|\Pr_\mathcal{D}\left[\sum_{i \in S} x_i \equiv 0 \pmod 2
  \right]-
  \Pr_\mathcal{D}\left[\sum _{i \in S} x_i \equiv 1 \pmod 2 \right]\right|.$$
\end{Def}
We say that $\mathcal{D}$ is \emph{$k$-wise $\eps$-biased} if for every $S \subseteq \{ 1 , \ldots , n \}$ of size at most $|S|\leq k$, we have ${\rm bias}_\mathcal{D}(S) \leq \eps$.
It is shown by \citet{naor1993small} that random variables that are $k$-wise $\eps$-biased are also $(\delta,k)$-dependent  for
$\delta = 2^{{k \over 2}}\cdot \eps$.

A construction by \cite{ABI86} for uniform $k$-wise random variables goes as follows.
Let $v_1,\ldots,v_n \in \{0,1\}^h$ be vectors that are
linearly $k$-wise independent over
GF[2], with $h = {k \over 2} \log n$. Such vectors are known to exist, and moreover can be constructed in polynomial time (e.g., rows of a parity check matrix of a BCH code).
Choose $r \in \{0,1\}^h$ uniformly at random  and
define $x_i := \langle v_i,r\rangle $, for $i=1,\ldots,n$.  It is easy to see that the resulting random
variables $x_1,x_2,\dots ,x_n$ over $\{0,1\}^n$ are $k$-wise independent.

To construct $k$-wise $\eps$-biased random variables, \citet{naor1993small} use the same construction as above,
except that $r$ is now sampled from an $\eps$-biased distribution over $h$ random variables (instead of a uniformly i.i.d. source).

\begin{claim}
  The $n$ random variables constructed by sampling $r$ from an
  $\eps$-biased source are $k$-wise $\eps$-biased random variables.
\end{claim}
\begin{proof} For every subset $S$ of cardinality at most $|S| \leq k$, we have
\[\sum_{i \in S}x_i = \sum_{i \in S} \langle v_i,r \rangle  =
r \cdot \sum_{i \in S} v_i = r \cdot M_S,\]
where $M_S \eqdef \sum_{i \in S} v_i$.  Since the vectors are $k$-wise
independent, $M_S \neq \vec{0}$. Denoting by $I\subseteq [h]$ the set of indices $i$ for which $(M_S)_i = 1$, we have

\begin{align*}
{\rm bias}(S) & = \left|\Pr\left[\sum_{i \in S}x_i\equiv 0 \pmod 2\right] - \Pr\left[\sum_{i \in S}x_i\equiv 1 \pmod 2\right]\right| \\
          & = \left|\Pr\left[r \cdot M_S \equiv 0 \pmod 2\right] - \Pr\left[r \cdot M_S \equiv 1 \pmod 2\right]\right| \\
          & = \left|\Pr\left[\sum_{i\in I} r_i \equiv 0 \pmod 2\right] - \Pr\left[\sum_{i\in I} r_i \equiv 1 \pmod 2\right]\right| \\
          & \leq \eps.\qedhere
\end{align*}
\end{proof}
Thus, the number of $\eps$-biased random variables required for the construction is $h=k \log n$. The cardinality of a sample space of an $\eps$-biased distribution over $h$ variables constructed by \cite{naor1993small} is linear in the number of random variables. Specifically, it is of size $O(h/\eps^3)$.
Therefore, we need $\log (O(h/\eps^3))$ random bits to sample uniformly at random from this distribution, from which we obtain $r$.
Finally, in order to compute any $x_i = \langle v_i, r\rangle$, we only spend $O(h) = O(k\cdot \log n)$ time. As mentioned before, the vectors $v_i$ can be computed in polynomial time \cite{naor1993small}.
This concludes the proof of
Lemma \ref{eps-k-constructions}, since to obtain $(\delta,k)$-dependence, we need $\delta\leq 2^{{k \over 2}}\cdot \eps$, and so picking $\eps = \delta\cdot 2^{-{k \over 2}}$ yields a $(\delta,k)$-dependent distribution, using $\log (O(h/\eps^3)) = \log\log n + O(k+\log(1/\delta)) + O(1)$ random bits.

	\bibliographystyle{acmsmall}
	\bibliography{abb,ultimate}

\begin{thebibliography}{}

\bibitem[\protect\citeauthoryear{Adamaszek, Czumaj, Englert, and
  R{\"a}cke}{Adamaszek et~al\mbox{.}}{2012}]{adamaszek2012log}
{\sc Adamaszek, A.}, {\sc Czumaj, A.}, {\sc Englert, M.}, {\sc and} {\sc
  R{\"a}cke, H.} 2012.
\newblock An $o (\log k)$-competitive algorithm for generalized caching.
\newblock In {\em Proceedings of the 23rd Annual ACM-SIAM Symposium on Discrete
  Algorithms (SODA)}. 1681--1689.

\bibitem[\protect\citeauthoryear{Ageev and Sviridenko}{Ageev and
  Sviridenko}{2004}]{ageev2004pipage}
{\sc Ageev, A.~A.} {\sc and} {\sc Sviridenko, M.~I.} 2004.
\newblock Pipage rounding: A new method of constructing algorithms with proven
  performance guarantee.
\newblock {\em Journal of Combinatorial Optimization\/}~{\em 8,\/}~3, 307--328.

\bibitem[\protect\citeauthoryear{Aggarwal, Goel, Karande, and Mehta}{Aggarwal
  et~al\mbox{.}}{2011}]{aggarwal2011online}
{\sc Aggarwal, G.}, {\sc Goel, G.}, {\sc Karande, C.}, {\sc and} {\sc Mehta,
  A.} 2011.
\newblock Online vertex-weighted bipartite matching and single-bid budgeted
  allocations.
\newblock In {\em Proceedings of the 22nd Annual ACM-SIAM Symposium on Discrete
  Algorithms (SODA)}. 1253--1264.

\bibitem[\protect\citeauthoryear{Alon, Awerbuch, Azar, Buchbinder, and
  Naor}{Alon et~al\mbox{.}}{2006}]{alon2006general}
{\sc Alon, N.}, {\sc Awerbuch, B.}, {\sc Azar, Y.}, {\sc Buchbinder, N.}, {\sc
  and} {\sc Naor, J.} 2006.
\newblock A general approach to online network optimization problems.
\newblock {\em ACM Transactions on Algorithms (TALG)\/}~{\em 2,\/}~4, 640--660.

\bibitem[\protect\citeauthoryear{Alon, Awerbuch, Azar, Buchbinder, and
  Naor}{Alon et~al\mbox{.}}{2009}]{alon2009online}
{\sc Alon, N.}, {\sc Awerbuch, B.}, {\sc Azar, Y.}, {\sc Buchbinder, N.}, {\sc
  and} {\sc Naor, J.} 2009.
\newblock The online set cover problem.
\newblock {\em SIAM Journal on Computing (SICOMP)\/}~{\em 39,\/}~2, 361--370.

\bibitem[\protect\citeauthoryear{Alon, Babai, and Itai}{Alon
  et~al\mbox{.}}{1986}]{ABI86}
{\sc Alon, N.}, {\sc Babai, L.}, {\sc and} {\sc Itai, A.} 1986.
\newblock A fast and simple randomized parallel algorithm for the maximal
  independent set problem.
\newblock {\em J. Algorithms\/}~{\em 7,\/}~4, 567--583.

\bibitem[\protect\citeauthoryear{Bansal, Buchbinder, Madry, and Naor}{Bansal
  et~al\mbox{.}}{2011}]{bansal2011polylogarithmic}
{\sc Bansal, N.}, {\sc Buchbinder, N.}, {\sc Madry, A.}, {\sc and} {\sc Naor,
  J.} 2011.
\newblock A polylogarithmic-competitive algorithm for the k-server problem.
\newblock In {\em Proceedings of the 52nd Symposium on Foundations of Computer
  Science (FOCS)}. 267--276.

\bibitem[\protect\citeauthoryear{Bansal, Buchbinder, and Naor}{Bansal
  et~al\mbox{.}}{2012a}]{bansal2012primal}
{\sc Bansal, N.}, {\sc Buchbinder, N.}, {\sc and} {\sc Naor, J.} 2012a.
\newblock A primal-dual randomized algorithm for weighted paging.
\newblock {\em Journal of the ACM (JACM)\/}~{\em 59,\/}~4, 1--24.

\bibitem[\protect\citeauthoryear{Bansal, Buchbinder, and Naor}{Bansal
  et~al\mbox{.}}{2012b}]{bansal2012randomized}
{\sc Bansal, N.}, {\sc Buchbinder, N.}, {\sc and} {\sc Naor, J.} 2012b.
\newblock Randomized competitive algorithms for generalized caching.
\newblock {\em SIAM Journal on Computing (SICOMP)\/}~{\em 41,\/}~2, 391--414.

\bibitem[\protect\citeauthoryear{Bansal, Buchbinder, and Naor}{Bansal
  et~al\mbox{.}}{2010}]{bansal2010metrical}
{\sc Bansal, N.}, {\sc Buchbinder, N.}, {\sc and} {\sc Naor, J.~S.} 2010.
\newblock Metrical task systems and the k-server problem on hsts.
\newblock In {\em Proceedings of the 37th International Colloquium on Automata,
  Languages and Programming (ICALP)}. 287--298.

\bibitem[\protect\citeauthoryear{Blanc and Charikar}{Blanc and
  Charikar}{2021}]{blanc2021multiway}
{\sc Blanc, G.} {\sc and} {\sc Charikar, M.} 2021.
\newblock Multiway online correlated selection.
\newblock In {\em Proceedings of the 62nd Symposium on Foundations of Computer
  Science (FOCS)}. 1277--1284.

\bibitem[\protect\citeauthoryear{B{\"o}ckenhauer, Komm, Kr{\'a}lovi{\v{c}}, and
  Rossmanith}{B{\"o}ckenhauer et~al\mbox{.}}{2014}]{bockenhauer2014online}
{\sc B{\"o}ckenhauer, H.-J.}, {\sc Komm, D.}, {\sc Kr{\'a}lovi{\v{c}}, R.},
  {\sc and} {\sc Rossmanith, P.} 2014.
\newblock The online knapsack problem: Advice and randomization.
\newblock {\em Theoretical Computer Science (TCS)\/}~{\em 527}, 61--72.

\bibitem[\protect\citeauthoryear{Border}{Border}{1991}]{border91}
{\sc Border, K.~C.} 1991.
\newblock Implementation of reduced form auctions: A geometric approach.
\newblock {\em Econometrica\/}~{\em 59,\/}~4, 1175–--1187.

\bibitem[\protect\citeauthoryear{Boyar, Favrholdt, Kudahl, Larsen, and
  Mikkelsen}{Boyar et~al\mbox{.}}{2017}]{boyar2017online}
{\sc Boyar, J.}, {\sc Favrholdt, L.~M.}, {\sc Kudahl, C.}, {\sc Larsen, K.~S.},
  {\sc and} {\sc Mikkelsen, J.~W.} 2017.
\newblock Online algorithms with advice: a survey.
\newblock {\em ACM Computing Surveys (CSUR)\/}~{\em 50,\/}~2, 1--34.

\bibitem[\protect\citeauthoryear{Bubeck, Cohen, Lee, and Lee}{Bubeck
  et~al\mbox{.}}{2019}]{BCLLM19}
{\sc Bubeck, S.}, {\sc Cohen, M.~B.}, {\sc Lee, J.~R.}, {\sc and} {\sc Lee,
  Y.~T.} 2019.
\newblock Metrical task systems on trees via mirror descent and unfair gluing.
\newblock In {\em Proceedings of the 30th Annual ACM-SIAM Symposium on Discrete
  Algorithms (SODA)}. 89--97.

\bibitem[\protect\citeauthoryear{Bubeck, Cohen, Lee, Lee, and
  M{\k{a}}dry}{Bubeck et~al\mbox{.}}{2018}]{bubeck2018k}
{\sc Bubeck, S.}, {\sc Cohen, M.~B.}, {\sc Lee, Y.~T.}, {\sc Lee, J.~R.}, {\sc
  and} {\sc M{\k{a}}dry, A.} 2018.
\newblock K-server via multiscale entropic regularization.
\newblock In {\em Proceedings of the 50th Annual ACM Symposium on Theory of
  Computing (STOC)}. 3--16.

\bibitem[\protect\citeauthoryear{Buchbinder, Gupta, Molinaro, and
  Naor}{Buchbinder et~al\mbox{.}}{2019}]{buchbinder2019k}
{\sc Buchbinder, N.}, {\sc Gupta, A.}, {\sc Molinaro, M.}, {\sc and} {\sc Naor,
  J.} 2019.
\newblock k-servers with a smile: online algorithms via projections.
\newblock In {\em Proceedings of the 30th Annual ACM-SIAM Symposium on Discrete
  Algorithms (SODA)}. 98--116.

\bibitem[\protect\citeauthoryear{Buchbinder, Jain, and Naor}{Buchbinder
  et~al\mbox{.}}{2007}]{buchbinder2007online}
{\sc Buchbinder, N.}, {\sc Jain, K.}, {\sc and} {\sc Naor, J.~S.} 2007.
\newblock Online primal-dual algorithms for maximizing ad-auctions revenue.
\newblock In {\em Proceedings of the 15th Annual European Symposium on
  Algorithms (ESA)}. 253--264.

\bibitem[\protect\citeauthoryear{Buchbinder and Naor}{Buchbinder and
  Naor}{2009}]{buchbinder2009design}
{\sc Buchbinder, N.} {\sc and} {\sc Naor, J.~S.} 2009.
\newblock The design of competitive online algorithms via a primal-dual
  approach.
\newblock {\em Foundations and Trends\textsuperscript{\textregistered} in
  Theoretical Computer Science\/}~{\em 3,\/}~2--3, 93--263.

\bibitem[\protect\citeauthoryear{Coester and Lee}{Coester and Lee}{2019}]{CL19}
{\sc Coester, C.} {\sc and} {\sc Lee, J.~R.} 2019.
\newblock Pure entropic regularization for metrical task systems.
\newblock In {\em Proceedings of the 32nd Conference on Computational Learning
  Theory (COLT)}. 835--848.

\bibitem[\protect\citeauthoryear{Cohen, Peng, and Wajc}{Cohen
  et~al\mbox{.}}{2019}]{cohen2019tight}
{\sc Cohen, I.~R.}, {\sc Peng, B.}, {\sc and} {\sc Wajc, D.} 2019.
\newblock Tight bounds for online edge coloring.
\newblock In {\em Proceedings of the 60th Symposium on Foundations of Computer
  Science (FOCS)}. 1--25.

\bibitem[\protect\citeauthoryear{Cohen and Wajc}{Cohen and
  Wajc}{2018}]{cohen2018randomized}
{\sc Cohen, I.~R.} {\sc and} {\sc Wajc, D.} 2018.
\newblock Randomized online matching in regular graphs.
\newblock In {\em Proceedings of the 29th Annual ACM-SIAM Symposium on Discrete
  Algorithms (SODA)}. 960--979.

\bibitem[\protect\citeauthoryear{Delong, Farhadi, Niazadeh, and Sivan}{Delong
  et~al\mbox{.}}{2022}]{delong2022online}
{\sc Delong, S.}, {\sc Farhadi, A.}, {\sc Niazadeh, R.}, {\sc and} {\sc Sivan,
  B.} 2022.
\newblock Online bipartite matching with reusable resources.
\newblock In {\em Proceedings of the 23rd ACM Conference on Economics and
  Computation}. 962--963.

\bibitem[\protect\citeauthoryear{Devanur, Jain, and Kleinberg}{Devanur
  et~al\mbox{.}}{2013}]{devanur2013randomized}
{\sc Devanur, N.~R.}, {\sc Jain, K.}, {\sc and} {\sc Kleinberg, R.~D.} 2013.
\newblock Randomized primal-dual analysis of ranking for online bipartite
  matching.
\newblock In {\em Proceedings of the 24th Annual ACM-SIAM Symposium on Discrete
  Algorithms (SODA)}. 101--107.

\bibitem[\protect\citeauthoryear{D{\"u}rr, Konrad, and Renault}{D{\"u}rr
  et~al\mbox{.}}{2016}]{durr2016power}
{\sc D{\"u}rr, C.}, {\sc Konrad, C.}, {\sc and} {\sc Renault, M.} 2016.
\newblock On the power of advice and randomization for online bipartite
  matching.
\newblock In {\em Proceedings of the 24th Annual European Symposium on
  Algorithms (ESA)}. 37:1--37:16.

\bibitem[\protect\citeauthoryear{Edmonds}{Edmonds}{1965}]{edmonds1965maximum}
{\sc Edmonds, J.} 1965.
\newblock Maximum matching and a polyhedron with 0, 1-vertices.
\newblock {\em Journal of research of the National Bureau of Standards
  B\/}~{\em 69,\/}~125-130, 55--56.

\bibitem[\protect\citeauthoryear{Emek, Fraigniaud, Korman, and Ros{\'e}n}{Emek
  et~al\mbox{.}}{2011}]{emek2011online}
{\sc Emek, Y.}, {\sc Fraigniaud, P.}, {\sc Korman, A.}, {\sc and} {\sc
  Ros{\'e}n, A.} 2011.
\newblock Online computation with advice.
\newblock {\em Theoretical Computer Science (TCS)\/}~{\em 412,\/}~24,
  2642--2656.

\bibitem[\protect\citeauthoryear{Fahrbach, Huang, Tao, and
  Zadimoghaddam}{Fahrbach et~al\mbox{.}}{2020}]{fahrbach2020edge}
{\sc Fahrbach, M.}, {\sc Huang, Z.}, {\sc Tao, R.}, {\sc and} {\sc
  Zadimoghaddam, M.} 2020.
\newblock Edge-weighted online bipartite matching.
\newblock In {\em Proceedings of the 61st Symposium on Foundations of Computer
  Science (FOCS)}. 412--423.

\bibitem[\protect\citeauthoryear{Feldman, Korula, Mirrokni, Muthukrishnan, and
  P{\'a}l}{Feldman et~al\mbox{.}}{2009}]{feldman2009online2}
{\sc Feldman, J.}, {\sc Korula, N.}, {\sc Mirrokni, V.}, {\sc Muthukrishnan,
  S.}, {\sc and} {\sc P{\'a}l, M.} 2009.
\newblock Online ad assignment with free disposal.
\newblock In {\em Proceedings of the 5th Conference on Web and Internet
  Economics (WINE)}. 374--385.

\bibitem[\protect\citeauthoryear{Fortuin, Kasteleyn, and Ginibre}{Fortuin
  et~al\mbox{.}}{1971}]{fortuin1971correlation}
{\sc Fortuin, C.~M.}, {\sc Kasteleyn, P.~W.}, {\sc and} {\sc Ginibre, J.} 1971.
\newblock Correlation inequalities on some partially ordered sets.
\newblock {\em Communications in Mathematical Physics\/}~{\em 22,\/}~2,
  89--103.

\bibitem[\protect\citeauthoryear{Gamlath, Kapralov, Maggiori, Svensson, and
  Wajc}{Gamlath et~al\mbox{.}}{2019}]{gamlath2019online}
{\sc Gamlath, B.}, {\sc Kapralov, M.}, {\sc Maggiori, A.}, {\sc Svensson, O.},
  {\sc and} {\sc Wajc, D.} 2019.
\newblock Online matching with general arrivals.
\newblock In {\em Proceedings of the 60th Symposium on Foundations of Computer
  Science (FOCS)}. 26--38.

\bibitem[\protect\citeauthoryear{Gandhi, Khuller, Parthasarathy, and
  Srinivasan}{Gandhi et~al\mbox{.}}{2006}]{gandhi2006dependent}
{\sc Gandhi, R.}, {\sc Khuller, S.}, {\sc Parthasarathy, S.}, {\sc and} {\sc
  Srinivasan, A.} 2006.
\newblock Dependent rounding and its applications to approximation algorithms.
\newblock {\em Journal of the ACM (JACM)\/}~{\em 53,\/}~3, 324--360.

\bibitem[\protect\citeauthoryear{Gao, He, Huang, Nie, Yuan, and Zhong}{Gao
  et~al\mbox{.}}{2021}]{gao2021improved}
{\sc Gao, R.}, {\sc He, Z.}, {\sc Huang, Z.}, {\sc Nie, Z.}, {\sc Yuan, B.},
  {\sc and} {\sc Zhong, Y.} 2021.
\newblock Improved online correlated selection.
\newblock In {\em Proceedings of the 62nd Symposium on Foundations of Computer
  Science (FOCS)}. 1265--1276.

\bibitem[\protect\citeauthoryear{Goel, Kapralov, and Khanna}{Goel
  et~al\mbox{.}}{2013}]{goel2013perfect}
{\sc Goel, A.}, {\sc Kapralov, M.}, {\sc and} {\sc Khanna, S.} 2013.
\newblock Perfect matchings in ${O}(n \log n)$ time in regular bipartite
  graphs.
\newblock {\em SIAM Journal on Computing (SICOMP)\/}~{\em 42,\/}~3, 1392--1404.

\bibitem[\protect\citeauthoryear{Gr{\"o}tschel, Lov{\'a}sz, and
  Schrijver}{Gr{\"o}tschel et~al\mbox{.}}{2012}]{grotschel2012geometric}
{\sc Gr{\"o}tschel, M.}, {\sc Lov{\'a}sz, L.}, {\sc and} {\sc Schrijver, A.}
  2012.
\newblock {\em Geometric algorithms and combinatorial optimization}. Vol.~2.
\newblock Springer Science \& Business Media.

\bibitem[\protect\citeauthoryear{Huang, Kang, Tang, Wu, Zhang, and Zhu}{Huang
  et~al\mbox{.}}{2020a}]{huang2020fully}
{\sc Huang, Z.}, {\sc Kang, N.}, {\sc Tang, Z.~G.}, {\sc Wu, X.}, {\sc Zhang,
  Y.}, {\sc and} {\sc Zhu, X.} 2020a.
\newblock Fully online matching.
\newblock {\em Journal of the ACM (JACM)\/}~{\em 67,\/}~3, 1--25.

\bibitem[\protect\citeauthoryear{Huang, Tang, Wu, and Zhang}{Huang
  et~al\mbox{.}}{2020b}]{huang2020fully2}
{\sc Huang, Z.}, {\sc Tang, Z.~G.}, {\sc Wu, X.}, {\sc and} {\sc Zhang, Y.}
  2020b.
\newblock Fully online matching ii: Beating ranking and water-filling.
\newblock In {\em Proceedings of the 61st Symposium on Foundations of Computer
  Science (FOCS)}. 1380--1391.

\bibitem[\protect\citeauthoryear{Huang and Tao}{Huang and
  Tao}{2019}]{huang2019understandingb}
{\sc Huang, Z.} {\sc and} {\sc Tao, R.} 2019.
\newblock Understanding zadimoghaddam's edge-weighted online matching
  algorithm: Unweighted case.
\newblock {\em arXiv preprint arXiv:1910.02569\/}.

\bibitem[\protect\citeauthoryear{Huang, Zhang, and Zhang}{Huang
  et~al\mbox{.}}{2020c}]{huang2020adwords}
{\sc Huang, Z.}, {\sc Zhang, Q.}, {\sc and} {\sc Zhang, Y.} 2020c.
\newblock Adwords in a panorama.
\newblock In {\em Proceedings of the 61st Symposium on Foundations of Computer
  Science (FOCS)}. 1416--1426.

\bibitem[\protect\citeauthoryear{Kalyanasundaram and Pruhs}{Kalyanasundaram and
  Pruhs}{2000}]{kalyanasundaram2000optimal}
{\sc Kalyanasundaram, B.} {\sc and} {\sc Pruhs, K.~R.} 2000.
\newblock An optimal deterministic algorithm for online $b$-matching.
\newblock {\em Theoretical Computer Science (TCS)\/}~{\em 233,\/}~1, 319--325.

\bibitem[\protect\citeauthoryear{Karlin, Kenyon, and Randall}{Karlin
  et~al\mbox{.}}{2001}]{KKR01}
{\sc Karlin, A.~R.}, {\sc Kenyon, C.}, {\sc and} {\sc Randall, D.} 2001.
\newblock Dynamic {TCP} acknowledgement and other stories about e/(e-1).
\newblock In {\em Proceedings of the 33rd Annual ACM Symposium on Theory of
  Computing (STOC)}. 502--509.

\bibitem[\protect\citeauthoryear{Karp, Vazirani, and Vazirani}{Karp
  et~al\mbox{.}}{1990}]{karp1990optimal}
{\sc Karp, R.~M.}, {\sc Vazirani, U.~V.}, {\sc and} {\sc Vazirani, V.~V.} 1990.
\newblock An optimal algorithm for on-line bipartite matching.
\newblock In {\em Proceedings of the 22nd Annual ACM Symposium on Theory of
  Computing (STOC)}. 352--358.

\bibitem[\protect\citeauthoryear{Lee}{Lee}{2018}]{lee2018fusible}
{\sc Lee, J.~R.} 2018.
\newblock Fusible hsts and the randomized k-server conjecture.
\newblock In {\em Proceedings of the 59th Symposium on Foundations of Computer
  Science (FOCS)}. 438--449.

\bibitem[\protect\citeauthoryear{Lov{\'a}sz and Plummer}{Lov{\'a}sz and
  Plummer}{2009}]{lovasz2009matching}
{\sc Lov{\'a}sz, L.} {\sc and} {\sc Plummer, M.~D.} 2009.
\newblock {\em Matching theory}. Vol. 367.
\newblock American Mathematical Society.

\bibitem[\protect\citeauthoryear{Mehta}{Mehta}{2013}]{mehta2013online}
{\sc Mehta, A.} 2013.
\newblock Online matching and ad allocation.
\newblock {\em Foundations and Trends\textsuperscript{\textregistered} in
  Theoretical Computer Science\/}~{\em 8,\/}~4, 265--368.

\bibitem[\protect\citeauthoryear{Mehta, Saberi, Vazirani, and Vazirani}{Mehta
  et~al\mbox{.}}{2007}]{mehta2007adwords}
{\sc Mehta, A.}, {\sc Saberi, A.}, {\sc Vazirani, U.}, {\sc and} {\sc Vazirani,
  V.} 2007.
\newblock Adwords and generalized online matching.
\newblock {\em Journal of the ACM (JACM)\/}~{\em 54,\/}~5, 22.

\bibitem[\protect\citeauthoryear{Mikkelsen}{Mikkelsen}{2016}]{mikkelsen2016randomization}
{\sc Mikkelsen, J.~W.} 2016.
\newblock Randomization can be as helpful as a glimpse of the future in online
  computation.
\newblock In {\em Proceedings of the 43rd International Colloquium on Automata,
  Languages and Programming (ICALP)}. 39:1--39:14.

\bibitem[\protect\citeauthoryear{Naor, Panigrahi, and Singh}{Naor
  et~al\mbox{.}}{2011}]{naor2011online}
{\sc Naor, J.}, {\sc Panigrahi, D.}, {\sc and} {\sc Singh, M.} 2011.
\newblock Online node-weighted steiner tree and related problems.
\newblock In {\em Proceedings of the 52nd Symposium on Foundations of Computer
  Science (FOCS)}. 210--219.

\bibitem[\protect\citeauthoryear{Naor and Naor}{Naor and
  Naor}{1993}]{naor1993small}
{\sc Naor, J.~S.} {\sc and} {\sc Naor, M.} 1993.
\newblock Small-bias probability spaces: Efficient constructions and
  applications.
\newblock {\em SIAM Journal on Computing (SICOMP)\/}~{\em 22,\/}~4, 838--856.

\bibitem[\protect\citeauthoryear{Papadimitriou, Pollner, Saberi, and
  Wajc}{Papadimitriou et~al\mbox{.}}{2021}]{papadimitriou2021online}
{\sc Papadimitriou, C.}, {\sc Pollner, T.}, {\sc Saberi, A.}, {\sc and} {\sc
  Wajc, D.} 2021.
\newblock Online stochastic max-weight bipartite matching: Beyond prophet
  inequalities.
\newblock In {\em Proceedings of the 22nd ACM Conference on Economics and
  Computation (EC)}. 763--764.

\bibitem[\protect\citeauthoryear{Pena and Borodin}{Pena and
  Borodin}{2019}]{pena2019extensions}
{\sc Pena, N.} {\sc and} {\sc Borodin, A.} 2019.
\newblock On extensions of the deterministic online model for bipartite
  matching and max-sat.
\newblock {\em Theoretical Computer Science (TCS)\/}~{\em 770}, 1--24.

\bibitem[\protect\citeauthoryear{Renault and Ros{\'e}n}{Renault and
  Ros{\'e}n}{2015}]{renault2015online}
{\sc Renault, M.~P.} {\sc and} {\sc Ros{\'e}n, A.} 2015.
\newblock On online algorithms with advice for the k-server problem.
\newblock {\em Theoretical Computer Science (TCS)\/}~{\em 56,\/}~1, 3--21.

\bibitem[\protect\citeauthoryear{Renault, Ros{\'e}n, and van Stee}{Renault
  et~al\mbox{.}}{2015}]{renault2015online2}
{\sc Renault, M.~P.}, {\sc Ros{\'e}n, A.}, {\sc and} {\sc van Stee, R.} 2015.
\newblock Online algorithms with advice for bin packing and scheduling
  problems.
\newblock {\em Theoretical Computer Science (TCS)\/}~{\em 600}, 155--170.

\bibitem[\protect\citeauthoryear{Saberi and Wajc}{Saberi and
  Wajc}{2021}]{saberi2021greedy}
{\sc Saberi, A.} {\sc and} {\sc Wajc, D.} 2021.
\newblock The greedy algorithm is \emph{not} optimal for on-line edge coloring.
\newblock In {\em Proceedings of the 48th International Colloquium on Automata,
  Languages and Programming (ICALP)}. 109:1--109:18.

\bibitem[\protect\citeauthoryear{Shin and An}{Shin and
  An}{2021}]{shin2021making}
{\sc Shin, Y.} {\sc and} {\sc An, H.-C.} 2021.
\newblock Making three out of two: Three-way online correlated selection.
\newblock In {\em Proceedings of the 32nd Annual International Symposium on
  Algorithms and Computation (ISAAC)}. 49:1--49:17.

\bibitem[\protect\citeauthoryear{Ta{-}Shma}{Ta{-}Shma}{2017}]{Ta-Shma17}
{\sc Ta{-}Shma, A.} 2017.
\newblock Explicit, almost optimal, epsilon-balanced codes.
\newblock In {\em Proceedings of the 49th Annual ACM Symposium on Theory of
  Computing (STOC)}. 238--251.

\bibitem[\protect\citeauthoryear{Wang and Wong}{Wang and
  Wong}{2015}]{wang2015two}
{\sc Wang, Y.} {\sc and} {\sc Wong, S. C.-w.} 2015.
\newblock Two-sided online bipartite matching and vertex cover: Beating the
  greedy algorithm.
\newblock In {\em Proceedings of the 42nd International Colloquium on Automata,
  Languages and Programming (ICALP)}. 1070--1081.

\bibitem[\protect\citeauthoryear{Yao}{Yao}{1977}]{yao1977lemma}
{\sc Yao, A. C.-C.} 1977.
\newblock Probabilistic computations: Toward a unified measure of complexity.
\newblock In {\em Proceedings of the 18th Symposium on Foundations of Computer
  Science (FOCS)}. 222--227.

\end{thebibliography}
	\end{document}